\documentclass[letterpaper,11pt]{article}
\usepackage[utf8]{inputenc}
\usepackage[margin=1in]{geometry}
\usepackage{authblk}

\usepackage{amsmath}
\usepackage{amsthm}
\usepackage{amssymb}
\usepackage{thmtools}
\usepackage{thm-restate}
\usepackage{enumitem}
\usepackage{dsfont}
\usepackage{xspace}
\usepackage{comment}
\usepackage{xparse}
\usepackage[ruled,noend,resetcount,linesnumbered]{algorithm2e}
\usepackage{tikz}
\usepackage{fixme}
\usepackage[colorlinks]{hyperref}
\usepackage[capitalize]{cleveref}
\usepackage{subcaption}

\fxsetup{mode=multiuser,theme=color, layout=inline}
\FXRegisterAuthor{tk}{atk}{\color{blue} {\bf Tomasz}}
\FXRegisterAuthor{hs}{ahs}{\color{orange} {\bf Hamed}}
\FXRegisterAuthor{jg}{ajg}{\color{purple} {\bf Jacob}}
\FXRegisterAuthor{mh}{amh}{\color{green} {\bf Mohammad}}

\newcommand{\op}{\textrm{\texttt{(}}}
\newcommand{\cl}{\textrm{\texttt{)}}}

\newcommand{\B}{\mathcal{B}}

\newcommand{\sub}{\subseteq}
\newcommand{\sm}{\setminus}
\newcommand{\Oh}{\mathcal{O}}
\newcommand{\tOh}{\tilde{\Oh}}
\newcommand{\ed}{\mathsf{ed}}
\newcommand{\ted}{\mathsf{ted}}
\newcommand{\ded}{\mathsf{ded}}
\newcommand{\per}{\mathsf{per}}
\newcommand{\eps}{\varepsilon}

\newcommand{\poly}{\mathrm{poly}}
\newcommand{\rot}{\mathsf{rot}}
\renewcommand{\aa}{\mathsf{A}}
\newcommand{\mm}{\mathsf{M}}
\newcommand{\ta}{\mathsf{TA}}

\newcommand{\A}{\mathcal{A}}
\newcommand{\M}{\mathcal{M}}

\newcommand{\dd}{\mathinner{.\,.\allowbreak}}

\newcommand{\rev}[1]{\overline{#1}}
\newcommand{\Zz}{\mathbb{Z}_{\ge 0}}
\newcommand{\Zp}{\mathbb{Z}_{+}}

\newcommand{\DYCK}{\mathsf{Dyck}}
\newcommand{\dyck}{\mathsf{dyck}}
\newcommand{\Qf}{\mathcal{Q}}

\setlist[enumerate]{nosep, topsep=1ex}
\setlist[itemize]{nosep, topsep=1ex}
\setlist[description]{nosep,topsep=1ex}

\newtheorem{theorem}{Theorem}[section]
\newtheorem{corollary}[theorem]{Corollary}
\newtheorem{proposition}[theorem]{Proposition}
\newtheorem{lemma}[theorem]{Lemma}
\newtheorem{fact}[theorem]{Fact}
\newtheorem{claim}[theorem]{Claim}

\theoremstyle{definition}
\newtheorem{definition}[theorem]{Definition}

\theoremstyle{remark}

\newtheorem{example}[theorem]{Example}

\SetKwProg{Fn}{}{:}{}

\title{Weighted Edit Distance Computation: Strings, Trees and Dyck}
\author[1]{Debarati Das}
\author[2]{Jacob Gilbert}
\author[2]{MohammadTaghi Hajiaghayi}
\author[3]{Tomasz Kociumaka}
\author[4]{Barna Saha}

\affil[1]{Pennsylvania State University, United States}
\affil[ ]{\texttt{debaratix710@gmail.com}}
\affil[2]{University of Maryland, United States}
\affil[ ]{\texttt{jgilber8@umd.edu}\; \texttt{hajiaghayi@gmail.com}}
\affil[3]{Max Planck Institute for Informatics, Germany}
\affil[ ]{\texttt{tomasz.kociumaka@mpi-inf.mpg.de}}
\affil[4]{University of California, San Diego, United States}
\affil[ ]{\texttt{barnas@ucsd.edu}}

\date{}

\begin{document}

\maketitle

\begin{abstract}

Given two strings of length $n$ over alphabet $\Sigma$, and an upper bound $k$ on their edit distance, the algorithm of Myers (Algorithmica’86) and Landau and Vishkin (JCSS’88) from almost forty years back computes the unweighted string edit distance in $\Oh(n+k^2)$ time.
Till date, it remains the fastest algorithm for exact edit distance computation, and it is optimal under the Strong Exponential Hypothesis (STOC'15).
Over the years, this result has inspired many developments, including fast approximation algorithms for string edit distance as well as similar $\tOh(n+\poly(k))$-time algorithms for generalizations to tree and Dyck edit distances. 
Surprisingly, all these results hold only for unweighted instances. 

While unweighted edit distance is theoretically fundamental, almost all real-world applications require weighted edit distance, where different weights are assigned to different edit operations (insertions, deletions, and substitutions), and the weights may vary with the characters being edited.
Given a weight function $w: \Sigma \cup \{\eps \}\times \Sigma \cup \{\eps\} \rightarrow \mathbb{R}_{\ge 0}$ (such that $w(a,a)=0$ and $w(a,b)\ge 1$ for all $a,b\in \Sigma \cup \{\eps\}$ with $a\ne b$), the goal is to find an alignment that minimizes the total weight of edits.
Except for the vanilla $\Oh(n^2)$-time dynamic-programming algorithm and its almost trivial $\Oh(nk)$-time implementation ($k$ being an upper bound on the sought total weight), none of the aforementioned developments on the unweighted edit distance applies to the weighted variant. 
In this paper, we propose the first $\Oh(n+\poly(k))$-time algorithm that computes weighted string edit distance exactly, thus bridging a fundamental decades-old gap between our understanding of unweighted and weighted edit distance.
We then generalize this result to weighted tree and Dyck edit distances, bringing in several new techniques, which lead to a deterministic algorithm that improves upon the previous work even for unweighted tree edit distance.
Given how fundamental weighted edit distance is, we believe our $\Oh(n+\poly(k))$ algorithm for weighted edit distance will be instrumental for further significant developments in the area.

\end{abstract}

% \setcounter{page}{0}
% \thispagestyle{empty}
% \newpage
% \sloppy
\section{Introduction}
String edit distance and its several variants have been studied for decades since the 1960s~\cite{L65,NW70,WF74}. Historically, most work on these problems assumed that the edit operations have unit weights in order to simplify the problem and streamline theoretical results.
Till date,  the fastest exact algorithm for unweighted edit distance is due to Myers~\cite{DBLP:journals/algorithmica/Meyers86} and Landau and Vishkin~\cite{LV88}, who obtained an $\Oh(n+k^2)$-time solution for two strings of length $n$ with an upper bound $k$ on their edit distance. 
This bound is now known to be optimal (up to subpolynomial factors) under the Strong Exponential Hypothesis~\cite{backurs2015edit}.
Over the years, the Holy-Grail result of~\cite{DBLP:journals/algorithmica/Meyers86,LV88} %whose running time is linear in $n$ as long as $k =\Oh(\sqrt{n})$ %is polynomially small (sub-quadratic) in $n$ 
has inspired many developments on fast approximation algorithms for (unweighted) string edit distance~\cite{10.1145/3422823,DBLP:conf/focs/GoldenbergKS19,GRS:20,KS20a} and similar $\tOh(n+\poly(k))$-time\footnote{The $\tOh(\cdot)$ notation suppresses factors polylogarithmic in the input size $n$.} algorithms for generalizations such as the (unweighted) Dyck and tree edit distances~\cite{BO16,OtherSubmission,DGHKSS'22}. 
However, almost all real-world applications require weighted edit distance, where different weights are assigned to different edit operations (insertions, deletions, and substitutions), and the weights may vary with the characters being edited~\cite{WF74, zhang1989simple, kurtz1996approximate, 10.5555/262228, peris2002fast, jurafsky2008speech, fontan2016using, skiena1998algorithm, koide2020fast, gerlach2021paradigm,Chumachenko2022WeightedED}. %\mhnote{Give several references for this?
%in particular Weighted references: 
%Algorithms on Strings, Trees, and Sequences (Section 11.5.2)
%Daniel Jurafsky; James H. Martin. Speech and Language Processing
%https://dl.acm.org/doi/pdf/10.1145/321796.321811 (J. ACM)
%The Algorithm Design Manual (Skiena)
%}
As a result, there is a major gap between the theoretical results of prior research and real-world utility of these results. 
In this paper, we bridge this fundamental gap between the understanding of unweighted and weighted edit distance:
We provide the first non-trivial algorithm computing the weighted edit distance and its generalizations to weighted tree and Dyck tree edit distance.

%Weighted XML, json, weighted bio informatics

More specifically, in this paper we propose the first $\Oh(n+\poly(k))$-time algorithm for exact weighted edit distance computation in which, given a weight function $w: \Sigma \cup \{\eps \}\times \Sigma \cup \{\eps\} \rightarrow \mathbb{R}_{\ge 0}$ (normalized so that $w(a,b)\ge 1$ for $a\ne b$),
the goal is to find an alignment that minimizes the total weight of edit operations (insertions, deletions, and substitutions) assuming that it does not exceed a provided threshold $k$.
Strikingly, except for the vanilla $\Oh(n^2)$-time dynamic-programming algorithm and its almost trivial $\Oh(nk)$-time implementation, none of the aforementioned developments on unweighted edit distance apply to this weighted variant.
We then generalize our result to weighted tree and Dyck edit distances, % for which even no  simple  $\Oh(nk)$ algorithms exist where again $k$ is an upper bound on the total weight. \tknote{O(nk) would violate the lower bound for the unweighted case, so of course no such algorithms exist}
bringing in several new techniques that lead to improvements even for the unweighted tree edit distance problem:
As a byproduct of our results, we present a deterministic $\Oh(n+k^{7}\log k)$-time solution, which is much faster than the randomized $\Oh(n\log{n}+k^{15}\log^2k)$-time algorithm of Das, Gilbert, Hajiaghayi, Kociumaka, Saha, and Saleh~\cite{DGHKSS'22}.

Can this apparent lack of progress in weighted edit distance computation be explained? 
As we observe later, even basic properties like monotonicity, which was fundamental for efficient  computation of unweighted edit distance~\cite{DBLP:journals/algorithmica/Meyers86, LV88}, break down when considering weighted operations. 
This precludes any local matching approach, which seemed necessary for a linear-time algorithm for bounded (unweighted) edit distance~\cite{DBLP:journals/algorithmica/Meyers86, LV88,BO16,OtherSubmission,DGHKSS'22}; instead, a global view of the sequences is needed to find matching substrings and yet maintain the linear runtime.  Faced with such barriers, our biggest contribution is 
%Interestingly, even some basic and fundamental properties of unweighted edit distance break down in presence of weighted operations. For example, given two strings $X$ and $Y$, it is not necessary that the edit distance ($\ed$) between $X[1..i], Y[1...j]$, $\ed(X[1..i], Y[1..j]) \geq \ed(X[1..i-1], Y[1..j-1])$, that is weighted edit distance can be non-monotone. Moreover, even if $X[i]=Y[j]$, it is possible $\ed(X[1..i], Y[1..j]) \neq \ed(X[1..i-1], Y[1..j-1])$, that is an optimum solution can change even if a matching suffix is added to both the strings. Monotonicity and greedy matching properties are instrumental for , and for all subsequent developments on unweighted edit distance. We overcome this substantial barrier in this paper.  We develop a 
a kernelization method for weighted edit distance, not just for strings, but also for tree and Dyck edit distance instances. 
Interestingly, our kernels are weight-agnostic, that is, the kernelization algorithms do not need to know the weight function~$w$.
Given how fundamental weighted edit distance is, we believe our $\Oh(n+\poly(k))$ algorithm for weighted edit distance will be instrumental for further significant developments in the area.

%\mhnote{can we talk  about the practicality of the tools in our paper?}

%New things: $ \Oh(n+ \poly(k))$ as opposed to $\Oh(n \log n + \poly(k))$; deterministic; no heavy black box (greedy);
%(to be more precise: we do use greedy alignments; previously, we relied on a black-box for computing the structure of such alignments)

%\url{https://epubs.siam.org/doi/pdf/10.1137/S0097539702419650?casa_token=XPe9TaQt7WwAAAAA:rdycx-DTaHul0C9Mxfx8u6H9LwGP2YpljtRd0hC_OLvuo-DigLXNyJIt6Ak-0Eccdnl0f46KL3tm3w}

%(Section 2.1 in the paper above)

\subsection{Related Work}
\paragraph{String Edit Distance:}
\emph{Edit distance} is one of the most fundamental problems in computer science studied since the 1960s~\cite{L65,NW70,WF74}. In the unweighted edit distance problem, given two strings of length at most $n$, the goal is to find the minimum number of edit operations (insertions, deletions, and substitutions) required to transform one string into the other.
Given a parameter $k$ as an upper bound on the  edit distance,
an algorithm proposed in the 1980s by Myers~\cite{DBLP:journals/algorithmica/Meyers86} and Landau and Vishkin~\cite{LV88} achieves this task in $\Oh(n+k^2)$ time by combining suffix trees with an elegant greedy approach.
As long as $k = \Oh(\sqrt{n})$, the running time of the above algorithm in linear in $n$. 
For larger values of $k$, approximation algorithms for edit distance have been studied extensively~\cite{10.5555/279082.279125,10.5555/874063.875596,10.1109/FOCS.2004.14,10.5555/1109557.1109644,10.1145/1536414.1536444,10.1109/FOCS.2010.43}, especially recently~\cite{10.1145/3313276.3316371,10.1145/3422823,GRS:20,KS20,BR19,DBLP:journals/jacm/BoroujeniEGHS21}.
This culminated with the currently best bound by Andoni and Nosatzki~\cite{DBLP:conf/focs/AndoniN20}, who obtained a constant-factor approximation algorithm  with running time $\Oh(n^{1+\epsilon})$ time for any constant $\epsilon > 0$. All of these works require monotonicity and assume that an optimal solution can be extended easily if matching suffixes are added to both strings, none of which may hold in weighted edit distance instances.
As a result, the state-of-the-art approximation algorithm for weighted edit distance, by Kuszmaul~\cite{K19}, offers much worse trade-off, with an $\Oh(n^{\tau})$-factor approximation in $\tOh(n^{2-\tau})$ time for any $0\le\tau\le 1$.

\paragraph{Tree Edit Distance:}
The \emph{tree edit distance} problem, first introduced by Selkow~\cite{SELKOW1977184}, is a generalization of edit distance in which the task is to compute a measure of dissimilarity
between two rooted ordered trees with node labels. In the unweighted version of tree edit distance, every node insertion, deletion, or
relabeling operation has unit cost. 
The problem has numerous applications in compiler optimization~\cite{10.1145/1644015.1644017}, structured data analysis~\cite{DBLP:conf/vldb/Chawathe99,10.5555/1315451.1315465,10.1145/1613676.1613680}, image analysis~\cite{10.1016/S0167-8655(97)00179-7}, and computational biology~\cite{10.1137/0213024,DBLP:journals/bioinformatics/ShapiroZ90,10.5555/262228,10.5555/279082.279125,10.1016/j.tcs.2004.12.030}. The current best bound on running time of an algorithm for finding exact tree edit distance is due to  D\"{u}rr~\cite{Duerr2022}  who obtained an $\Oh(n^{2.9149})$-time algorithm for the problem, after a long series of improvements from  $\Oh(n^6)$~\cite{10.1145/322139.322143} to  $\Oh(n^4)$~\cite{zhang1989simple},
to $\Oh(n^3\log n)$~\cite{10.5555/647908.740125}, to  $\Oh(n^3)$~\cite{10.1145/1644015.1644017}, and to
 $\Oh(n^{2.9546})$~\cite{Xiao21}. 
Moreover, there is a $(1+\epsilon)$-approximation algorithm for tree
edit distance with running time $\tOh(n^2)$ time due to Boroujeni, Ghodsi, Hajiaghayi, and Seddighin~\cite{DBLP:conf/stoc/BoroujeniGHS19}. 
Recently, Seddighin and Seddighin~\cite{DBLP:conf/innovations/SeddighinS22} gave an $\Oh(n^{1.99})$-time $(3+\epsilon)$-approximation algorithm
for tree~edit~distance (building on a previous $\tOh(n)$-time $\Oh(\sqrt{n})$-factor approximation algorithm of~\cite{DBLP:conf/stoc/BoroujeniGHS19}).
Furthermore, Das, Gilbert, Hajiaghayi, Kociumaka, Saha, and Saleh~\cite{DGHKSS'22} obtained an $\tOh(n+k^{15})$-time algorithm for exact tree edit distance with an upper bound $k$ on the distance (see also an $\tOh(nk^2)$-time algorithm of Akmal and Jin~\cite{DBLP:conf/icalp/AkmalJ21}, which improves upon a previous algorithm with running time $\Oh(nk^3)$ for the bounded tree edit distance problem~\cite{10.1007/11496656_29}). 

As far as the weighted tree edit distance is concerned, the fastest algorithm, by Demaine, Mozes, Rossman, and Weimann~\cite{10.1145/1644015.1644017},
takes $\Oh(n^3)$ time, which matches the conditional lower-bound of Bringmann, Gawrychowski, Mozes, and Weimann~\cite{10.1145/3381878} (earlier conjectured by Abboud~\cite{Amir14}). Specifically, there is no truly subcubic-time algorithm for weighted tree edit distance unless APSP has a truly subcubic-time solution. The lower bound still holds for trees over a constant-size alphabet unless the weighted $k$-clique problem admits an $\Oh(n^{k-\epsilon})$-time~algorithm.

\paragraph{Dyck Edit Distance:}
The \emph{Dyck edit distance} problem is another variation of edit distance which falls under the umbrella of general language edit distance~\cite{aho1972minimum, M95,saha2017fast,BGSW19} and has numerous practical applications, e.g., for fixing hierarchical data files, in particular XML and JSON files~\cite{h:78,k:12}.
In the unweighted version of this problem, given a string of $n$ parentheses, the goal is to find the minimum number of edits (character insertions, deletions, and substitutions)  to make the string well-balanced.
Several algorithms for both  exact~\cite{BGSW19,CDX22,Duerr2022} and approximation~\cite{s:14,DKS21} versions of the problem have been obtained.
Finding exact Dyck edit distance is at least as hard as Boolean matrix multiplication~\cite{abw:15}.
The bounded Dyck edit problem was subject to several recent studies as well:  
Backurs and Onak~\cite{BO16} obtained the first algorithm with running time $\Oh(n+k^{16})$, which was further improved to $\Oh(n+k^5)$ \cite{OtherSubmission}, and finally to $\Oh(n+k^{4.5442})$ using fast matrix multiplication~\cite{F22a,Duerr2022}.
Except for the $\Oh(n^3)$-time exact algorithm for language edit distance~\cite{M95}, these results are not applicable to the weighted setting.

\subsection{Our Contribution}

The main contributions of our paper are new algorithms for weighted string, tree, and Dyck edit distance. 
We define a \emph{weight function} as a function  $w: \Sigma \cup \{\eps \}\times \Sigma \cup \{\eps\} \rightarrow \mathbb{R}_{\ge 0}$
such that $w(a,a)=0$ and $w(a,b)\ge 1$ for $a\ne b$. If $a,b\in \Sigma$, then $w(a,\eps)$ is the cost of deleting $a$,
$w(\eps,b)$ is the cost of inserting $b$, whereas $w(a,b)$ is the cost of substituting $a$ for $b$. 
The assumption $w(a,a)=0$ indicates that matching symbols can be aligned at no cost, whereas the assumption $w(a,b)\ge 1$ for $a\ne b$ indicates that the weights are normalized so that every edit costs at least one.
A weight function is a \emph{quasimetric} if it also satisfies the triangle inequality (which we assume for tree and Dyck edit distance).
When it comes to computations on weights, we consider any uniform model in which real numbers are subject to only comparison and addition~\cite{pettie2005shortest}, e.g., the RAM model.

We define $\ed^w(X, Y)$ to be the minimum cost of an alignment of strings $X$ and $Y$ for weight function $w$. Furthermore, we define $\ed^w_{\le k}(X, Y)$ as $\ed^w(X, Y)$ (if it is at most $k$) or $\infty$ (otherwise). We give the first weighted bounded edit distance algorithm with runtime $\Oh(n + \poly(k))$.

\begin{restatable}{theorem}{weighted}\label{thm:weighted_ed}
   Given strings $X,Y$ of length at most $n$, an integer $k\in \Zp$, and a weight function~$w$,
the value $\ed_{\le k}^w(X,Y)$ can be computed in $\Oh(n+k^5)$ time.
\end{restatable}

Similarly to string edit distance, we define $\ted^w(F, G)$ as the minimum cost of a tree alignment of forests $F$ and $G$ for weight function $w$. We define $\ted^w_{\le k}(F, G)$ analogously and give the first weighted tree edit distance algorithm with runtime $\Oh(n + \poly(k))$. In the unweighted case, 
our deterministic algorithm is significantly faster than the state-of-the-art randomized algorithm from~\cite{DGHKSS'22}.

\begin{restatable}{theorem}{weightedtree}\label{thm:weighted_tree_ed}
Given forests $F,G$ of length at most $n$, an integer $k\in \Zp$, and a quasimetric $w$,
the value $\ted_{\le k}^w(F,G)$ can be computed in $\Oh(n+k^{15})$ time.
Moreover,  $\ted_{\le k}(F,G)$ can be computed in $\Oh(n+k^{7}\log k)$ time.
\end{restatable}

Finally, we define $\dyck_{\le k}^w(X)$ to be the minimum distance $\ed_{\le k}^w(X,Y)$ between $X$ and a string $Y$ in the Dyck language.
We give the first algorithm for weighted Dyck edit distance with runtime $\Oh(n + \poly(k))$. In this setting, the alphabet consists of opening and closing parentheses, and we need to assume that the weight function, apart from satisfying the triangle inequality, treats opening and closing parentheses of the same type similarly. This is captured in the notion of a \emph{skewmetric} formally defined in \cref{sec:dyckweight}.

\begin{restatable}{theorem}{weightedDyck}\label{thm:weighted_dyck}
  Given a string $X$ of length $n$, an integer $k\in \Zp$, and a skewmetric $w$, the value $\dyck_{\le k}^w(X)$ can be computed in $\Oh(n+k^{12})$ time.
\end{restatable}

We note that, although our algorithms assume $k$ is given, one can also obtain running times analogous to those of \cref{thm:weighted_ed,thm:weighted_tree_ed,thm:weighted_dyck} but with the sought distance instead of the threshold~$k$. 
For this, it suffices to start from the largest value $k$ that results in the running time of $\Oh(n)$, e.g., $k=\Theta(n^{1/5})$ for strings, and keep doubling the threshold $k$ as long as the algorithm outputs $\infty$.
The first finite outcome is guaranteed to be the sought distance and, since the running times of the subsequent iterations form a geometric progression,
the overall runtime is dominated by the last iteration, where $k$ is at most twice the sought distance.

\subsection{Overview}\label{subsec:overview}
The folklore algorithms to compute edit distance for unweighted and weighted instances use dynamic programming and runs in $\Oh(n^2)$ time. Given two strings $X$ and $Y$, the entry $D[i,j]$ of the dynamic programming table $D$ holds the weighted (unweighted) edit distance of prefixes of $X$ and $Y$ up to indices $i$ and $j$ respectively. That is $D[i, j] := \ed(X[0 \dd i), Y[0 \dd j))$. Then
\[
    D[i+1, j+1] = \min\{D[i, j+1] + 1, D[i+1, j] + 1, D[i, j] + \delta(X[i], Y[j])\} \, \textit{ :unweighted edit distance}
\]
\begin{multline*}
    D[i+1, j+1] = \min\{D[i, j+1] + w(X[i], \eps), D[i+1, j] + w(\eps, Y[j]), D[i, j] + w(X[i], Y[j])\}.\\ \textit{ :weighted edit distance}
\end{multline*}
The first entry in the recursive definition corresponds to deleting $X[i]$, the second entry corresponds to inserting $Y[j]$, and the third entry corresponds to either matching or substitution ($\delta(X[i], Y[j]) = 0$ if $X[i] = Y[j]$, otherwise $\delta(X[i], Y[j]) = 1$).
%Traditional solutions for edit distance and its variants rely on recursive formulations of the problem that can be solved using dynamic programming. The folklore $\Oh(n^2)$-time dynamic programming algorithm extends to weighted edit distance as well. In the unweighted algorithm, a table $D$ is constructed dynamically where $D_{i, j} := \ed(X[0 \dd i), Y[0 \dd j))$, or in other words, the $(i, j)$-entry of $D$ is equal to the edit distance between the prefix of $X$ up to $i$ and the prefix of $Y$ up to $j$. 
Clearly, $D[|X|, |Y|]$ equals the total weighted (unweighted) edit distance between $X$ and $Y$, and can be computed in $\Oh(n^2)$ time. 

It is possible to improve the running time to $\Oh(nk)$ if the weighted (unweighted) edit distance is bounded by $k < n$. In this case the entries corresponding to only $2k+1$ diagonals surrounding the main diagonal of $D$ need to be computed. However, the similarities between the developments on unweighted and weighted edit distance computations end here.

The first major breakthrough in the unweighted edit distance computation came in the late eighties~\cite{DBLP:journals/algorithmica/Meyers86,LV88}.
An $\Oh(n+k^2)$-time algorithm for unweighted edit distance was developed whenever edit distance is bounded by $k$, thereby giving a linear time algorithm for $k \leq \sqrt{n}$. The algorithm utilizes two simple but powerful properties of unweighted edit distance, namely (i) \emph{monotonicity}: $D[i+1,j+1] \geq D[i,j]$, and (ii) \emph{greedy extension}: if $X[i]=Y[j]$ then $D[i+1,j+1]=D[i,j]$.
These two properties together imply that if we can find maximal equal substrings in $X$ and $Y$ through a preprocessing step, only $\Oh(k^2)$ entries of $D$ need to be computed. More precisely, for each of the $2k+1$ diagonals, these are the at most $k+1$ entries with $k+1\ge D[i+1,j+1]>D[i,j]$. 
The preprocessing step utilizes a linear-time construction of a suffix tree to answer any maximal equal substring queries in constant time, leading to an overall running time of $\Oh(n + k^2)$. All subsequent developments on fast approximation algorithms for unweighted string edit distance rely on the above two properties without exception.

Unfortunately, none of the above two properties hold for weighted edit distance computation. The following simple examples will make this observation clear.
\renewcommand{\a}{\mathtt{a}}
\renewcommand{\b}{\mathtt{b}}
\renewcommand{\c}{\mathtt{c}}
\renewcommand{\d}{\mathtt{d}}

\begin{enumerate}
\item \emph{No monotonicity:}  Let $X = \a\b$, $Y = \c$, and $w(\a,\c) + w(\b,\eps) < w(\a, \eps)$. Then $D[1,0] = w(\a,\eps)$ and $D[2,1] = w(\a,\c) + w(\b,\eps) < w(\a, \eps)$.
\item  \emph{No greedy extension:} Let $X = \a\b$, $Y = \b$, and $w(\a, \b) + w(\b, \eps) < w(\a, \eps)$. Then substituting $\a$ to $\b$ and deleting $\b$ from $X$ is cheaper than deleting $\a$ and matching the subsequent $\b$.
\end{enumerate}
In some sense, this explains the lack of progress on weighted instances in this field. We need a very different approach and new ideas.

When $k$, the minimum weighted edit distance, is small for two input strings, clearly most characters of the input strings are perfectly matched and contribute no cost to the edit distance computation. The main idea of our algorithm is to find small representative instances for the input strings and then run the $\Oh(nk)$-time weighted edit distance solution on these representatives to find the original weighted edit distance. In fact, we prove that any instance of the bounded-weighted edit distance can be solved using strings of size $\Oh(k^4)$. Our algorithm constructs such an $\Oh(k^4)$-size kernel from strings of size $\Oh(n)$ in time $\Oh(n)$, and then the resulting small instances can be solved in time $\Oh(k^5)$ using the $\Oh(nk)$-time weighted extension of the dynamic programming.

We are also able to extend the idea of kernelization to weighted instances of tree and Dyck edit distances by giving the first $\Oh(n + \poly(k))$ algorithms for them. Notably, our algorithms are deterministic and give significant improvements over the recent randomized algorithms on unweighted tree edit distance~\cite{DGHKSS'22}.  We show it is possible to compute small $\Oh(\poly(k))$-size kernels from the original instances of each problem in linear time, and then run dynamic programming based algorithms to compute the final edit distance values. 

To find such kernels, we utilize substrings that have \emph{synchronized occurrences} in both input strings $X$ and $Y$,
that is, they occur in $X$ and $Y$ at positions $x$ and $y$, respectively, satisfying $|x-y|=\Oh(k)$.
Our kernelization algorithm first tries to cover the input strings (almost entirely) with $\Oh(k)$ pairs of synchronized occurrences.
If this is impossible, then we conclude that the edit distance must be large, that is, $\ed^w_{\le k}(X,Y)=\infty$.
Otherwise, we apply a novel notion of \emph{edit-distance equivalence} so that synchronized occurrences of a substring $P$
can be substituted with synchronized occurrences of an equivalent substring $P'$ without affecting the edit distance $\ed^w_{\le k}(X,Y)$. 
To this end, we provide a linear-time algorithm that, given any string $P$, computes an edit-distance equivalent string $P'$ of size $\poly(k)$.

A similar notion of equivalent pieces is also central to our algorithms for weighted tree and Dyck edit distance.
Our three algorithms all utilize the following high-level steps:
\begin{enumerate}
    \item Partition the input objects into $\Oh(k)$ pieces most of which can be paired up to form synchronized occurrences.
    \item If the algorithm failed to find sufficiently long synchronized occurrences, report that the edit distance exceeds $k$.
    \item Otherwise, for every pair of synchronized occurrences, substitute the original piece with a small equivalent replacement.
    \item Solve the resulting small instance with a known dynamic-programming algorithm.
\end{enumerate}

\subsubsection{Weighted String Edit Distance}
We now describe how to obtain~\cref{thm:weighted_ed} by implementing the aforementioned high-level scheme.

\paragraph{Edit-Distance Equivalent Strings}
The biggest technical contribution behind our weighted edit distance algorithm is a linear-time procedure (of~\cref{cor:string_reduction}) that, given a string $P$, computes an equivalent string of length $\Oh(k^3)$.
In the first phase, it eliminates \emph{$k$-periodicity}: as long as the processed string contains a fragment of the form $Q^{4k+1}$ with $|Q|\in [1\dd 2k]$, this fragment is replaced by $Q^{4k}$.
As shown in \cref{lem:periodic_reduction}, the strings $Q^{4k+1}$ and $Q^{4k}$ are equivalent,
so this step preserves equivalence with the input string~$P$. 
Eventually, the first phase results in a string that \emph{avoids $k$-periodicity} and is equivalent with $P$ (see \cref{fig:per_reduc} for example).
It is implemented in~\cref{lem:perred}, where the underlying algorithm processes the input string $P$ from left to right and removes the first copy of $Q$
for every encountered fragment of the form $Q^{4k+1}$ with $|Q|\in [1\dd 2k]$.

In \cref{lem:aperiodic_reduction}, we prove that if $P$ avoids $k$-periodicity and satisfies $|P|\ge 42k^3$, then it is equivalent to $P[0\dd  21k^3) \cdot P[|P| - 21k^3\dd |P|)$, that is, the concatenation of its prefix of length $21k^3$ and its suffix of length $21k^3$ (with the characters in the middle removed).
For this, we consider synchronized occurrences of $P$ in strings $X$ and $Y$ and an optimal alignment $\A$ of cost $\ed^w(X,Y)\le k$ that maps $X$ onto~$Y$.
We observe that $\A$ must perfectly match a length-$10k^2$ fragment within the length-$21k^3$ prefix of the occurrence of $P$ in $X$.
Moreover, since $P$ avoids $k$-periodicity, $\A$ can only match this fragment to the corresponding fragment of the occurrence of $P$ in $Y$.
Symmetrically, $\A$ must match the two copies of a length-$10k^2$ fragment within the length-$21k^3$ suffix of $P$.
We conclude that $\A$ aligns the two copies of $P[d\dd |P|-e)$ for some $d,e\in [0\dd 21k^3]$.
Since $\A$ is optimal, it must perfectly match the two copies of $P[d\dd |P|-e)$.
Thus, $P[21k^3\dd |P|-21k^3)$ can be removed from the synchronized occurrences of $P$ in $X$ and $Y$ without affecting the cost $\ed^w_{\le k}(X,Y)$.
Consequently, if the first phase returns a string of length at least $42k^3$,
then the algorithm of \cref{cor:string_reduction} removes all but the leading $21k^3$ and the trailing $21k^3$ characters of that string (see \cref{fig:ed_equiv} for example).

\paragraph{Linear-Time Kernel}
In order to apply the notion of edit-distance equivalence, we need to identify synchronized occurrences within $X$ and $Y$.
To this end, we check whether $\ed(X,Y)\le k$. 
If this is not the case, then $\ed^w(X,Y)\ge \ed(X,Y)>k$ holds for every normalized weight function~$w$, and thus we already know that $\ed_{\le k}^w(X,Y)=\infty$. If $\ed(X,Y)\le k$, on the other hand, then we construct an optimal unweighted $\A$ alignment mapping $X$ onto $Y$. 
As formally proved in \cref{fct:str_decomp}, the unedited characters of $X$ form at most $k+1$ fragments that $\A$ matches perfectly. Each of these fragments of $X$ forms a synchronized occurrence together with its image under $\A$ in $Y$.
Thus, we can replace the synchronized occurrences with occurrences of an equivalent string of length $\Oh(k^3)$.
Since we have partitioned $X$ and $Y$ into $\Oh(k)$ edited characters plus $\Oh(k)$ synchronized occurrences, 
this yields strings  $X'$ and $Y'$ of length $\Oh(k^4)$ satisfying $\ed^w_{\le k}(X', Y') = \ed^w_{\le k}(X, Y)$. In order to construct $\A$ efficiently, we use the $\Oh(n + k^2)$ unweighted edit distance algorithm of~\cite{DBLP:journals/algorithmica/Meyers86,LV88}.
However, if $n \leq k^4$, then we do not need to worry about reducing the size of $X$ and $Y$ in the first place and therefore do not construct an optimal unweighted alignment; otherwise, $\Oh(n + k^2) = \Oh(n)$ and constructing the $\Oh(k^4)$-size kernel takes linear time; see~\cref{thm:stringKernel} for details
on our kernel for weighted string edit distance.

As mentioned earlier, once we have a kernel $(X', Y')$ of size $\Oh(k^4)$, we can run the $\Oh(nk)$-time weighted edit-distance algorithm to compute 
$\ed^w_{\le k}(X', Y')=\ed^w_{\le k}(X,Y)$ in $\Oh(k^5)$ time.

\subsubsection{Weighted Tree Edit Distance}
Our algorithm for weighted tree edit distance follows the same high-level approach. However, compared to the string edit distance, two major challenges arise. First, the structure of periodicity is much richer and requires two notions: \emph{horizontal periodicity} of \emph{forests} and \emph{vertical periodicity} of \emph{contexts}. 
As a result, we need separate definitions of tree-edit-distance equivalence for forests and contexts.
Nevertheless, assuming that the weight function $w$ satisfies the triangle inequality, we can still construct equivalent
forests and contexts of size $\Oh(k^3)$ and $\Oh(k^4)$, respectively.
The second challenge is that the state-of-the-art algorithm for computing the unweighted tree edit distance is randomized and takes $\Oh(n\log n + \poly(k))$
time rather than $\Oh(n + \poly(k))$ time.
Thus, in order to achieve a deterministic linear-time kernel, we need another method for identifying large synchronizing occurrences.
Our workaround is to shrink the input in multiple iterations (essentially halving the size each time) rather than in a single shot.
This way, we can still obtain a kernel of size $\Oh(k^5)$, which is asymptotically as small as we would get from an optimum unweighted alignment.

\paragraph{Periodicity in Trees}
Intuitively, the two types of periodicity in trees correspond to the two ways to interpret strings as trees.
For a string $X$, the \emph{horizontal embedding} constructs a tree with $|X|$ leafs attached to the root and labeled by subsequent characters of $X$, whereas the \emph{vertical embedding} constructs a path with $|X|$ nodes labeled by subsequent characters of $X$.
Similarly, forest algebras (see~\cite{BW08} for a survey) in formal language theory involve two natural monoids: a \emph{horizontal monoid} of forests
(with concatenation, denoted $\cdot$) and a \emph{vertical monoid} of contexts (with composition, denoted $\star$).
A context can be defined as a tree with a single \emph{hole} in some leaf, and contexts can be composed by placing one of them
in the hole of the other. Moreover, placing a forest in the hole of a context yields a forest.
In order to formalize these notions and easily port combinatorial and algorithmic tools designed for strings,
we interpret forests as balanced strings of parentheses; see \cref{subsec:treeprelim}.

Following~\cite{DGHKSS'22}, a horizontal power is the concatenation of multiple copies of the same forest,
whereas a vertical power is the composition of multiple copies of the same context; see \cref{fig:tree_per} for an example.
More specifically, we say that a forest contains horizontal $k$-periodicity if it has a subforest of the form $Q^{4k+1}$ for some forest $Q$ of size $|Q|\le 4k$,
whereas a context contains vertical $k$-periodicity if it can be expressed as a composition of several contexts, including $Q^{6k+1}$
for some context $Q$ of size $|Q|\le 8k$.

\paragraph{Tree-Edit-Distance Equivalent Forests}
The first ingredient of our algorithm for weighted tree edit distance is a linear-time procedure that, given a forest $P$,
constructs an equivalent forest of size $\Oh(k^3)$. 
The first phase of this subroutine eliminates horizontal $k$-periodicity: as long as the processed forest contains a subforest of the form $Q^{4k+1}$ with $|Q|\in [1\dd 4k]$, this subforest is replaced by $Q^{4k}$.
As shown in \cref{lem:horizontal_periodic_reduction}, the forests $Q^{4k+1}$ and $Q^{4k}$ are equivalent,
so this step preserves equivalence with the input forest~$P$.  
An efficient implementation of this phase relies on the fact that, if $P$ is interpreted as a string,
then horizontal $k$-periodicity can be interpreted as a substring of the form $Q^{4k+1}$ for a sufficiently short \emph{balanced} string $Q$.
Thus, we can reuse \cref{lem:perred} to obtain a forest equivalent with $P$ that avoids horizontal $k$-periodicity.

In \cref{lem:horizontal_aperiodic_reduction}, we show the equivalence of any two forests of size at least $74k^3$ that avoid horizontal $k$-periodicity.%
\footnote{This statement is stronger that its counterpart for strings, \cref{lem:aperiodic_reduction}, because we now assume that the weight function $w$ satisfies the triangle inequality.}
Based on this result, if horizontal periodicity reduction yields a forest of size at least $74k^3$, we return a canonical forest of size exactly $74k^3$; see~\cref{lem:horizontal_reduction} for details.

\paragraph{Tree-Edit-Distance Equivalent Contexts}
Our next ingredient is a linear-time algorithm that, given a context $P$, constructs an equivalent context of size $\Oh(k^4)$.
First, we use the previous procedure for every maximal forest in $P$ (that does not contain the hole).
Then, we eliminate vertical $k$-periodicity: as long as $P$ contains a context of the form $Q^{6k+1}$ with $|Q|\in [1\dd 8k]$, this context is replaced by $Q^{6k}$. As shown in \cref{lem:vertical_periodic_reduction}, the contexts $Q^{6k+1}$ and $Q^{6k}$ are equivalent,
so this step preserves equivalence with the input context~$P$.  
For an efficient implementation, the \emph{spine}, i.e., the path from the root of $P$ to the hole, 
is interpreted as a string, with each character encoding the label of the underlying node and the subtrees attached there to the left and to the right of the spine.
This way, vertical-$k$ periodicity can be interpreted as periodicity in the constructed string, and hence \cref{lem:perred} can be used again.

In \cref{lem:vertical_aperiodic_reduction}, we show the equivalence of any two contexts of size at least $578k^4$ that avoid vertical $k$-periodicity
and subforests of size more than $74k^3$. Thus, if vertical periodicity reduction yields a context of size at least $578k^4$, we replace it with a canonical context of size exactly $578k^4$; see~\cref{lem:vertical_reduction} for details.

\paragraph{Linear-Time Kernel}
As for strings, in order to apply the notion of tree-edit-distance equivalence, we need to identify synchronized occurrences of forests and contexts
within the input forests $F$ and $G$. As mention above, in order to obtain a deterministic linear-time kernel, we cannot use the algorithm of~\cite{DGHKSS'22}
to obtain a tree alignment mapping $F$ to $G$ with at most $k$ edits.
Instead, we develop an iterative workaround. At each step, we decompose $F$ into $\Oh(k)$ contexts and forests (jointly called pieces)
of size at most $\frac{n}{2k}$ each; see \cref{lem:decomp} for details.
Next, we maximize the number of pieces (from the decomposition) that admit disjoint synchronized occurrences in $G$; 
\cref{lem:dp} implements this step in $\Oh(n+k^4)$ time using dynamic programming. 
If $\ted(F,G)\le k$, then no more than $k$ of the pieces are left unmatched (an optimal alignment may edit at most $k$ pieces).
We replace the matched pieces with equivalent pieces of size $\Oh(k^5)$, obtaining forests of size at most $\frac{n}{2}+\Oh(k^5)$,
where the first term corresponds to the unmatched pieces; see~\cref{thm:forestKernel}.
As long as $n = \omega(k^5)$, this procedure essentially halves the input size.
Hence, as shown in~\cref{cor:forestKernel}, this still yields a linear-time algorithm producing forests $F'$ and $G'$ of size $\Oh(k^5)$ such that $\ted_{\le k}^w(F', G') = \ted_{\le k}^w(F, G)$.

Once we have such a kernel $(F', G')$ of size $\Oh(k^5)$, we can run the cubic-time weighted edit-distance algorithm~\cite{10.1145/1644015.1644017} to compute 
$\ted^w_{\le k}(X', Y')=\ted^w_{\le k}(X,Y)$ in $\Oh(k^{15})$ time, for a total runtime of $\Oh(n+k^{15})$.
Additionally, we significantly improve the state-of-the-art of the unweighted tree edit distance problem by using the $\Oh(nk^2\log n)$-time algorithm from~\cite{DBLP:conf/icalp/AkmalJ21}, which gives us a total runtime of $\Oh(n + k^7\log k)$ for unweighted tree edit distance.

\subsubsection{Weighted Dyck Edit Distance}
In the final section of our paper, the weighted Dyck edit distance algorithm follows a similar approach to that of the string and tree edit distance algorithm.  However, many of the proofs and details are specific to Dyck edit distance problem and come with their own set of intricacies and difficulties that we outline in the following. 

Given a string $X$ over an alphabet $\Sigma=T\cup \rev{T}$ (where $T$ and $\rev{T}$ are the sets of opening and closing parentheses, respectively), an integer $k\in \Zp$, and a skewmetric weight function $w$ representing the cost of each edit operation (parenthesis insertion, deletion, and substitution), our objective is to compute the minimum weight of a sequence of edits that convert $X$ to a well-parenthesized expression over $\Sigma$ provided the total weight of all edits is bounded by $k$. In this work we design a deterministic algorithm that achieves this goal in $\Oh(n+k^{12})$ time. For the unweighted counterpart of this problem, the recent solution of~\cite{F22a,Duerr2022} computes the Dyck edit distance in time $\Oh(n+k^{4.5442})$. That algorithm,
consistently with its predecessors~\cite{BO16,OtherSubmission}, starts with a greedy preprocessing step that exhaustively
removes any two adjacent characters $X[i]X[i+1]$ such that $X[i]$ is an opening parenthesis and $X[i+1]$ is a closing parenthesis of the same type.
Following a simple argument, it can be shown that the Dyck edit distance of the preprocessed string stays exactly the same as the input string $X$. 

\paragraph{Preprocessing}
We tried to follow a similar approach for the weighted version, but it turns out that such a simple analysis is not enough to construct a reduced string.
For example, let the input string be $\texttt{(\{()}$. 
For a general weight function $w$, it is not evident that the optimal matching should always match the last two parentheses. 
In fact, if we consider a weight function where the cost of substituting $\texttt{\{}$ with $\texttt{)}$ is $10$ whereas the cost of substituting $\texttt{(}$ with $\texttt{\}}$ is 5, 
then any optimal matching should match the first and the last parentheses instead of the last two. 
Thus, in this work, we consider our weight function $w$ to be a skewmetric. 
Formally, we assume that $w$ satisfies the triangle inequality and skew-symmetry, that is, $w(p_1,p_2)=w(\rev{p_2},\rev{p_1})$ holds for all $p_1,p_2\in \Sigma\cup \{\eps\}$, where $\rev{p}$ is the parenthesis complementary to $p$ (and $\rev{\eps}=\eps$).
Following this property of $w$, we show that one can apply a similar greedy preprocessing (as described for the unweighted version) to reduce $X$ to a string $X'$ while preserving the weighted Dyck edit distance. 
Our argument is substantially more elaborate, though, and follows a case-by-case analysis depending on the structure of the other alternate alignments (\cref{claim:preprocess1}).
Nevertheless, it is trivial to observe the greedy preprocessing can be done in linear time. 

\paragraph{Dyck-Edit-Distance Equivalent Strings} 
Next, following a similar strategy as described for string edit distance, we further reduce $X'$ to generate a string $X''$ of length $\Oh(k^4)$ while preserving the weighted Dyck edit distance. For this first we introduce the concept of $k$-synchronicity. A substring $P$ containing only opening parentheses and a substring $\rev{P}$ containing only closing parentheses are $k$-synchronized if $\rev{P}$ appears after $P$, they are of same length and their height difference is at most $2k$. 
Following this and the non-crossing property of Dyck matching, first we argue that if the lengths of $P, \rev{P}$ are large and the distance is bounded by $k$, then there exist a substring $\ell\in P$ that is matched with a substring $\ell'\in \rev{P}$ in the optimal alignment (we fix one for the analysis purpose). Now if we replace $P$ with $P\setminus \ell$ and  $P'$ with $P'\setminus \ell'$ then in the resulting string the distance stays the same (Fact~\ref{fct:sync}). Following this, for any two $k$-synchronized substrings $P, \rev{P}$, we can reduce their periodicity as follows: if $P=Q^e$ and $\rev{P}=\rev{Q^e}$, (where $Q$ is a primitive string with large exponent $e$) then at least one occurrence of $Q$ is matched with its reverse complement counterpart $\rev{Q}$ in $\rev{P}$. Thus, we can remove the matched part while not changing the distance, and it reduces the exponent by one. Repeat this until $e$ become small (Lemma~\ref{lem:periodic_reduction_dyck}). 

Next assuming that $P, \rev{P}$ avoid periodicity, it can be shown that there exists a pair of indices $i,j\in [0\dd  78k^3]$ such that $P[i]$ is matched with, $\rev{P}[|P|-1-i]$ and $P[|P|-1-j]$ is matched with, $\rev{P}[j]$ in the optimal alignment. 
Thus, following the fact that $|P|=|\rev{P}|$ and the non-crossing property of the Dyck optimal alignment, all the indices between $i$ and $|P|-1-j$ are also matched, and thus removing these matched characters from both $P, \rev{P}$ does not affect the Dyck edit distance.
Consequently, we replace each $k$-synchronized pairs with substrings of length just $156k^3$ (replace $P,\rev{P}$ with their first and last $78k^3$ characters) to generate a string $X''$ such that the weighted Dyck edit distance of $X$ and $X''$ is the same (\cref{lem:aperiodic_reduction_dyck}, \cref{cor:dyck_reduction}). 

\paragraph{Linear-Time Kernel} 
Lastly, we show if the distance is bounded by $k$, then $X$ can be partitioned in time $\Oh(n+k^5)$ into $\Oh(k)$ disjoint $k$-synchronized pairs of substrings (plus $\Oh(k)$ individual characters) and thus the total length of $X''$ is bounded by $\Oh(k^4)$. Start by preprocessing input string $X$ to generate $X'$.
Next, we check if $\ded(X')\le k$ and, if so, we compute an unweighted optimal Dyck alignment $\M$ of $X'$ in time $\Oh(n+k^5)$~\cite{OtherSubmission}. Then, we argue any pair of substrings of $X'$ that are matched by $\M$ are $k$-synchronized. Thus, using $\M$, we identify the set of maximal substrings from $T^*$ and $\rev{T}^*$ that are matched by $\M$. A substring is maximal in a sense that either the substring itself or its matched counterpart can not be extended to the right or left without paying an edit. As unweighted Dyck edit distance is no more than the weighted version and hence, assuming cost of $\M$ is bounded by $k$, we can show string $X'$ can be partitioned into $\Oh(k)$ different $k$-synchronized pairs. 
Also, these maximal fragments can be found in linear time with a left-to-right scan of $X'$. Subsequently, we create a string $X''$ from $X'$ as follows: (i) for each $k$-synchronized pairs we reduce them following the algorithm as discussed above and add two corresponding strings each of length $O(k^3)$ (ii) add all the characters that are edited by $\M$ just the same to $X''$ (\cref{thm:dyckkernel}). 

Finally, we compute the weighted Dyck edit distance of $X''$ using the dynamic program algorithm of~\cite{M95} in time $\Oh(k^{12})$.

\section{String Edit Distance}
\label{sec:ed}

\subsection{Preliminaries}
A \emph{string} $Y\in \Sigma^n$ is a sequence of $|Y|:=n$ characters from an \emph{alphabet} $\Sigma$.
For $i\in [0\dd n)$, we denote the $i$th character of $Y$ with $Y[i]$. %The \emph{reverse} of a string $Y$ is $\rev{Y}:=Y[n-1] Y[n-2] \cdots Y[0]$.
We say that a string $X$ \emph{occurs} as a \emph{substring} of a string $Y$ if $X=Y[i]\cdots Y[j-1]$ holds for some integers $0\le i \le j \le |Y|$.
We denote the underlying \emph{occurrence} of $X$ as $Y[i\dd j)$.
Formally, $Y[i\dd j)$ is a \emph{fragment} of $Y$ that can be represented using a reference to $Y$ as well as its endpoints $i,j$.
The fragment $Y[i\dd j)$ can be alternatively denoted as $Y[i\dd j-1]$, $Y(i-1\dd j-1]$, or $Y(i-1\dd j)$.
A fragment of the form  $Y[0\dd j)$ is a \emph{prefix} of $Y$, whereas a fragment of the form $Y[i\dd n)$ is a \emph{suffix} of $Y$.

\begin{theorem}[LCE queries~\cite{LV88,FFM00}]\label{thm:lce}
Strings $X,Y$ can be preprocessed in linear time so that the following \emph{longest common extension} (LCE) queries can be answered 
in $\Oh(1)$ time: given positions $x\in [0\dd |X|]$ and $y\in [0\dd |Y|]$, compute the largest 
$\ell$ such that $X[x\dd x+\ell)=Y[y\dd y+\ell)$.
\end{theorem}

As mentioned in \cref{subsec:overview}, high-power periodicity plays a key role in our algorithms, which we may now formally define for strings here. An integer $p\in [1\dd n]$ is a \emph{period} of a string $Y\in \Sigma^n$ if $Y[i]=Y[i+p]$ holds for all $i\in [0\dd n-p)$.
In this case, the prefix $Y[0\dd p)$ is called a \emph{string period} of $Y$.
By $\per(Y)$ we denote the smallest period of $Y$. 
The exponent of a string $Y$ is defined as $\exp(Y):=\frac{|Y|}{\per(Y)}$,
and we say that a string $Y$ is \emph{periodic} if $\exp(Y)\ge 2$.

\begin{theorem}[2-Period queries~\cite{KRRW15,DBLP:journals/siamcomp/BannaiIINTT17}]\label{thm:2per}
A string $X$ can be preprocessed in linear time so that
one can decide in constant time whether any given fragment $X[i\dd j)$ is periodic and, if so, compute its shortest period $\per(X[i\dd j))$.
\end{theorem}
% A \emph{run} in a string $Y$ is a periodic fragment $Y[i\dd j)$ which cannot be extended (to the left or right) without increasing its period $p:=\per(Y[i\dd j))$, that is, neither $Y[i\dd j-1)$ nor $Y[i-1\dd j)$ is a fragment with period~$p$.

% \newcommand{\Runs}{\textsf{Runs}}

% \begin{theorem}[\cite{DBLP:conf/focs/KolpakovK99}]\label{thm:runs}
% There exists a linear-time algorithm that, given a string $S$, returns the set $\Runs(S)$ of all runs in $S$, each associated with its period. 
% \end{theorem}

% \begin{lemma}[Fact 2.2.4,~\cite{phd}]\label{lem:runoverlap}
%     Any two distinct runs $S[i_1\dd j_1)$ with period $p_1$ and $S[i_2\dd j_2)$ with period $p_2$ in a string $S$ satisfy $|[i_1 \dd j_1) \cap [i_2 \dd j_2)| < p_1 + p_2 - \gcd(p_1, p_2)$.
% \end{lemma}

For a string $Y$ and an integer $m\ge 0$, we define the $m$th power of $Y$, denoted $Y^m$, as the concatenation of $m$ copies of~$Y$. 
A non-empty string $Y\in \Sigma^n$ is \emph{primitive} if it cannot be expressed as $Y=X^m$ for some string $X$ and integer $m>1$.
For a string $Y\in \Sigma^n$, we define a \emph{forward rotation} $\rot(Y)=Y[1] \cdots Y[n-1]Y[0]$.
In general, a \emph{cyclic rotation} $\rot^s(Y)$ with \emph{shift} $s\in \mathbb{Z}$ is obtained by iterating $\rot$ or the inverse operation $\rot^{-1}$.
A string $Y$ is primitive if and only if it is distinct from its non-trivial rotations, i.e., if $Y=\rot^s(Y)$ holds only when $s$ is a multiple of $n$.

\subsection{Edit-Distance Alignments and Weighted Edit Distance}
\label{sec:editdistance}
In this subsection, we discuss alignments and their weighted cost, which provide a formal way to describe a sequence of edits needed to transform a string $X$ into $Y$.
\begin{definition}
  A sequence $\A = (x_t,y_t)_{t=0}^m$ is an \emph{alignment}
  of a fragment $X[x\dd x')$ onto a fragment $Y[y\dd y')$ if $(x_0,y_0)=(x,y)$,
  $(x_m,y_m)=(x',y')$, and $(x_{t+1},y_{t+1})\in \{(x_t+1,y_t+1),\allowbreak ({x_t+1},y_t),(x_t,y_t+1)\}$ for $t\in [0\dd m)$.
  The set of all alignments of $X[x\dd x')$ onto $Y[y\dd y')$ is denoted with $\aa(X[x\dd x'),Y[y\dd y'))$.
\end{definition}
Given an alignment $\A = (x_t,y_t)_{t=0}^m \in \aa(X[x\dd x'),Y[y\dd y'))$, for every $t\in [0\dd m)$, we say that
\begin{itemize}
  \item  $\A$ \emph{deletes} $X[x_t]$ if $(x_{t+1},y_{t+1})=(x_t+1,y_t)$.
  \item $\A$ \emph{inserts} $Y[y_t]$ if $(x_{t+1},y_{t+1})=(x_t,y_t+1)$.
  \item $\A$ \emph{aligns} $X[x_t]$ to $Y[y_t]$, denoted by $X[x_t] \sim_\A Y[y_t]$, if $(x_{t+1},y_{t+1})=(x_t+1,y_t+1)$. 
  \item $\A$ \emph{matches} $X[x_t]$ with $Y[y_t]$, denoted by $X[x_t] \simeq_\A Y[y_t]$, if $X[x_t]\sim_\A Y[y_t]$ and $X[x_t]=Y[y_t]$.
  \item $\A$  \emph{substitutes} $X[x_t]$ for $Y[y_t]$ if $X[x_t]\sim_\A Y[y_t]$ but $X[x_t]\ne Y[y_t]$.
\end{itemize}
Insertions, deletions, and substitutions are jointly called \emph{(character) edits}.

\begin{example}
For an example of an alignment, consider strings $X =\a\b\c$ and $Y= \b\d$.  One optimal alignment $\A$ might be $\{(0, 0), (1, 0), (2, 1), (3, 2)\}$.  The pairs $(0, 0), (1, 0)$ represent a deletion of $X[0] =\a$ by $\A$.  The pairs $(1, 0), (2, 1), (3, 2)$ signify that $\A$ aligns $X[1 \dd 2] \sim_\A Y[0 \dd 1]$, i.e. $\b\c \sim_\A \b\d$. Moreover, $X[1]$ is matched to $Y[0]$ since $X[1] = Y[0] = \b$ while $X[2]$ is substituted for $Y[1]$ since $X[2] = \c \ne \d = Y[1]$.
\end{example}

\newcommand{\eSigma}{\bar{\Sigma}}
For an alphabet $\Sigma$, we define $\eSigma := \Sigma\cup\{\eps\}$, where $\eps$ is the empty string over $\Sigma$.
We say that a function $w : \eSigma\times \eSigma \to \mathbb{R}_{\ge 0}\cup \{\infty\}$ is a \emph{weight function} if $w(a,a)=0$ holds for all $a\in \eSigma$. %We say that a character $\$\in \Sigma$ is a \emph{sentinel character} if $w(a,\$)=w(\$,a)=\infty$ holds for all $a\in \eSigma \setminus \{\$\}$.
The \emph{cost} of an alignment $\A\in \aa(X[x\dd x'),Y[y\dd y'))$ with respect to a weight function $w$, denoted $\ed_\A^w(X[x\dd x'),Y[y\dd y'))$, is defined as the total cost of edits that $\A$ performs, where:
\begin{itemize}
    \item the cost of deleting $X[x]$ is $w(X[x],\eps)$,
    \item the cost of inserting $Y[y]$ is $w(\eps, Y[y])$,
    \item the cost of substituting $X[x]$ for $Y[y]$ is $w(X[x],Y[y])$.
\end{itemize}
The \emph{width} of an alignment $(x_t,y_t)_{t=0}^m\in \aa(X[x\dd x'),Y[y\dd y'))$ is defined as $\max_{t=0}^m |x_t-y_t|$.

We usually consider alignments of the entire string $X[0\dd |X|)$ onto the entire string $Y[0\dd |Y|)$,
and we denote the set of all such alignment with $\aa(X,Y)=\aa(X[0\dd |X|),Y[0\dd |Y|))$.
The weighted edit distance of strings $X,Y\in \Sigma^*$ with respect to a weight function $w$ is defined as $\ed^w(X,Y)=\min_{\A \in \aa(X,Y)} \ed^w_{\A}(X,Y)$. For $k\in \mathbb{R}_{\ge 0}$, we also denote
\[\ed^w_{\le k}(X,Y) = \begin{cases}
\ed^w(X,Y) & \text{if }\ed^w(X,Y)\le k,\\
\infty & \text{otherwise.}
\end{cases}\]

In the literature, the (weighted) edit distance of $X$ and $Y$ is sometimes defined as the minimum cost of a sequence of edits that transform $X$ into $Y$.
As shown in the following fact (whose technical proof is deferred to \cref{app:proofs}), this sequence-based view is equivalent to our alignment-based view provided that $w$ is a \emph{quasimetric}, that is, it satisfies the triangle inequality $w(a,b)+w(b,c)\ge w(a,c)$ for every $a,b,c\in \eSigma$.
The assumption of $w$ being quasimetric can be made without loss of generality in the sequence-based view (a single character can be edited multiple times, so one can replace $w$ by its distance closure without affecting the edit distances).
Our alignment-based view, on the other hand, is more general and captures weighted edit distances violating the triangle inequality.

\newcommand{\hx}{\hat{x}}
\newcommand{\hy}{\hat{y}}
\newcommand{\hm}{\hat{m}}
\newcommand{\hz}{\hat{z}}

\begin{restatable}{fact}{edtri}\label{fct:edtri}
If $w$ is a quasimetric on $\eSigma$, then $\ed^w$ is a quasimetric on $\Sigma^*$. In this case,
$\ed^w(X,Y)$ can be equivalently defined as the minimum cost of a sequence of edits transforming $X$ into $Y$.
\end{restatable}

Although our algorithm for strings works for any weight function, its tree and Dyck counterparts
assume that $w$ is a quasimetric. Specifically, they rely on the following fact proved in \cref{app:proofs}.

\begin{restatable}{fact}{fctquasi}\label{fct:quasi}
    Consider a string $X$ and its fragment $X[i\dd j)$.
    Then, for every quasimetric $w$, we have $\ed^w(X,X[i\dd j))=\ed^w(X[0\dd i)\cdot X[j\dd |X|),\varepsilon)$.
\end{restatable}

While our main results are on the weighted version of    edit distance, our algorithm relies on unweighted edit distance procedures as well. If $w$ is the discrete metric on $\eSigma$ (that is, for every $a,b\in \eSigma$, we have $w(a,b)=0$ if $a=b$ and $w(a,b)=1$ otherwise),
then we drop the superscript $w$ in $\ed^w$ and $\ed^w_{\A}$. This yields the unit-cost edit distance (also known as the unweighted edit distance or the Levenshtein distance).
We consider weight function $w$ to be \emph{normalized} that is $w(a,b)\ge 1$ holds for all $a,b\in \eSigma$ with $a\ne b$.
In this case, $\ed^w_{\A}(X,Y)\ge \ed_\A(X,Y)$ holds for all strings $X,Y\in \Sigma^*$ and alignments $\A \in \aa(X,Y)$.

Given an alignment $\A= (x_t,y_t)_{t=0}^m\in \aa(X,Y)$, for every $\ell,r\in [0\dd m]$ with $\ell\le r$,
we say that $\A$ \emph{aligns} $X[x_\ell\dd x_{r})$ to $Y[y_\ell\dd y_{r})$, denoted $X[x_\ell\dd x_{r})\sim_\A Y[y_{\ell}\dd y_{r})$. 
In this case, for any weight function $w$, we write $\ed^w_{\A}(X[x_\ell\dd x_{r}), Y[y_\ell\dd y_{r}))$
to denote the cost of the induced alignment of $X[x_\ell\dd x_{r})$ onto $Y[y_\ell\dd y_{r})$.
If $\ed^w_{\A}(X[x_\ell\dd x_{r}), Y[y_\ell\dd y_{r}))=0$, we say that $\A$ \emph{matches} $X[x_\ell\dd x_{r})$ with $Y[y_\ell\dd y_{r})$, denoted $X[x_\ell\dd x_{r})\simeq_\A Y[y_{\ell}\dd y_{r})$. 

\begin{fact}\label{fct:str_decomp}
    Consider $k\in \Zz$, strings $X,Y$, and an alignment $\A\in \aa(X,Y)$ of cost $\ed_\A(X,Y)\le k$.
    Then, the string $X$ can be partitioned into at most $k$ individual characters (that $\A$ deletes or substitutes)
    and at most $k+1$ fragments that $\A$ matches perfectly to fragments of $Y$.
\end{fact}
% \begin{algorithm}
%     $i \gets 0$\;
%     $S_0, S_1, \ldots, S_k \gets \eps$\;
%     $C \gets \emptyset$\;
%     \For{$t = 1$ \KwSty{to} $m$}{
%         \If{$(x_t, y_t)$ = $(x_{t-1} + 1, y_{t-1} + 1)$ \KwSty{and} $X[x_{t-1}] = Y[y_{t-1}]$}{
%             $S_i \gets S_i \cdot X[x_{t-1}]$\;
%         }
%         \ElseIf{$x_t = x_{t-1} + 1$}{
%             $C \gets C \cup \{x_{t-1}\}$\;
%             $i \gets i + 1$\;
%         }
        
%     }
%     \KwRet{$S_0, S_1, \ldots, S_k, C$}\;
%     \caption{Partition($X, Y, (x_t, y_t)_{t=0}^m$): Constructs a partition of string $X$ into $k$ fragments}\label{alg:strpartition}
% \end{algorithm}
\begin{proof}
Let $\A = (x_t,y_t)_{t=0}^m$ and let $t_1<\cdots<t_e$ be the indices in $[0\dd m)$ corresponding to edits in $\A$.
Then, the maximal fragments that $\A$ matches perfectly are $X[0\dd x_{t_1})$, $X[x_{t_i+1}\dd x_{t_{i+1}})$ for $i\in [1\dd e)$,
and $X[x_{t_e+1}\dd |X|)$.  Moreover, $\A$ deletes or substitutes $X[x_{t_i}]$ for every $i\in [1\dd e]$ such that $x_{t_{i}+1}>x_{t_i}$.
Each edit contributes one unit to the cost of $\A$, so the decomposition contains at most $e\le k$ edited characters and $e+1\le k+1$ fragments matched perfectly.
\end{proof}
\subsection{Combinatorial Foundations}\label{subsec:str_comb}

Before giving our algorithms for weighted string edit distance, we discuss edit distance equivalent substrings, one of our main technical contributions. 

\begin{definition}\label{def:strequi}
    For $k\in \Zz$ and a weight function $w$, strings $P,P'$ are called \emph{$\ed_{\le k}^w$-equivalent}
    if 
    \[\ed_{\le k}^w(X,Y) = \ed_{\le k}^w(X[0 \dd p_X) \cdot P' \cdot X[p_X+|P|\dd |X|),Y[0 \dd p_Y) \cdot P' \cdot Y[p_Y+|P|\dd |Y|))\]
    holds for all strings $X$ and $Y$ in which $P$ occurs at positions $p_X$ and $p_Y$, respectively, satisfying $|p_X-p_Y|\le k$. We say that such occurrences of $P$ in $X$ and $Y$ are \emph{$k$-synchronized occurrences}.
\end{definition}

First, we prove that changing the power of a periodic substring does not change the edit distance cost of any synchronized occurrences of that substring. Second, we prove that we only need to consider the small prefixes and suffixes of substrings when calculating the edit distance. In both cases, we are able to show that any optimal alignment must align such large substrings (first periodic and then non-periodic) that have synchronized occurrences in edit distance instances, and so, we do not have to worry about most of these large substrings when calculating edit distance. See \cref{fig:per_reduc} and \cref{fig:ed_equiv} for examples of these two edit distance equivalent steps.

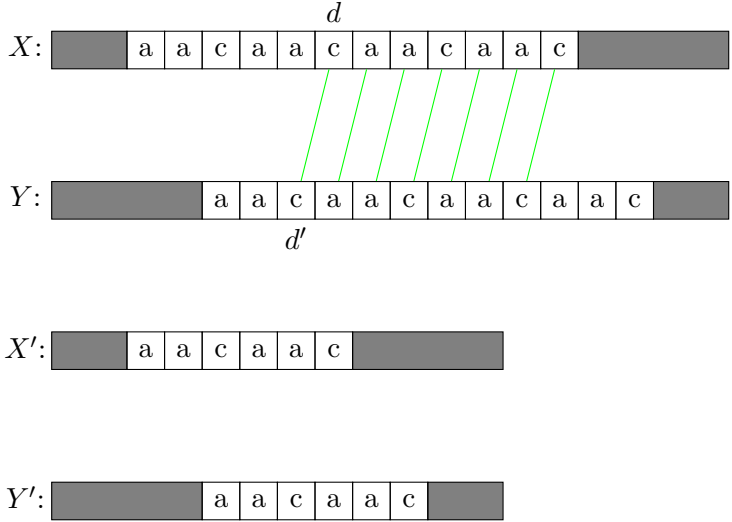
\begin{figure}
\begin{tikzpicture}
    \draw [fill=gray] (0, 2) rectangle (1, 2.5);
    \draw [fill=gray] (7, 2) rectangle (9, 2.5);    
    
    \node[draw,align=center, minimum size = .5cm] at  (1.25, 2.25) (top1) {a};
    \node[draw,align=center, minimum size = .5cm] at  (1.75, 2.25) (top2) {a};
    \node[draw,align=center, minimum size = .5cm] at  (2.25, 2.25) (top3) {c};
    \node[draw,align=center, minimum size = .5cm] at  (2.75, 2.25) (top4) {a};
    \node[draw,align=center, minimum size = .5cm] at  (3.25, 2.25) (top5) {a};
    \node[draw,align=center, minimum size = .5cm] at  (3.75, 2.25) (top6) {c};
    \node[draw,align=center, minimum size = .5cm] at  (4.25, 2.25) (top7) {a};
    \node[draw,align=center, minimum size = .5cm] at  (4.75, 2.25) (top8) {a};
    \node[draw,align=center, minimum size = .5cm] at  (5.25, 2.25) (top9) {c};
    \node[draw,align=center, minimum size = .5cm] at  (5.75, 2.25) (top10){a};
    \node[draw,align=center, minimum size = .5cm] at  (6.25, 2.25) (top11) {a};
    \node[draw,align=center, minimum size = .5cm] at  (6.75, 2.25) (top12) {c};

    \draw [fill=gray] (0, 0) rectangle (2, .5);
    \draw [fill=gray] (8, 0) rectangle (9, .5);    
    
    \node[draw,align=center, minimum size = .5cm] at  (2.25, .25) (bottom1) {a};
    \node[draw,align=center, minimum size = .5cm] at  (2.75, .25) (bottom2) {a};
    \node[draw,align=center, minimum size = .5cm] at  (3.25, .25) (bottom3) {c};
    \node[draw,align=center, minimum size = .5cm] at  (3.75, .25) (bottom4) {a};
    \node[draw,align=center, minimum size = .5cm] at  (4.25, .25) (bottom5) {a};
    \node[draw,align=center, minimum size = .5cm] at  (4.75, .25) (bottom6) {c};
    \node[draw,align=center, minimum size = .5cm] at  (5.25, .25) (bottom7) {a};
    \node[draw,align=center, minimum size = .5cm] at  (5.75, .25) (bottom8) {a};
    \node[draw,align=center, minimum size = .5cm] at  (6.25, .25) (bottom9) {c};
    \node[draw,align=center, minimum size = .5cm] at  (6.75, .25) (bottom10) {a};
    \node[draw,align=center, minimum size = .5cm] at  (7.25, .25) (bottom11) {a};
    \node[draw,align=center, minimum size = .5cm] at  (7.75, .25) (bottom12) {c};
    
    %\draw [color=blue] (top1) -- (1, .5);
    %\draw [color=blue]  (top2) -- (bottom1);
    %\draw [color=blue] (top3) -- (bottom2);
    %\draw [color=blue] (top4) -- (bottom2);
    %\draw [color=blue] (top5) -- (bottom3);
    \draw [color=green] (top6) -- (bottom3);
    \draw [color=green] (top7) -- (bottom4);
    \draw [color=green] (top8) -- (bottom5);
    \draw [color=green] (top9) -- (bottom6);
    \draw [color=green] (top10) -- (bottom7);
    \draw [color=green] (top11) -- (bottom8);
    \draw [color=green] (top12) -- (bottom9);
    %\draw [color=blue] (bottom10) -- (7.5, 2);
    %\draw [color=blue] (bottom11) -- (7.75, 2);
    %\draw [color=blue] (bottom12) -- (8, 2);
    
    \draw [fill=gray] (0, -2) rectangle (1, -1.5);
    \draw [fill=gray] (4, -2) rectangle (6, -1.5);    
    
    \node[draw,align=center, minimum size = .5cm] at  (1.25, -1.75) (top1-2) {a};
    \node[draw,align=center, minimum size = .5cm] at  (1.75, -1.75) (top2-2) {a};
    \node[draw,align=center, minimum size = .5cm] at  (2.25, -1.75) (top3-2) {c};
    \node[draw,align=center, minimum size = .5cm] at  (2.75, -1.75) (top4-2) {a};
    \node[draw,align=center, minimum size = .5cm] at  (3.25, -1.75) (top5-2) {a};
    \node[draw,align=center, minimum size = .5cm] at  (3.75, -1.75) (top6-2) {c};

    \draw [fill=gray] (0, -4) rectangle (2, -3.5);
    \draw [fill=gray] (5, -4) rectangle (6, -3.5);    
    
    \node[draw,align=center, minimum size = .5cm] at  (2.25, -3.75) (bottom1-2) {a};
    \node[draw,align=center, minimum size = .5cm] at  (2.75, -3.75) (bottom2-2) {a};
    \node[draw,align=center, minimum size = .5cm] at  (3.25, -3.75) (bottom3-2) {c};
    \node[draw,align=center, minimum size = .5cm] at  (3.75, -3.75) (bottom4-2) {a};
    \node[draw,align=center, minimum size = .5cm] at  (4.25, -3.75) (bottom5-2) {a};
    \node[draw,align=center, minimum size = .5cm] at  (4.75, -3.75) (bottom6-2) {c};
    
    %\draw [color=blue] (top1-2) -- (1, -3.5);
    %\draw [color=blue] (top2-2) -- (bottom1-2);
    %\draw [color=blue] (top3-2) -- (bottom2-2);
    %\draw [color=blue] (top4-2) -- (bottom2-2);
    %\draw [color=blue] (top5-2) -- (bottom3-2);
    %\draw [color=blue] (top6-2) -- (bottom3-2);
    %\draw [color=blue] (bottom4-2) -- (4.5, -2);
    %\draw [color=blue] (bottom5-2) -- (4.75, -2);
    %\draw [color=blue] (bottom6-2) -- (5, -2);
    
    \node at (-.35,2.3) {$X$:};
    \node at (-.35, .3) {$Y$:};
    \node at (-.35,-1.7) {$X'$:};
    \node at (-.35,-3.7) {$Y'$:};
    
    \node at (3.75, 2.75) {$d$};
    \node at (3.25, -.25) {$d'$};
\end{tikzpicture}
\caption{Periodicity reduction in $X$ and $Y$ with an optimal alignment depicted by lines connecting characters of the two strings. At indices $d, d'$ in $X$ and $Y$ respectively, the periodic substring is fully aligned (depicted by green lines), and so, we may reduce the power of these periodic substrings to construct $X'$ and $Y'$ with $\ed(X, Y) = \ed(X', Y')$.}\label{fig:per_reduc}
\end{figure}

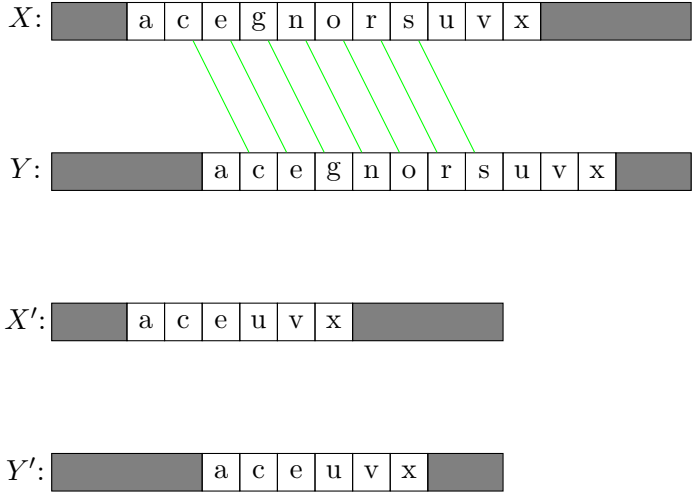
\begin{figure}
\begin{tikzpicture}
    \draw [fill=gray] (0, 2) rectangle (1, 2.5);
    \draw [fill=gray] (6.5, 2) rectangle (8.5, 2.5);    
    
    \node[draw,align=center, minimum size = .5cm] at  (1.25, 2.25) (top1) {a};
    \node[draw,align=center, minimum size = .5cm] at  (1.75, 2.25) (top2) {c};
    \node[draw,align=center, minimum size = .5cm] at  (2.25, 2.25) (top3) {e};
    \node[draw,align=center, minimum size = .5cm] at  (2.75, 2.25) (top4) {g};
    \node[draw,align=center, minimum size = .5cm] at  (3.25, 2.25) (top5) {n};
    \node[draw,align=center, minimum size = .5cm] at  (3.75, 2.25) (top6) {o};
    \node[draw,align=center, minimum size = .5cm] at  (4.25, 2.25) (top7) {r};
    \node[draw,align=center, minimum size = .5cm] at  (4.75, 2.25) (top8) {s};
    \node[draw,align=center, minimum size = .5cm] at  (5.25, 2.25) (top9) {u};
    \node[draw,align=center, minimum size = .5cm] at  (5.75, 2.25) (top10){v};
    \node[draw,align=center, minimum size = .5cm] at  (6.25, 2.25) (top11) {x};

    \draw [fill=gray] (0, 0) rectangle (2, .5);
    \draw [fill=gray] (7.5, 0) rectangle (8.5, .5);    
    
    \node[draw,align=center, minimum size = .5cm] at  (2.25, .25) (bottom1) {a};
    \node[draw,align=center, minimum size = .5cm] at  (2.75, .25) (bottom2) {c};
    \node[draw,align=center, minimum size = .5cm] at  (3.25, .25) (bottom3) {e};
    \node[draw,align=center, minimum size = .5cm] at  (3.75, .25) (bottom4) {g};
    \node[draw,align=center, minimum size = .5cm] at  (4.25, .25) (bottom5) {n};
    \node[draw,align=center, minimum size = .5cm] at  (4.75, .25) (bottom6) {o};
    \node[draw,align=center, minimum size = .5cm] at  (5.25, .25) (bottom7) {r};
    \node[draw,align=center, minimum size = .5cm] at  (5.75, .25) (bottom8) {s};
    \node[draw,align=center, minimum size = .5cm] at  (6.25, .25) (bottom9) {u};
    \node[draw,align=center, minimum size = .5cm] at  (6.75, .25) (bottom10) {v};
    \node[draw,align=center, minimum size = .5cm] at  (7.25, .25) (bottom11) {x};
    
    %\draw [color=blue] (top1) -- (1, .5);
    %\draw [color=blue]  (top2) -- (bottom1);
    \draw [color=green] (top2) -- (bottom2);
    \draw [color=green] (top3) -- (bottom3);
    \draw [color=green] (top4) -- (bottom4);
    \draw [color=green] (top5) -- (bottom5);
    \draw [color=green] (top6) -- (bottom6);
    \draw [color=green] (top7) -- (bottom7);
    \draw [color=green] (top8) -- (bottom8);
    %\draw [color=blue] (top9) -- (bottom9);
    %\draw [color=blue] (top9) -- (bottom10);
    %\draw [color=blue] (top9) -- (bottom11);
    %\draw [color=blue] (top10) -- (7.75, .5);
    %\draw [color=blue] (top11) -- (8, .5);
    
    \draw [fill=gray] (0, -2) rectangle (1, -1.5);
    \draw [fill=gray] (4, -2) rectangle (6, -1.5);    
    
    \node[draw,align=center, minimum size = .5cm] at  (1.25, -1.75) (top1-2) {a};
    \node[draw,align=center, minimum size = .5cm] at  (1.75, -1.75) (top2-2) {c};
    \node[draw,align=center, minimum size = .5cm] at  (2.25, -1.75) (top3-2) {e};
    \node[draw,align=center, minimum size = .5cm] at  (2.75, -1.75) (top4-2) {u};
    \node[draw,align=center, minimum size = .5cm] at  (3.25, -1.75) (top5-2) {v};
    \node[draw,align=center, minimum size = .5cm] at  (3.75, -1.75) (top6-2) {x};

    \draw [fill=gray] (0, -4) rectangle (2, -3.5);
    \draw [fill=gray] (5, -4) rectangle (6, -3.5);    
    
    \node[draw,align=center, minimum size = .5cm] at  (2.25, -3.75) (bottom1-2) {a};
    \node[draw,align=center, minimum size = .5cm] at  (2.75, -3.75) (bottom2-2) {c};
    \node[draw,align=center, minimum size = .5cm] at  (3.25, -3.75) (bottom3-2) {e};
    \node[draw,align=center, minimum size = .5cm] at  (3.75, -3.75) (bottom4-2) {u};
    \node[draw,align=center, minimum size = .5cm] at  (4.25, -3.75) (bottom5-2) {v};
    \node[draw,align=center, minimum size = .5cm] at  (4.75, -3.75) (bottom6-2) {x};
    
    %\draw [color=blue] (top1-2) -- (1, -3.5);
    %\draw [color=blue] (top2-2) -- (bottom1-2);
    %\draw [color=blue] (top2-2) -- (bottom2-2);
    %\draw [color=blue] (top3-2) -- (bottom3-2);
    %\draw [color=blue] (top4-2) -- (bottom4-2);
    %\draw [color=blue] (top4-2) -- (bottom5-2);
    %\draw [color=blue] (top4-2) -- (bottom6-2);
    %\draw [color=blue] (top5-2) -- (5.25, -3.5);
    %\draw [color=blue] (top6-2) -- (5.5, -3.5);
    
    \node at (-.35,2.3) {$X$:};
    \node at (-.35, .3) {$Y$:};
    \node at (-.35,-1.7) {$X'$:};
    \node at (-.35,-3.7) {$Y'$:};

\end{tikzpicture}
\caption{For any synchronized occurrences of a substring $P$ that avoids $k$-periodicity, any optimal alignment (depicted by lines connecting characters of the two strings) must match most of the inner characters of $P$ (see green lines). We can construct strings $X', Y'$ removing these matched characters such that $\ed(X, Y) = \ed(X', Y')$.}\label{fig:ed_equiv}
\end{figure}

\begin{lemma}\label{lem:periodic_reduction}
    Let $k\in \Zp$, let $Q$ be a string, and let $e,e'\in \mathbb{Z}_{\ge 4k}$.
    Then, $Q^e$ and $Q^{e'}$ are $\ed_{\le k}^w$-equivalent for every weight function~$w$.
 \end{lemma}
\begin{proof}
    We assume without loss of generality that $Q$ is primitive. (If $Q=R^m$ for $m\in \mathbb{Z}_{\ge 2}$, then $Q^e=R^{me}$ and $Q^{e'}=R^{me'}$
    can be interpreted as powers of $R$ rather than powers of $Q$.)
    Suppose that $Q^e$ occurs in strings $X$ and $Y$ at positions $p_X$ and $p_Y$, respectively, satisfying $|p_X-p_Y|\le k$.
    Denote $X' = X[0 \dd p_X) \cdot Q^{e'} \cdot X[p_X+|Q^e| \dd |X|)$ and 
    $Y' = Y[0 \dd p_Y) \cdot Q^{e'} \cdot Y[{p_Y+|Q^e|} \dd |Y|)$.
    Moreover, let $q=|Q|$ and let $\A\in \aa(X,Y)$ be an alignment such that $\ed^w(X,Y)= \ed^w_\A(X,Y) \le k$. 

    \begin{claim}\label{clm:periodic_reduction}
        There exist $i_X,i_Y\in [0\dd 3k]$ such that 
        \[X[p_X+i_X\cdot q\dd p_X+(i_X+1)\cdot q) \simeq_{\A} Y[p_Y+i_Y\cdot q\dd p_Y+(i_Y+1)\cdot q).\]
    \end{claim}
    \begin{proof} 
    Let $(t_X,t_Y)\in \A$ be the leftmost element of $\A$ such that $t_X \ge p_X$ and $t_Y \ge p_Y$.
    By symmetry between $X$ and $Y$, we assume without loss of generality that $t_X = p_X$.
    Consider the $k+1$ occurrences of $Q$ in $X$ starting at positions $p_X+i\cdot q$ for $i\in [0\dd k]$.
    The alignment $\A$ matches at least one of them exactly;
    we can thus define $i_X\in [0\dd k]$ so that $\A$ matches $X[p_X+i_X\cdot  q\dd p_X+(i_X+1)\cdot q)$ exactly to some fragment $Y[s_Y\dd s_Y+q)$.
    Due to $(t_X,t_Y)\in \A$, the non-crossing property of $\A$ implies that $s_Y \ge t_Y \ge p_Y$.
    Moreover, since $\ed_\A(X,Y) \leq k$ and $X[p_X + i_X \cdot q] \simeq_\A Y[s_Y]$, we have $s_Y \le (p_X+i_X \cdot q)+k \le p_X + kq+k \le p_Y + kq + 2k \le p_Y + 3kq$.
    Furthermore, since $Q$ is primitive (i.e., distinct from all its non-trivial cyclic rotations), we conclude that $s_Y = p_Y+i_Y \cdot q$ for some $i_Y\in [0\dd 3k]$. 
    \end{proof}

    Now, if $Q^e=X[p_X\dd p_X+e\cdot q)=Y[p_Y\dd p_Y + e\cdot q)$ is replaced with $Q^{e'}$ for $e'\ge e-1$,
    we can interpret this as replacing $Q=X[p_X+i_X\cdot q\dd p_X+(i_X+1)\cdot q)=Y[p_Y+i_Y\cdot q\dd p_Y+(i_Y+1)\cdot q)$ 
    with $Q^{1+e'-e}$. By \cref{clm:periodic_reduction}, $\A$ can be trivially adapted without modifying its cost, and hence $\ed^w(X',Y')\le \ed^w_{\A}(X,Y)=\ed^w(X,Y)$.
    If $e'< e-1$, we repeat the above argument to decrement the exponent $e$ one step at a time, still concluding that $\ed^w(X',Y')\le \ed^w(X,Y)$.
    In either case, the converse inequality follows by symmetry between $(X,Y,e)$ and $(X',Y',e')$.
\end{proof}

We say that a string avoids $k$-periodicity if it does not contain any substring of the form $Q^{4k+1}$ with $|Q|\in [1\dd 2k]$.
\begin{lemma}\label{lem:aperiodic_reduction}
    Let $k\in \Zp$ and let $P,P'$ be strings of lengths at least $42k^3$ such that $P[0\dd 21k^3)=P'[0\dd 21k^3)$ and $P[|P|-21k^3\dd |P|)=P'[|P'|-21k^3\dd |P'|)$ avoid $k$-periodicity. Then, $P$ and $P'$ are $\ed_{\le k}^w$-equivalent for every weight function~$w$.
\end{lemma}
\begin{proof}
    Suppose that $P$ occurs in strings $X$ and $Y$ at positions $p_X$ and $p_Y$, respectively, satisfying $|p_X-p_Y|\le k$.
    Denote $X' = X[0 \dd p_X)\cdot P' \cdot X[p_X + |P|\dd |X|)$ and 
    $Y' = Y[0 \dd p_Y) \cdot P' \cdot Y[p_Y +|P| \dd |Y|)$.
    Moreover, let $\A\in \aa(X,Y)$ be an alignment such that $\ed^w(X,Y)= \ed^w_\A(X,Y) \le k$. 

    \begin{claim}\label{clm:aperiodic_reduction}
        There exist $d,e\in [0\dd 21k^3]$ such that \[X[p_X+d\dd p_X+|P|-e)\sim_{\A} Y[p_Y+d\dd p_Y+|P|-e).\]
    \end{claim}
    \begin{proof}
        Let us partition $X[p_X\dd p_X+21k^3)$ into individual characters representing deletions or substitutions of $\A$
        and maximal fragments that $\A$ matches perfectly (to fragments of $Y$).
        By \cref{fct:str_decomp}, the number of such maximal fragments is at most $k+1$ and their total length is at least $21k^3-k\ge 20k^3$.
        Hence, one of these fragments is of length at least $\frac{20k^3}{k+1}\ge 10k^2$.
        Thus, let $R:=X[r_X\dd r_X+|R|)$ be a fragment of length at least $10k^2$ contained in $X[p_X\dd p_X+21k^3)$
        that $\A$ matches perfectly to $Y[r_Y\dd r_Y+|R|)$.
        Moreover, let $r'_Y := r_X + p_Y - p_X$.
        If $r_Y = r'_Y$, then we set $d := r_X-p_X = r_Y-p_Y$ so that $(p_X+d,p_Y+d)\in \A$.
        Otherwise, both $Y[r_Y\dd r_Y+|R|)$ and $Y[r'_Y\dd r'_Y+|R|)$ are occurrences of $R$ in $Y$.
        Moreover, $0<|r_Y-r'_Y| \le |r_Y - r_X| + |r'_Y - r_X| \le \ed^{w}_{\A}(X,Y)+|p_Y-p_X|\le 2k$.
        Hence, $\per(R)\le |r_Y-r'_Y|\le 2k$ and $\exp(R)\ge \frac{|R|}{2k}\ge 4k+1$.
        Since $Y[r'_Y\dd r'_Y+|R|)$ is contained in $Y[p_Y\dd p_Y+21k^3)=P[0\dd 21k^3)$, this contradicts the assumption about $P[0\dd 21k^3)$
        avoiding $k$-periodicity.

        A symmetric argument shows that $(p_X+|P|-e,p_Y+|P|-e)$ holds for some $e\in [0\dd 21k^3]$,
        which lets us conclude that $X[p_X+d\dd p_X+|P|-e)\sim_{\A} Y[p_Y+d\dd p_Y+|P|-e)$.
    \end{proof}

    By \cref{clm:aperiodic_reduction}, we have $X[p_X+d\dd p_X+|P|-e)\sim_\A Y[p_Y+d\dd p_Y+|P|-e)$.
    Both fragments match $P[d\dd |P|-e)$, so the optimality of $\A$ guarantees $X[p_X+d\dd p_X+|P|-e)\simeq_\A Y[p_Y+d\dd p_Y+|P|-e)$.
    Hence, if $P=X[p_X\dd p_X+|P|)=Y[p_Y\dd p_Y+|P|)$ is replaced with $P'$,
    we can interpret this as $P[d\dd |P|-e)=X[p_X+d\dd p_X+|P|-e)=Y[p_Y+d\dd p_Y+|P|-e)$ with $P'[d\dd |P'|-e)$. 
    Since $X[p_X+d\dd p_X+|P|-e)\simeq_\A Y[p_Y+d\dd p_Y+|P|-e)$, the alignment $\A$ can be trivially adapted without modifying its cost,
    and therefore $\ed^w(X',Y')\le \ed^w_{\A}(X,Y)=\ed^w(X,Y)$.
    The converse inequality follows by symmetry between $(X,Y,P)$ and $(X',Y',P')$.
\end{proof}

% \begin{corollary}\label{cor:string_reduction}
% Let $k\in \Zp$. For every string $P$, there exists a string of length at most $42k^3$ that is $\ed_{\le k}^w$-equivalent to $P$ for every normalized weight function~$w$.
% \end{corollary}
% \begin{proof}
% We proceed by induction on $|P|$ with the trivial base case of $|P|\le 42k^3$.
% If $|P|\ge 42k^3$ and $P$ avoids $k$-periodicity, then \cref{lem:aperiodic_reduction} implies that $P$ is equivalent to  a string $P':= P[0\dd 21k^3)\cdot P[|P|-21k^3\dd |P|)$ of length $42k^3$.
% Thus, suppose that $P[i\dd j)=Q^{4k+1}$ with $|Q|\in [1\dd 2k]$.
% By \cref{lem:periodic_reduction}, $Q^{4k+1}$ is equivalent to $Q^{4k}$,
% and thus $P$ is equivalent to $P':=P[0\dd i)\cdot P[i+|Q|\dd |P|)$.
% By the inductive assumption, $P'$ is equivalent to some string $P''$ of length at most $42k^3$, and, by transitivity of the considered equivalence,
% $P$ is also equivalent to~$P''$.
% \end{proof}

\subsection{Algorithm}\label{subsec:stralg}
The following lemma lets us transform any string $P$ to a string $P'$ that avoids $k$-periodicity
and is $\ed_{\le k}^w$-equivalent to $P$ for every  weight function~$w$.
It is stated in a general form so that it can be reused in subsequent sections. 

\begin{lemma}\label{lem:perred}
Let $e\in \Zp$ and let $\Qf$ be a family of primitive strings of length at most $e$.
There is an algorithm that repeatedly transforms an input string $P$
by replacing an occurrence of $Q^{e+1}$ (for some $Q\in \Qf$) with an occurrence of $Q^e$,
arriving at a string $P'$ that does not contain any occurrence of $Q^{e+1}$ (for any $Q\in \Qf$).
Moreover, this algorithm can be implemented in linear time using a constant-time oracle that tests whether a given primitive fragment of $P$
belongs to $\Qf$.
\end{lemma}
\SetKwFunction{PeriodicityReduction}{PeriodicityReduction}
\begin{algorithm}[ht]

    \Fn{$\PeriodicityReduction(P,e,\Qf)$}{
    $R \gets \eps$\;
    $r \gets 0$\;
    \While{$r < |P|$}{
        \lIf{$r+2e \le |P|$ \KwSty{and} $P[r\dd r+2e)$ is periodic}{$q \gets \per(P[r\dd r+2e))$}
        \lElse{$q \gets 1$}
        $m \gets \max\{\ell : P[r\dd r+\ell)=P[r+q\dd r+q+\ell)\}$\;
        \If{$m \ge eq$ \KwSty{and} $P[r\dd r+q)\in \Qf$}{
            $r \gets r + q$\;
        }\Else{
            $R \gets R\cdot P[r]$\;
            $r \gets r+1$\;
        }
    }
    \KwRet{$R$}\;
    }    \caption{Caps the exponent of every power of $Q\in \Qf$ occurring in $P$ to at most $e$.}\label{alg:str_per_reduc}
\end{algorithm}
\begin{proof}
    At preprocessing, we construct data structures for LCE and 2-Period queries in $P$; see~\cref{thm:lce,thm:2per}.
    In the main phase, our algorithm scans the string $P$ from left to right maintaining a string $R$ and an index $r\in [0\dd |P|]$
    such that $R\cdot P[r\dd |P|)$:
    \begin{itemize}
        \item is obtained from $P$ by repeatedly replacing an occurrence of $Q^{e+1}$ (for some $Q\in \Qf$) with an occurrence of $Q^e$,
        \item does not contain any occurrence of $Q^{e+1}$ (for $Q\in \Qf$) starting at position smaller than $|R|$.
    \end{itemize}
    We initialize the process with $R:=\eps$ and $r:=0$.
    At each step, we test if $P[r\dd r+2q)$ is periodic; if so, we retrieve its shortest period $q$;
    otherwise, we set $q:=1$. 
    Then, we further check whether $P[r\dd r+q)\in \Qf$ and $P[r\dd r+eq)=P[r+q\dd r+q+eq)$.
    If both tests are successful, we move the index $r$ to position $r+q$.
    Otherwise, we append $P[r]$ to $R$ and increment $r$.

    Let us analyze the correctness of this algorithm.
    First, suppose that $P[r\dd |P|)$ does not have a prefix of the form $Q^{e+1}$ for any $Q\in \Qf$.
    In particular, $P[r\dd r+q)\notin \Qf$ or $m<eq$. Thus, our algorithm appends $P[r]$ to $R$ and increments $r$.
    The invariant remains satisfied because $R\cdot P[r\dd |P|)$ did not change and $P[r\dd |P|)$ had no prefix of the form  $Q^{e+1}$ for any $Q\in \Qf$.

    Next, suppose that $P[r\dd |P|)$ has a prefix of the form $Q^{e+1}$ for some $Q\in \Qf$.
    If $|Q|\ne 1$, then $|Q|$ is the shortest period of $P[r\dd r+2e)$ because $Q$ is primitive and $|Q|\le e$.
    If $|Q|=1$, on the other hand, then $P[r\dd r+2e)$ either has period 1 or at least $e+2$.
    In all cases, the algorithm correctly identifies $q=|Q|$.
    Moreover, the subsequent tests whether $P[r\dd r+q)$ belongs to $\Qf$ and  $m\ge eq$
    are successful. Hence, the algorithm transforms $R\cdot P[r\dd |P|)$ into $R\cdot P[r+q\dd |P|)$,
    which is a valid operation because the prefix $Q^{e+1}$ of $P[r\dd |P|)$ is replaced with the prefix $Q^e$ of $P[r+q\dd |P|)$.
    Thus, it remains to prove that $R\cdot P[r+q\dd |P|)$ does not contain any occurrence of $\hat{Q}^{e+1}$ (for any $\hat{Q}\in \Qf$)
    starting at position smaller than $|R|$.
    Since  $R\cdot P[r\dd |P|)$ did not contain such an occurrence, the occurrence of $\hat{Q}^{e+1}$ would need to end at position $|R|+m$ or larger.
    The fragment $P[r+q\dd r+q+m)$ thus has periods $q$ and $\hat{q}:=|\hat{Q}|$.
    Moreover, by primitivity of $Q$ and $\hat{Q}$, the Periodicity Lemma~\cite{FW65} implies $q=\hat{q}$ due to $m \ge eq \ge e+q-1 \ge \hat{q}+q-1$.
    However, this means that $q$ is a period of $P[r\dd r+q+m]$, contradicting the definition of $m$.

    The overall running time is linear, including the preprocessing and the query time of the data structures
    of \cref{thm:lce,thm:2per}, because each iteration of the \KwSty{while} loop costs constant time.
\end{proof}

\begin{corollary}\label{cor:string_reduction}
    There exists a linear-time algorithm that, given a string $P$ and an integer $k\in \Zp$,
    constructs a string of length at most $42k^3$ that is $\ed_{\le k}^w$-equivalent to $P$ for every  weight function~$w$.
\end{corollary}
\SetKwFunction{StringReduction}{StringReduction}
\begin{algorithm}[ht]
    \caption{Construct a string of length at most $42k^3$ that is $\ed_{\le k}^w$-equivalent to $P$.}\label{alg:str_nonper_reduc}
    \Fn{$\StringReduction(P,k)$}{
        $P' \gets \PeriodicityReduction(P, 4k, \{Q\in \Sigma^+ : |Q|\le 2k \text{ and $Q$ is primitive}\})$\;
        \lIf{$|P'| \geq 42k^3$}{\KwRet{$P'[0\dd 21k^3) \cdot P'[|P'| - 21k^3 \dd |P'|)$}}
        \lElse{\KwRet{$P'$}}
    }
\end{algorithm}
\begin{proof}
    We set $P':=\PeriodicityReduction(P,4k,\Qf)$ with $\Qf$ consisting of all primitive strings of length in $[1\dd 2k]$. 
    We return $P'':=P'[0\dd 21k^3) \cdot P'[|P'| - 21k^3 \dd |P'|)$ or $P'$ depending on whether $|P'|\ge 42k^3$ or not.
    By \cref{lem:periodic_reduction,lem:perred}, the string $P'$ is $\ed_{\le k}^w$-equivalent to $P$ and avoids $k$-periodicity. 
    Thus, if  $|P'| \le 42k^3$, then the algorithm is correct.
    Otherwise, \cref{lem:aperiodic_reduction} implies that $P''$ is  $\ed_{\le k}^w$-equivalent to $P'$ (and, by transitivity, to $P$)
    because $P'[0\dd 21k^3)$  and $P'[|P'| - 21k^3 \dd |P'|)$ avoid $k$-periodicity.
    Due to \cref{lem:perred}, the running time is linear (a primitive fragment belongs to $\Qf$ if and only if its length does not exceed $2k$,
    which takes $\Oh(1)$ time to test).
\end{proof}

%Pseudo-code
\begin{theorem}\label{thm:stringKernel}
    There exists a linear-time algorithm that, given strings $X$, $Y$ and an integer $k\in \Zp$,
    constructs strings $X'$, $Y'$ of lengths at most $85k^4$ such that $\ed^w_{\le k}(X,Y)=\ed^w_{\le k}(X',Y')$ holds for every weight function $w$.
\end{theorem}
\SetKwFunction{StringKernel}{StringKernel}
\begin{algorithm}[ht]
    \caption{Construct strings $X', Y'$ of length at most $85k^4$ such that $\ed^w_{\le k}(X,Y)=\ed^w_{\le k}(X',Y')$}\label{alg:str_kernel}
    \Fn{$\StringKernel(X,Y,k)$}{
    \lIf{$|X|\le 85k^4$ \KwSty{and} $|Y| \leq 85k^4$}{\KwRet{$(X, Y)$}}\label{ln:trivial}
    \lIf{$\ed(X, Y) > k$}{\KwRet{$(a^{k+1}, \eps)$} for some $a\in \Sigma$}\label{ln:large}
    Let $(x_t,y_t)_{t=0}^m\in \aa(X,Y)$ be an alignment satisfying $\ed_{\A}(X,Y)\le k$\;
    $X',Y',P\gets \varepsilon$\;
    \For{$t \gets 0$ \KwSty{to} $m$}{
        \If{$t <m$ \KwSty{and} $x_{t+1}>x_{t}$ \KwSty{and} $y_{t+1}>y_t$ \KwSty{and} $X[x_t]= Y[y_t]$}{
            $P \gets P\cdot X[x_t]$
        }\Else{
            $P \gets \StringReduction(P,k)$\;\label{ln:reduce}
            $X' \gets X'\cdot P$\;
            $Y' \gets Y'\cdot P$\;
            $P \gets \varepsilon$\;\label{ln:peps}
            \lIf{$t < m$ \KwSty{and} $x_{t+1}>x_t$}{$X' \gets X'\cdot X[x_t]$}
            \lIf{$t < m$ \KwSty{and} $y_{t+1}>y_t$}{$Y' \gets Y'\cdot Y[y_t]$}
        }  
    }
    % $S_0, \ldots, S_k, \{c_1, c_2, \ldots, c_i\} \gets$ Partition($X, Y, \A$)\;
    % $T_0, \ldots, T_k, \{d_1, d_2, \ldots, d_j\} \gets$ Partition($Y, X, \A^{-1})$\;
    % \ForEach{$\ell \in [0 \ldots k]$}{
    %     $S_\ell' \gets $
    %     StringPeriodReduction($S_\ell$)\;
    %     $T_\ell' \gets$ StringPeriodReduction($T_\ell$)\;
    % }
    % $X' \gets S_0' \cdot c_1 \cdot S_1' \cdot c_2 \cdot \ldots \cdot c_i \cdot S_i'$\;
    % $Y' \gets T_0' \cdot d_1 \cdot T_1' \cdot d_2 \ldots \cdot d_i \cdot T_i'$\;
    \KwRet{$(X', Y')$} 
    }
    \end{algorithm}
\begin{proof}
Our procedure is implemented as \cref{alg:str_kernel}. First, if $X$ and $Y$ are already of length at most $85k^4$, then we return $X$ and $Y$ unchanged.
If $\ed(X, Y)>k$, we return strings $a^{k+1}$ and $\eps$, where $a\in \Sigma$ is an arbitrary character.
If $\ed(X,Y)\le k$, we construct an alignment $\A := (x_t,y_t)_{t=0}^m\in \aa(X,Y)$ of (unweighted) cost at most $k$.
We then build the output strings $X'$ and $Y'$ during a left-to-right scan of the alignment $\A$:
We append to $X'$ and $Y'$ every character of $X$ and $Y$ (respectively) that $\A$ edits.
Moreover, for every pair of maximal fragments in $X$ and $Y$ that $\A$ matches perfectly,
we apply the reduction of \cref{cor:string_reduction} and append the resulting string to both $X'$ and $Y'$.

Let us now prove that the resulting instance $(X',Y')$ satisfies $\ed^w_{\le k}(X,Y)=\ed^w_{\le k}(X',Y')$.
This is trivial when the algorithm returns $(X,Y)$ in \cref{ln:trivial}.
If $\ed(X, Y)>k$, then $\ed_{\le k}(X,Y)=\infty = \ed_{\le k}(a^{k+1},\eps)$
and thus also $\ed^w_{\le k}(X,Y)= \infty = \ed^w_{\le k}(a^{k+1},\eps)$ because the weighted edit distance with a normalized weight function is at least as large as the unweighted edit distance.
In the remaining case of $\ed(X,Y)\le k$, we maintain an invariant that  $|X'|-|Y'|=x_t-y_t$ and $\ed_{\le k}^w(X,Y)=\ed_{\le k}^w(X'\cdot P \cdot X[x_t\dd x_m),Y'\cdot P\cdot Y[y_t\dd y_m))$ hold at the beginning of every iteration of the \KwSty{for} loop as well as after every execution of Line~\ref{ln:reduce} and Line~\ref{ln:peps}. It is easy to see that the strings $X'\cdot P \cdot X[x_t\dd x_m)$ and $Y'\cdot P\cdot Y[y_t\dd y_m)$ change only at Line~\ref{ln:reduce}, when $P$ is replaced with $\StringReduction(P,k)$. The correctness of this step follows directly from the definition of $\ed_{\le k}^w$-equivalence (\cref{def:strequi}) since $\StringReduction(P,k)$ is $\ed_{\le k}^w$-equivalent to $P$.

Next, we show that the returned strings are of length at most $85k^4$.
This is clear when the algorithm terminates at Line~\ref{ln:trivial} or~\ref{ln:large}.
Otherwise, we apply \cref{fct:str_decomp} to observe that $X$ is decomposed into at most $k$ characters that $\A$ deletes or substitutes
(which are copied to $X'$) and at most $k+1$ maximal fragments that $\A$ matches perfectly to fragments of $Y$ (which are copied to $X'$ after applying $\StringReduction$). By the guarantee of \cref{cor:string_reduction}, we conclude that $|X'|\le k + (k+1)\cdot 42k^3 \le 85k^4$.
Symmetrically, we have $|Y'|\le 85k^4$.

It remains to analyze the time complexity of our procedure.
We use the Landau--Vishkin algorithm~\cite{LV88} to check whether $\ed(X,Y)\le k$ and, if so, construct the alignment $\A$.
This costs $\Oh(n+k^2)$ time, which is $\Oh(n)$ because we perform this step only if $n \ge k^4 \ge k^2$.
The scan of the alignment $\A$ takes $\Oh(m)=\Oh(n)$ time, including the applications of \cref{cor:string_reduction},
which operate on strings of total length at most $n$.
\end{proof}

Having reduced the string lengths to $\Oh(k^4)$, we can use the classic dynamic programming~\cite{WF74} to compute $\ed_{\le k}^w(X,Y)$
in $\Oh(k^8)$ time. However, since $w$ is a normalized, the running of~\cite{WF74} can be reduced to $\Oh(nk)$.
For completeness, we describe this improvement below.

\begin{proposition}\label{prp:nk}
Given strings $X,Y$ of length at most $n$, an integer $k\in \Zp$, and a weight function $w$,
the value $\ed_{\le k}^w(X,Y)$ can be computed in $\Oh(nk)$ time.
\end{proposition}
\begin{proof}
Recall that the algorithm of~\cite{WF74} maintains a table $D[0\dd |X|,0\dd |Y|]$ such that $D[i,j]=\ed^w(X[0\dd i),Y[0\dd j))$
holds for each $i\in [0\dd |X|]$ and $j\in [0\dd |Y|]$.
We have $D[0,0]=0$, whereas the remaining entries are constructed in $\Oh(1)$ time each using the following formula:
\begin{equation}\label{eq:DP} D[i,j] = \min \begin{cases}
    D[i-1,j] + w(X[i-1],\eps) & \text{if }i>0,\\
    D[i,j-1] + w(\eps, Y[j-1]) & \text{if }j>0,\\
    D[i-1,j-1] + w(X[i-1],Y[j-1]) & \text{if }i,j>0.\end{cases}\end{equation}
In order to compute $\ed^w_{\le k}(X,Y)$, we use a modified table $D'[0\dd |X|,0\dd |Y|]$
such that $D'[0,0]=0$, $D'[i,j] = \infty$ if $|i-j|>k$, whereas the remaining entries are computed using~\eqref{eq:DP} (with $D$ replaced by $D'$).
A straightforward inductive argument shows that $D'[i,j]\ge D[i,j]$ holds for all $i\in [0\dd |X|]$ and $j\in [0\dd |Y|]$
and, moreover, $D[i,j]\le k$ implies $D'[i,j]= D[i,j]$.
For $|i-j|>k$, this is true because $w$ is normalized and thus $D[i,j]=\ed^w(X[0\dd i),Y[0\dd j))\ge \ed(X[0\dd i),Y[0\dd j)) \ge |i-j|>k$.
For $|i-j|\le k$, on the other hand, the argument is based on the inductive hypothesis and the fact that the weight function $w$ has non-negative values.
The entries $D'[i,j]=\infty$ for $|i-j|>k$ can be set implicitly, which reduces the running time to $\Oh(nk)$.
\end{proof}

\weighted*

%\begin{theorem}
%Given strings $X,Y$ of length at most $n$, an integer $k\in \Zp$, and a normalized weight function $w$,
%the value $\ed_{\le k}^w(X,Y)$ can be computed in $\Oh(n+k^5)$ time.
%\end{theorem}
\begin{proof}
We first apply \cref{thm:stringKernel} to build strings $X',Y'$ of length $\Oh(k^4)$ such that 
$\ed_{\le k}^w(X',Y')=\ed_{\le k}^w(X,Y)$. Then, we compute $\ed_{\le k}^w(X',Y')$ using \cref{prp:nk}.
The running times of these two steps are $\Oh(n)$ and $\Oh(k^4 \cdot k)=\Oh(k^5)$, respectively,
for a total of $\Oh(n+k^5)$.
\end{proof}

\section{Tree Edit Distance}
\label{sec:ted}

\subsection{Preliminaries}\label{subsec:treeprelim}

\newcommand{\PSigma}{\mathsf{P}_\Sigma}
\newcommand{\FSigma}{\mathcal{F}_\Sigma}

For an alphabet $\Sigma$, we define a set $\PSigma := \bigcup_{a\in \Sigma}\{\op_a,\cl_a\}$ of parentheses with labels over $\Sigma$.
A \emph{forest} with node labels over $\Sigma$ is a \emph{balanced} string of parentheses over $\Sigma$. Formally, the set of forests with labels over $\Sigma$ is defined as the smallest subset $\FSigma\sub \PSigma^*$ satisfying the following conditions:
\begin{itemize}
    \item $\eps \in \FSigma$,
    \item $F\cdot G \in \FSigma$ for every $F,G\in \FSigma$,
    \item $\op_a \cdot F \cdot \cl_a \in \FSigma$ for every $F\in \FSigma$ and $a\in \Sigma$.
\end{itemize}

For a forest $F$, we define the set of \emph{nodes} $V_F$ as the set of pairs $(i,j)\in [0\dd |F|)$
such that $F[i]$ is an opening parenthesis, $F[j]$ is a closing parenthesis, and $F[i\dd j]$ is balanced.
For a node $u= (i,j)\in V_F$, we denote the positions of the opening and the closing parenthesis by $o(u):=i$ and $c(u):= j$.
A forest $F$ is a tree if $(0,|F|-1)\in V_F$.

\begin{fact}\label{fct:balanced}
A forest $F$ can be preprocessed in linear time so that one can test in constant time whether any given fragment $F[i\dd j)$
is balanced.
\end{fact}
\begin{proof}
Let us define the height function $H:[0\dd |F|]\to \mathbb{Z}$ so that $H(i)$ equals the number of opening parentheses in $F[0\dd i)$
    minus the number of closing parentheses in $F[0\dd i)$. Since $F$ is balanced, the fragment $F[i\dd j)$
    is balanced if and only if $H(i)=H(j)=\min_{m\in [i\dd j]} H(m)$. This condition can be tested in $\Oh(1)$
    time after linear-time preprocessing using range minimum queries (RMQ)~\cite{10.1137/0213024}.
\end{proof}

A context with node labels over $\Sigma$ is a pair $C=\langle C_L;C_R \rangle\in \PSigma \times \PSigma$ such that $C_L\cdot C_R$ is a tree.
The node set $V_C$ of a context $C$ is identified with the node set of the underlying tree $C_L\cdot C_R$.
The \emph{depth} of a context $C$ is the number of nodes $u\in V_C$ whose opening parenthesis belongs to $C_L$ and closing parenthesis belongs to $C_R$,
that is, $o(u) < |C_L| < c(u)$.

The (vertical) composition of contexts $C,D$ results in a context $C\star D := \langle C_L\cdot D_L; D_R\cdot C_R \rangle$.
Moreover, vertical composition of a context $C$ and a forest $F$ results in a tree $C \star F := C_L\cdot F \cdot C_R$.
A context $C$ is primitive if it cannot be expressed as vertical composition of $e\ge 2$ copies of the same context.

A context $C$ \emph{occurs} in a forest $F$ at node $u\in V_F$
if $C_L=F[o(u)\dd o(u)+|C_L|)$ and $C_R = F(c(u)-|C_R|\dd c(u)]$, or equivalently, $F[o(u)\dd c(u)]= C\star G$ for some forest $G$

\subsection{Forest Alignments and Weighted Forest Edit Distance}

We begin our discussion of weighted tree edit distance by formally defining forest alignments, which are similar to alignments on strings with just a few additional restrictions to make sure the alignments make valid edits on forests.

\begin{definition}\label{def:ta}
    We say that an alignment $\A\in \aa(F,G)$ is a \emph{forest alignment} of forests $F$ and $G$
    if the following \emph{consistency} conditions are satisfied for each $u\in V_F$:
    \begin{itemize}
        \item either $\A$ deletes both $F[o(u)]$ and $F[c(u)]$, or
        \item there exists $v\in V_G$ such that $F[o(u)] \sim_\A G[o(v)]$ and $F[c(u)] \sim_\A G[c(v)]$.
    \end{itemize}
    The set of all forests alignments of $F$ onto $G$ is denoted with $\ta(F,G)\sub \aa(F,G)$.
\end{definition}

\newcommand{\EPSigma}{\overline{\PSigma}}
\newcommand{\PW}{\tilde{w}}

Define $\EPSigma = \PSigma \cup \eps$ and a mapping
$\lambda : \EPSigma \to \eSigma$ such that $\lambda(\op_a)=\lambda(\cl_a)=a$ for each $a\in \Sigma$, and $\lambda(\eps)=\eps$.
For a weight function $w: \eSigma\times \eSigma \to \mathbb{R}_{\ge 0}$, we define a corresponding weight function $\PW: \EPSigma \times \EPSigma \to \mathbb{R}_{\ge 0}$ so that $\PW(p,q)=w(\lambda(p),\lambda(q))$ for all $p,q\in \EPSigma$.
The \emph{cost} of a forest alignment $\A\in \ta(F,G)$ with respect to a weight function $w$
is defined as $\ted^w_{\A}(F,G) := \frac12 \ed^{\PW}_{\A}(F,G)$.
Moreover, for any two forests $F,G$, we define the \emph{weighted tree edit distance} $\ted^w(F,G)=\min_{\A\in \ta(F,G)}\ted^w_\A(F,G)$,
and for a threshold $k\in \mathbb{R}_\ge 0$, we set 
    \[\ted^w_{\le k}(F,G) = \begin{cases}
        \ted^w(F,G) & \text{if }\ted^w(F,G)\le k,\\
        \infty & \text{otherwise.}
    \end{cases}
    \]
The superscript is omitted if $w$ is the discrete metric over $\eSigma$.

\begin{fact}\label{fct:tedtri}
If $w$ is a quasimetric on $\eSigma$, then $\ted^w$ is a quasimetric on $\FSigma$. In that case, $\ted^w(F,G)$ can be equivalently defined as the minimum cost of a sequence of edits transforming $F$ into $G$, where inserting a node with label $b$ costs $w(\eps, b)$, deleting a node with label $a$ costs $w(a,\eps)$,
and changing a node label from $a$ to $b$ costs $w(a,b)$.
\end{fact}
\begin{proof}
    Consider arbitrary forests $F, G, H \in \FSigma$ as well as alignments $A = (x_t, y_t)_{t=0}^m \in \ta(F, G)$ and $B = (\hy_t, \hz_t)_{t=0}^{\hm} \in \ta(G, H)$.  We can construct the product alignment $\A \otimes \B$ as in the proof of Fact~\ref{fct:edtri}, which has $\ed_{\A \otimes \B}^w (F, H) \leq \ed_{\A}^w (F, G) + \ed_{\B}^w (G, H)$. Therefore, it remains to prove that $\A \otimes \B$ is a tree alignment.
        
    Consider an arbitrary node $u_F \in V_F$. If $\A$ deletes $u_F$ (that is, it deletes both characters $F[o(u_F)]$ and $F[c(u_F)]$), then $\A \otimes \B$
    also deletes $u_F$; see~Case~\ref{cs:del} in the recursive definition of $\A\otimes \B$.
    The other possibility is that $\A$ aligns $u_F$ with some node $u_G\in V_G$ (that is, it aligns $F[o(u_F)]$ with $G[o(u_G)]$ and $F[c(u_F)]$ with $G[c(u_G)]$).
    If $\B$ deletes $u_G$, then $\A \otimes \B$ deletes $u_F$; see Case~\ref{cs:ad} in the recursive definition of $\A\otimes \B$.
    Finally, if $\B$ aligns $u_G$ with some node $u_H\in V_H$, then $\A \otimes \B$ aligns $u_F$ with $u_H$; see Case~\ref{cs:aa} in the recursive definition of $\A\otimes \B$.
\end{proof}
\subsection{Combinatorial Foundations}

\subsubsection{Forests}

Similar to our discussion of weighted string edit distance, before giving our tree edit distance algorithms we prove the existence of small edit distance equivalent forests for synchronized occurrences of large subforests in the input instance forests.

\begin{definition}
    For $k\in \Zz$ and a weight function $w$, forests $P,P'$ are called \emph{$\ted_{\le k}^w$-equivalent}
    if 
    \[\ted_{\le k}^w(F,G) = \ted_{\le k}^w(F[0 \dd p_F) \cdot P' \cdot F[p_F+|P|\dd |F|),G[0 \dd p_G) \cdot P' \cdot G[p_G+|P|\dd |G|))\]
    holds for all forests $F$ and $G$ in which $P$ occurs at positions $p_F$ and $p_G$, respectively, satisfying $|p_F-p_G|\le 2k$.
\end{definition}

\begin{lemma}\label{lem:horizontal_periodic_reduction}
    Let $k\in \Zp$, let $Q$ be a forest, and let $e,e'\in \mathbb{Z}_{\ge 4k}$.
    Then, $Q^e$ and $Q^{e'}$ are $\ted_{\le k}^w$-equivalent for every normalized weight function $w$.
\end{lemma}
\begin{proof}
    We assume without loss of generality that $Q$ is primitive. (If $Q=R^m$ for $m\in \mathbb{Z}_{\ge 2}$, then $Q^e=R^{me}$ and $Q^{e'}=R^{me'}$
    can be interpreted as powers of $R$ rather than powers of $Q$.)
    Suppose that $Q^e$ occurs in forests $F$ and $G$ at positions $p_F$ and $p_G$, respectively, satisfying $|p_F-p_G|\le 2k$.
    Denote $F' = F[0\dd p_F)\cdot Q^{e'} \cdot F[p_F+|Q^e|\dd |F|)$ and $G' = G[0\dd p_G)\cdot Q^{e'}\cdot G[p_G + |Q^e|\dd |G|)$.
    Moreover, let $q=|Q|$ and let $\A$ be a forest alignment such that $\ted^{w}(F,G)=\ted^{w}_{\A}(F,G)\le k$.
    \begin{claim}\label{clm:horizontal_periodic_reduction}
        There exist $i_F,i_G\in [0\dd 3k]$ such that 
        \[F[p_F + i_F \cdot q \dd p_F + (i_F+1)\cdot q) \simeq_{\A} G[p_G + i_G\cdot q\dd p_G + (i_G+1)\cdot q).\]
    \end{claim}
    \begin{proof} Let $(f_b,g_b)\in \A$ be the leftmost element of $\A$ such that $f_b \ge p_F$ and $g_b \ge p_G$.
    By symmetry between $F$ and $G$, we assume without loss of generality that $f_b = p_F$.
    Consider the $k+1$ occurrences of $Q$ in $F$ starting at positions $p_F+i\cdot q$ for $i\in [0\dd k]$.
    Since $Q$ is balanced, the alignment $\A$ (of unweighted cost at most $k$) matches at least one of them exactly;
    we can thus define $i_F\in [0\dd k]$ so that $\A$ matches $F[p_F+i_F\cdot  q\dd p_F+(i_F+1)\cdot q)$ exactly to some fragment $G[g_a\dd g_a+q)$.
    By definition of $b$, we have $a \ge b$ and thus $g_a \ge g_b \ge p_G$.
    Moreover, since $\ted_\A(F, G) \leq k$ and $F[p_F + i_F \cdot q] \sim_\A G[g_a]$, we have $g_a \le (p_F+i_F \cdot q)+2k \le p_F + kq+2k \le p_G + kq + 4k \le p_G + 3kq$,
    where the last inequality follows from $q\ge 2$ (recall that $Q$ is balanced, so its length is even).
    Furthermore, since $Q$ is primitive (i.e., distinct from all its non-trivial cyclic rotations), we conclude that $g_a = p_G+i_G \cdot q$ for some $i_G\in [0\dd 3k]$. 
    \end{proof}

    Now, if $Q^e = F[p_F\dd p_F+e\cdot q)=G[p_G\dd p_G+e\cdot q)$ is replaced with $Q^{e'}$ for $e'\ge e-1$,
    we can interpret this as replacing $Q=F[p_F + i_F \cdot q \dd p_F + (i_F+1)\cdot q) = G[p_G + i_G\cdot q\dd p_G + (i_G+1)\cdot q)$
    with $Q^{1+e'-e}$. By \cref{clm:horizontal_periodic_reduction}, $\A$ can be trivially adapted without modifying its cost,
    and hence $\ted^w(F',G')\le \ted^w_{\A}(F,G)=\ted^w(F,G)$.
    If $e'< e-1$, we repeat the above argument to decrement the exponent one step at a time, still concluding that $\ted^w(F',G')\le \ted^w(F,G)$.
    In either case, the converse inequality follows by symmetry between $(F,G,e)$ and $(F',G',e')$.
\end{proof}

We say that a forest $F$ avoids horizontal $k$-periodicity
if there is no forest $Q$ of length $|Q|\in [1\dd 4k]$ such that $Q^{4k+1}$ occurs in $F$. 

\begin{figure}
\begin{tikzpicture}
\draw [fill=cyan] (2.375,3) circle (.25) node (root) {()};
\draw [fill=violet] (0,0) -- (.5,1.5) -- (1, 0) -- cycle;
\draw [fill=violet] (1.25,0) -- (1.75,1.5) -- (2.25, 0) -- cycle;
\draw [fill=violet] (2.5,0) -- (3,1.5) -- (3.5, 0) -- cycle;
\draw [fill=violet] (3.75,0) -- (4.25,1.5) -- (4.75, 0) -- cycle;
\draw (root) -- (.5,1.5);
\draw (root) -- (1.75, 1.5);
\draw (root) -- (3, 1.5);
\draw (root) -- (4.25, 1.5);
\draw [fill=cyan] (10, 3) circle (.25) node (2-1) {()};
\draw [fill=green] (10, 1) circle (.25) node (2-2) {\{\}};
\draw [fill=cyan] (10, -1) circle (.25) node (2-3) {()};
\draw [fill=green] (10, -3) circle (.25) node (2-4) {\{\}};
\draw [fill=violet] (8,1) -- (8.5,2.5) -- (9, 1) -- cycle;
\draw [fill=violet] (8,-3) -- (8.5,-1.5) -- (9, -3) -- cycle;
\draw (2-1) -- (8.5, 2.5);
\draw (2-3) -- (8.5, -1.5);
\draw [fill=orange] (11, 1.5) -- (11.25, 2.5) -- (11.5, 1.5) -- cycle;
\draw [fill=orange] (11, -.5) -- (11.25, .5) -- (11.5, -.5) -- cycle;
\draw [fill=orange] (11, -2.5) -- (11.25, -1.5) -- (11.5, -2.5) -- cycle;
\draw [fill=orange] (11, -4.5) -- (11.25, -3.5) -- (11.5, -4.5) -- cycle;
\draw (2-1) -- (11.25, 2.5);
\draw (2-2) --  (11.25, .5);
\draw (2-3) -- (11.25, -1.5);
\draw (2-4) -- (11.25, -3.5);
\draw [fill=gray] (9, -6) -- (10, -4) -- (11, -6) -- cycle;
\draw (2-1) -- (2-2) -- (2-3) -- (2-4) -- (10, -4);

\node[draw,align=left] at (2.4,4) {Horizontal Periodicity:};
\node[draw,align=right] at (9.8, 4) {Vertical Periodicity:};
\end{tikzpicture}

\caption{Pictured left: horizontal periodicity with string representation ``\textbf{\textcolor{cyan}{(}\textcolor{violet}{[\ldots][\ldots][\ldots][\ldots]}\textcolor{cyan}{)}}''. Pictured right: vertical periodicity with string representation ``\textbf{\textcolor{cyan}{(}\textcolor{violet}{[\ldots]}\textcolor{green}{\{}\textcolor{cyan}{(}\textcolor{violet}{[\ldots]}\textcolor{green}{\{}\textcolor{gray}{[\ldots]}\textcolor{orange}{[\ldots]}\textcolor{green}{\}}\textcolor{orange}{[\ldots]}\textcolor{cyan}{)}\textcolor{orange}{[\ldots]}\textcolor{green}{\}}\textcolor{orange}{[\ldots]}\textcolor{cyan}{)}}''.} 
\label{fig:tree_per}
\end{figure}
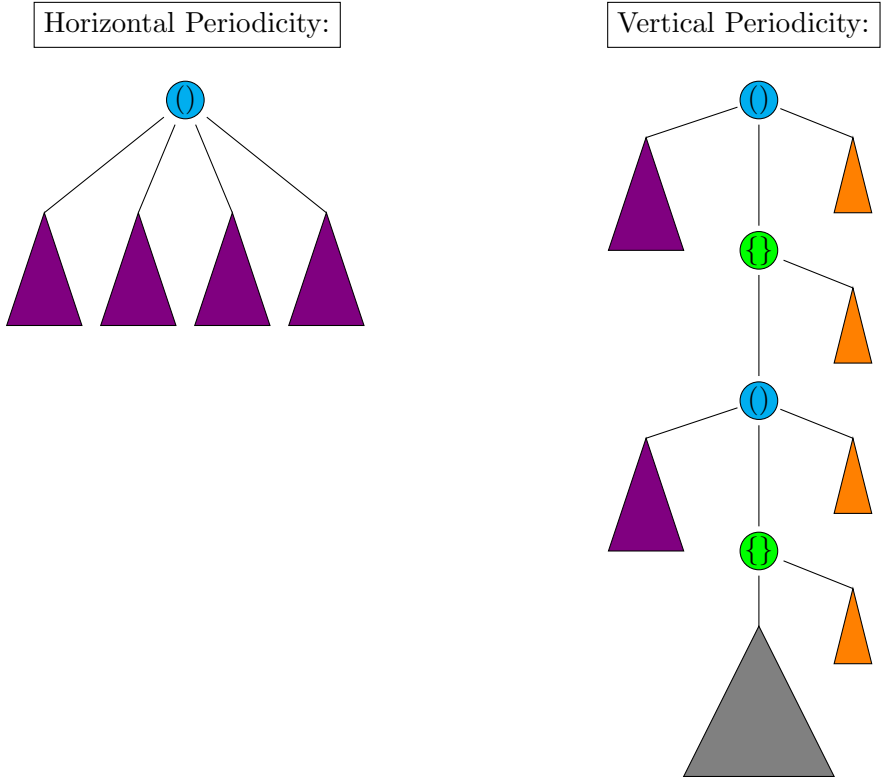

\begin{lemma}\label{lem:horizontal_aperiodic_reduction}
    Let $k\in \Zp$ and let $P,P'$ be forests of length $|P|,|P'|\ge 74k^3$ avoiding horizontal $k$-periodicity.
    Then, $P$ and $P'$ are $\ted_{\le k}^w$-equivalent for every normalized quasimetric $w$.
    \end{lemma}
    \begin{proof}
        Suppose that $P$ occurs in forests $F$ and $G$ at positions $p_F$ and $p_G$, respectively, satisfying $|p_F-p_G|\le 2k$.
        Denote $F' = F[0\dd p_F)\cdot P' \cdot F[p_F+|P|\dd |F|)$  and $G' = G[0\dd p_G)\cdot P' \cdot G[p_G + |P|\dd |G|)$.

        Let $\A=(f_t,g_t)_{t=0}^m$ be an alignment such that $\ted^w(F,G)= \ted^w_\A(F,G) \le k$. 
        Moreover, let $(f_a,g_a)\in \A$ be the leftmost element of $\A$ such that $f_a \ge p_F$ or $g_a \ge p_G$,
        and let $(f_b,g_b) \in \A$ be the leftmost element of $\A$ such that $f_b \ge p_F+|P|$ and $g_b \ge p_G+|P|$.
        We construct an alignment $\A'$ so that it:
        \begin{itemize}
            \item aligns $F[0\dd f_a)$ with $G[0\dd g_a)$ in the same way as $\A$ does;
            \item deletes $F[f_a\dd p_F)$ and inserts $G[g_a\dd p_G)$ (at least one of these fragments is empty);
            \item matches $F[p_F\dd p_F+|P|)=P$ with $G[p_G\dd p_G+|P|)=P$;
            \item deletes $F[p_F+|P|\dd f_b)$ and inserts $G[p_G+|P|\dd g_b)$  (at least one of these fragments is empty);
            \item aligns $F[f_b\dd |F|)$ with $G[g_b\dd |G|)$ in the same way as $\A$ does.
        \end{itemize}
        To prove that $\A'$ is a forest alignment, let us consider several possibilities for a node $u$ in $F$.
        \begin{itemize}
            \item If $u$ is inside $P=F[p_F\dd p_F+|P|)$, then $\A'$ matches $u$ to the corresponding node inside $P= G[p_G\dd p_G+|P|)$.
            \item If $u$ is outside $P=F[p_F\dd p_F+|P|)$ and $\A$ aligns $u$ to a node $v$ of $G$ inside $P= G[p_G\dd p_G+|P|)$, then $o(u),c(u)\in [f_a\dd p_F)\cup [p_F+|P|\dd f_b)$ because of the non-crossing property of $\A \ni (f_a,g_a),(f_b,g_b)$. Hence, $\A'$ deletes $u$.
            \item If $u$ is outside $P=F[p_F\dd p_F+|P|)$ and $\A$ aligns $u$ to a node $v$ of $G$ outside $P= G[p_G\dd p_G+|P|)$, then $o(u),c(u)\in [0\dd f_a)\cup [f_b\dd |F|)$ because of the non-crossing property of $\A \ni (f_a,g_a),(f_b,g_b)$. Hence, $\A'$ also aligns $u$ to $v$.
            \item If $u$ is outside $P=F[p_F\dd p_F+|P|)$ and $\A$ deletes $u$, then $\A'$ also deletes $u$. 
        \end{itemize}
    
        Our next goal is to prove that $\ted_{\A'}^w(F,G)\le \ted_\A^w(F,G)$. This relies on the following claim.
        \begin{claim}
            There exists $t\in [a\dd b]$ such that $f_t-g_t = p_F - p_G$.
        \end{claim}
        \begin{proof}
            Let us partition $P=F[p_F\dd p_F+|P|)$ into individual characters representing deletions or substitutions of $\A$
            and maximal fragments that $\A$ matches perfectly (to fragments of $G$).
            By \cref{fct:str_decomp}, the number of such fragments is at most $2k+1$ and their total length is at least $|P|-2k$.
            Hence, one of these fragments, denoted $R=F[r_F\dd r_F+|R|)$, is of length at least $\frac{|P|-2k}{2k+1}\ge 24k^2$.
            Suppose that the fragment of $G$ matched perfectly to $R$ is $G[r_G\dd r_G+|R|)$.
            If $r_F-p_F = r_G-p_G$, the claim holds for $t$ such that $(f_t,g_t)=(r_F,r_G)$. 
            Otherwise, we note that $R$ has period $q:=|(r_F-p_F)-(r_G-p_G)|\in [1\dd 4k]$. 
            Let $q=R[0\dd q)$ and observe that $\frac{|R|}{q}\ge \frac{24k^2}{4k}\ge 4k+2$.
            Hence, $Q^{4k+2}$ is a substring of $P$; since $P$ avoids horizontal $k$-periodicity, we conclude that no cyclic rotation of $Q$ is balanced. 

            As $Q$ is a substring of a balanced string $P$, this means that the number of opening parentheses in $Q$
            does not match the number of closing parentheses in $Q$.
            By symmetry (up to reversal), we assume without loss of generality that $Q$ has more opening than closing parentheses.
            Thus, there exists a node $u$ in $F$ such that $o(u)\in [r_F\dd r_F+q)$
            yet $c(u)\ge r_F+|R|$. In particular, $c(u)-o(u) \ge |R|-q \ge 24k^2-4k > 8k$.
            Let $v,v'$ be the nodes in $G$ matched with $u$ by $\A$ and $\A'$, respectively.
            Note that $|o(v)-o(v')|\le 4k$ and $|c(v)-c(v')|\le 4k$.
            Due to $c(v')-o(v')=c(u)-o(u) > 8k$, we conclude that $v$ is ancestor of $v'$ or vice versa.
            In either case, we have $0 \ge (o(v')-o(v))\cdot (c(v')-c(v)) = ((o(u)-o(v))-(p_F-p_G))\cdot  ((c(u)-c(v))-(p_F-p_G))$.
            The value $(f_t-g_t)-(p_F-p_G)$ can change by at most one for subsequent indices $t$.
            The sign of this value is different when $(f_t,g_t)=(o(u),o(v))$ and $(f_t,g_t)=(c(u),c(v))$,
            so it must be equal to $0$ at some intermediate index $t$.
            \end{proof}
    
        The alignments $\A$  and $\A'$ only differ in how they align $F[f_a\dd f_t)$ with $G[g_a\dd g_t)$
        and $F[f_t\dd f_b)$ with $G[g_t\dd g_b)$, and, by \cref{fct:quasi}, $\A'$ provides an optimum alignment of these fragments. 
        Now, if $P=F[p_F\dd p_F+|P|)=G[p_G\dd p_G+|P|)$ is modified to $P'$,
        then $\A'$ can be trivially adapted without modifying its cost and hence $\ted^w(F',G')\le \ted^w_{\A'}(F,G)=\ted^w(F,G)$.
        The converse inequality follows by symmetry between $(F,G,P)$ and $(F',G',P')$.
    \end{proof}

    % \begin{corollary}\label{cor:horizontal_reduction}
    %     Let $k\in \Zp$. For every forest $P$, there exists a forest of length at most $74k^3$ that is $\ted_{\le k}^w$-equivalent to $P$
    %     for every normalized quasimetric $w$.
    % \end{corollary}
    % \begin{proof}
    %     We proceed by induction on $|P|$ with the trivial base case of $|P|\le 74k^3$.
    %     If $|P|\ge 74k^3$ and $P$ avoids horizontal $k$-periodicity, then \cref{lem:horizontal_aperiodic_reduction} implies that $P$ is equivalent to any
    %     forest of length exactly  $74k^3$ that avoids horizontal $k$-periodicity (for example, a path of $37k^3$ nodes with arbitrary node labels).
    %     Thus, suppose that $P$ contains an occurrence $P[i\dd j)=Q^{4k+1}$ of a forest $Q$ of length $q\in [1\dd 4k]$.
    %     By \cref{lem:horizontal_periodic_reduction}, $Q^{4k+1}$ is equivalent to $Q^{4k}$,
    %     and thus $P$ is equivalent to a forest $P':=P[0\dd i)\cdot P[i+q\dd |P|)$.
    %     By the inductive assumption, $P'$ is equivalent to some forest $P''$ of length at most $74k^3$, and, by transitivity of the considered equivalence,
    %     $P$ is also equivalent to~$P''$.
    % \end{proof}

\subsubsection{Contexts}

\begin{definition}
    For $k\in \Zz$ and a weight function $w$, contexts $P=\langle P_L;P_R\rangle$ and $P'=\langle P'_L;P'_R\rangle$ are called \emph{$\ted_{\le k}^w$-equivalent}
    if 
    \begin{multline*}\ted_{\le k}^w(F,G) = \ted_{\le k}^w(F[0 \dd o(u)) \cdot P'_L \cdot F[o(u)+|P_L| \dd c(u)-|P_R|] \cdot P'_R \cdot F(c(u) \dd |F|),\\ G[0 \dd o(v)) \cdot P'_L \cdot G[o(v)+|P_L| \dd c(u)-|P_R|] \cdot P'_R \cdot G(c(v) \dd |G|))\end{multline*}
    holds for all forests $F$ and $G$ in which $P$ occurs at nodes $u$ and $v$, respectively, satisfying $|o(u)-o(v)|\le 2k$ and  $|c(u)-c(v)|\le 2k$.
\end{definition}

\begin{lemma}\label{lem:vertical_periodic_reduction}
    Let $k\in \Zp$, let $Q$ be a context, and let $e,e'\in \mathbb{Z}_{\ge 6k}$.
    Then, $Q^e$ and $Q^{e'}$ are $\ted_{\le k}^w$-equivalent for every normalized weight function $w$.
\end{lemma}
\begin{proof}
    We assume without loss of generality that $Q$ is primitive. (If $Q=R^m$ for $m\in \mathbb{Z}_{\ge 2}$, then $Q^e=R^{me}$ and $Q^{e'}=R^{me'}$
    can be interpreted as powers of $R$ rather than powers of $Q$.)
    Let $Q=\langle Q_L;Q_R\rangle$ with $q_L = |Q_L|$ and $q_R = |Q_R|$.
    Suppose that $Q^e$ occurs in forests $F$ and $G$ at nodes $u$ and $v$, respectively,  satisfying $|o(u)-o(v)|\le 2k$ and  $|c(u)-c(v)|\le 2k$.
%      that is,
%     \begin{align*}
% F = F[0 \dd o(u)) \cdot Q_L^e \cdot F[o(u)+e\cdot q_L \dd c(u)-e\cdot q_R] \cdot Q_R^e \cdot F(c(u) \dd |F|), \\ 
% G = G[0 \dd o(v)) \cdot Q_L^e \cdot G[o(v)+e\cdot q_L \dd c(u)-e\cdot q_R] \cdot Q_R^e \cdot G(c(v) \dd |G|).
% \end{align*}
Denote
\begin{align*}
    F' = F[0 \dd o(u)) \cdot Q_L^{e'} \cdot F[o(u)+|Q_L^e| \dd c(u)-|Q_R^e|] \cdot Q_R^{e'} \cdot F(c(u) \dd |F|), \\ 
    G' = G[0 \dd o(v)) \cdot Q_L^{e'} \cdot G[o(v)+|Q_L^e| \dd c(u)-|Q_R^e|] \cdot Q_R^{e'} \cdot G(c(v) \dd |G|).
\end{align*}
    For $i\in [0\dd e)$, let $u_i$ be the node of $F$
    with $o(u_i)=o(u)+i\cdot q_L$ (and $c(u_i)=c(u)-i\cdot q_R$) and let $v_i$ be the node of $G$ with $o(v_i)=o(v)+i\cdot q_L$ (and $c(v_i)=c(v)-i\cdot q_R$).
    Moreover, let $\A$ be an optimal forest alignment such that $\ted(F, G) = \ted_{\A} (F, G) \le k$.
   
    \begin{claim}\label{clm:vertical_periodic_reduction}
        There exist $i_F,i_G\in [0\dd 5k]$ such that \begin{align*}
            F[o(u)+i_F\cdot q_L\dd o(u)+(i_F+1)\cdot q_L)&\simeq_{\A} G[o(v)+i_G\cdot q_L\dd o(v)+(i_G+1)\cdot q_L),\\
            F(c(u)-(i_F+1)\cdot q_R\dd c(u)-i_F\cdot q_R]&\simeq_{\A} G(c(v)-(i_F+1)\cdot q_R\dd c(v)-i_G\cdot q_R].\end{align*}
    \end{claim}
    \begin{proof}
        Let $(f_b,g_b)\in \A$ be the leftmost element of $\A$ such that $f_b \ge o(u)$ and $g_b \ge o(v)$.
        By symmetry between $F$ and $G$, we may assume without loss of generality that $f_b = o(u)$.
        Consider the $k+1$ disjoint occurrences of $C$ in $F$ at positions $(o(u)+i\cdot q_L,c(u)+i\cdot q_R)$ for $i\in [0\dd k]$.
        The alignment $\A$ (of unweighted cost at most $k$) must match one of these occurrences perfectly to a context within $G$.
        We pick the index $i_F \in [0\dd k]$ of one such perfectly matched occurrence and suppose that it occurs at a node $v'$ of $G$.

        In particular, \begin{align*}F[o(u_{i_F})\dd o(u_{i_F})+q_L)&\simeq_{\A} G[o(v')\dd o(v')+q_L),\\
            F(c(u_{i_F})-q_R\dd c(u_{i_F})]&\simeq_{\A} G(c(v')-q_R\dd c(v')].\end{align*}
        Since $(f_b,g_b)\in \A$, we must have $o(v')\ge g_b \ge o(v)$ by the non-crossing property of $\A$.
        At the same time, since the unweighted cost of $\A$ does not exceed $k$, 
        we have $o(v')\le o(u_{i_F}) + 2k \le o(u)+kq_L + 2k \le o(v)+kq_L + 4k \le o(v)+5kq_L$.
        Similarly, $c(v') \ge c(v)-5kq_R$, which also implies $c(v') \le c(v)$.

        Our next goal is to show that $v'=v_{i_G}$ for some $i_G\in [0\dd 5k]$.
        For a proof by contradiction, suppose that $o(v_i)<o(v')<o(v_{i+1})$ for some $i\in [0\dd 5k)$.
        Due to $c(v')> c(v)-5kq_R$, this also implies that $c(v_i)>c(v')>c(v_{i+1})$,
        i.e., that $v'$ is a node on the path between $v_i$ and $v_{i+1}$.
        Suppose that the length of this path is $\ell$ and the node $v'$ is at distance $\ell'$ from $v_i$.
        Hence,  $G[o(v_i)\dd o(v'))$ has $\ell'$ unmatched opening parentheses out of the $\ell$  unmatched opening parentheses in $Q_L$. 
        Moreover, $G[o(v_i)\dd o(v'))\cdot G[o(v')\dd o(v_{i+1}))=Q_L = G[o(v')\dd o(v_{i+1}))\cdot G[o(v_i)\dd o(v'))$,
        and thus there is a primitive string $Q_L$ such that $G[o(v')\dd o(v_{i+1}))$ and $G[o(v_i)\dd o(v'))$ are both powers of $Q_L$.
        The number of unmatched opening parentheses is $Q_L$ must be a common divisor of $\ell$ and $\ell'$,
        i.e., $Q_L$ can be expressed as a string power with exponent $\ell/\gcd(\ell,\ell')$.
        A symmetric argument shows that $Q_R$ can be expressed as a string power with exponent $\ell/\gcd(\ell,\ell')$.
        Overall, we conclude that $C$ can be expressed as a context power with exponent $\ell/\gcd(\ell,\ell')$,
        contradicting the primitivity of $C$.
        Hence, $v'=v_{i_G}$ for some $i_G\in [0\dd 5k]$ holds as claimed and, in particular,
        $o(v') = o(v)+i_G q_L$ and $c(v')=c(v)-i_Gq_R$.
    \end{proof}

    Now, if the occurrences of $Q^e$ at nodes $u,v$ are replaced with $Q^{e'}$ for $e'\ge e-1$,
    we can interpret this as replacing the occurrences of $Q$ at nodes $u_{i_F},v_{i_G}$
    with $Q^{1+e'-e}$. By \cref{clm:vertical_periodic_reduction}, $\A$ can be trivially adapted without modifying its cost,
    and hence $\ted^w(F',G')\le \ted^w_{\A}(F,G)=\ted^w(F,G)$.
    If $e'< e-1$, we repeat the above argument to decrement the exponent one step at a time, still concluding that $\ted^w(F',G')\le \ted^w(F,G)$.
    In either case, the converse inequality follows by symmetry between $(F,G,e)$ and $(F',G',e')$.
\end{proof}

We say that a context $P=\langle P_L; P_R\rangle$ avoids vertical $k$-periodicity if 
it cannot be expressed as $P = C \star Q^{6k+1}\star D$ for some contexts $C,Q,D$
satisfying $|Q| \in [1\dd 8k]$.

\begin{lemma}\label{lem:vertical_aperiodic_reduction}
    Let $k\in \Zp$, let $P=\langle P_L;P_R\rangle,P'=\langle P'_L;P'_R\rangle$ be contexts of length $|P_L|+|P_R|,|P'_L|+|P'_R|\ge 578k^4$ that avoid
    vertical $k$-periodicity and whose halves do not contain any balanced substring of length more than $74k^3$.
    Then, $P$ and $P'$ are $\ted_{\le k}^w$-equivalent for every normalized weight function $w$.
\end{lemma}
\begin{proof}
    Suppose that $P$ occurs in forests $F$ and $G$ at nodes $u$ and $v$, respectively, satisfying $|o(u)-o(v)|\le 2k$ and  $|c(u)-c(v)|\le 2k$,
    Denote
    \begin{align*}
        F' = F[0 \dd o(u)) \cdot P'_L \cdot F[o(u)+|P_L| \dd c(u)-|P_R|] \cdot P'_R \cdot F(c(u) \dd |F|), \\ 
        G' = G[0 \dd o(v)) \cdot P'_R \cdot G[o(v)+|P_L| \dd c(v)-|P_R|] \cdot P'_R \cdot G(c(v) \dd |G|).
    \end{align*}
    Let $\A=(f_t,g_t)_{t=0}^m$ be an optimal forest alignment such that $\ted(F, G) = \ted_{\A} (F, G) \le k$.
    Moreover, let $(f_a,g_a)\in \A$ be the leftmost element of $\A$ such that $f_a \ge o(u)$ or $g_a \ge o(v)$,
    $(f_b,g_b)\in \A$ be the leftmost element of $\A$ such that $f_b \ge o(u)+|P_L|$ and $g_b \ge o(v)+|P_L|$,
    $(f_c,g_c)\in \A$ be the leftmost element of $\A$ such that $f_c > c(u)-|P_R|$ or $g_c > c(v)-|P_R|$,
    and let $(f_d, g_d)$ be the leftmost element of $\A$ such that $f_d > c(u)$ and $g_d > c(v)$.
    We construct an alignment $\A'$ so that it:
    \begin{itemize}
        \item aligns $F[0\dd f_a)$ with $G[0\dd g_a)$ in the same way as $\A$ does;
        \item deletes $F[f_a\dd o(u))$ and inserts $G[g_a\dd o(v))$ (at least one of these fragments is empty);
        \item matches $F[o(u)\dd o(u)+|P_L|)=P_L$ with $G[o(v)\dd o(v)+|P_L|)=P_L$;
        \item if $b > c$, deletes $F[o(u)+|P_L|\dd c(u)-|P_R|]$ and inserts $G[o(v)+|P_L|\dd c(v)-|P_R|]$;
        \item if $b \le c$, deletes $F[o(u)+|P_L|\dd f_b)$ and inserts $G[o(v)+|P_L|\dd g_b)$  (at least one of these fragments is empty);
        \item if $b \le c$, aligns $F[f_b\dd f_c)$ with $G[g_b\dd g_c)$ in the same way as $\A$ does;
        \item if $b \le c$, deletes $F[f_c\dd c(u)-|P_R|]$ and inserts $G[g_c\dd c(v)-|P_R|]$  (at least one of these fragments is empty);
        \item matches $F(c(u)-|P_R|\dd c(u)]=P_R$ with $G(c(v)-|P_R|\dd c(v)]=P_R$;
        \item deletes $F(c(u)\dd f_d)$ and inserts $G(c(v)\dd g_d)$  (at least one of these fragments is empty);
        \item aligns $F[f_d\dd |F|)$ with $G[g_d\dd |G|)$ in the same way as $\A$ does.
    \end{itemize}
    To prove that $\A'$ is a forest alignment, let us consider several possibilities for a node $u'$ in $F$.
    \begin{itemize}
        \item If $u'$ belongs to $P=\langle F[o(u)\dd o(u)+|P_L|); F(c(u)-|P_R|\dd c(u)]\rangle$, then $\A'$ matches $u'$ to the corresponding node that belongs to $P= \langle G[o(v)\dd o(v)+|P_L|); G(c(v)-|P_R|\dd c(v)]\rangle$.
        \item If $u'$ is outside $F[o(u)\dd c(u)]$ and $\A$ aligns $u'$ to a node $v'$ of $G$ inside $G[o(v)\dd c(v)]$,
        then $o(u'),c(u') \in [f_a\dd o(u))\cup (c(u)\dd f_d)$ because of the non-crossing property of $\A \ni (f_a,g_a),(f_d,g_d)$. Hence, $\A'$ deletes $u'$.
        \item If $u'$ is outside $F[o(u)\dd c(u)]$ and $\A$ aligns $u'$ to a node $v'$ of $G$ outside $G[o(v)\dd c(v)]$, then $o(u'),c(u')\in [0\dd f_a)\cup [f_d\dd |F|)$ because of the non-crossing property of $\A \ni (f_a,g_a),\allowbreak(f_d,g_d)$. Hence, $\A'$ also aligns $u'$ to $v'$.
        \item If $u'$ is outside $F[o(u)\dd c(u)]$ and $\A$ deletes $u'$, then $\A'$ also deletes $u'$. 
        \item If $b > c$ and $u'$ is inside $F[o(u)+|P_L|\dd c(u)-|P_R|]$, then $\A'$ deletes $u'$.
        \item If $b\le c$, $u'$ is inside $F[o(u)+|P_L|\dd c(u)-|P_R|]$, and $\A$ aligns $u'$ to a node $v'$ of $G$ outside $G[o(v)+|P_L|\dd c(v)-|P_R|]$,
        then $o(u'),c(u') \in [o(u)+|P_L|\dd f_b)\cup [f_c\dd c(u)-|P_R|]$ because of the non-crossing property of $\A \ni (f_b,g_b),(f_c,g_c)$. Hence, $\A'$ deletes $u'$.
        \item If $b\le c$,  $u'$ is inside $F[o(u)+|P_L|\dd c(u)-|P_R|]$, and $\A$ aligns $u'$ to a node $v'$ of $G$ inside $G[o(v)+|P_L|\dd c(v)-|P_R|]$, then $o(u'),c(u')\in [f_b\dd f_c)$ because of the non-crossing property of $\A \ni (f_b,g_b),(f_c,g_c)$. Hence, $\A'$ also aligns $u'$ to $v'$.
        \item If $b\le c$, $u'$ is inside $F[o(u)+|P_L|\dd c(u)-|P_R|]$, and $\A$ deletes $u'$, then $\A'$ also deletes $u'$. 
    \end{itemize}
    Let us now prove that $\ted^w_{\A'}(F,G)\le \ted^w_{\A}(F,G)$. This relies on the following claim.
    \begin{claim}\label{clm:vertical_aperiodic_reduction}
        There exist $t_L\in [a\dd b]$ such that $f_{t_L}-g_{t_L} = o(u)-o(u)$
        and $t_R\in [c\dd d]$ such that $f_{t_L}-g_{t_L} = c(u)-c(u)$.
    \end{claim}
    \begin{proof}
        By symmetry (up to reversal), we can focus without loss of generality on the first claim.
        Moreover, by symmetry between $F$ and $G$, we can assume without loss of generality that $f_a = o(u)$;
        in particular, this implies $f_a - g_a \ge o(u)-o(v)$.
        If there exists $t\in [a\dd b]$ such that $f_t - g_t \le o(u)-o(v)$,
        then, since $f_t-g_t$ may change by at most one for subsequent positions,
        there is also $t_L\in [a\dd b]$ such that $f_{t_L} - g_{t_L} = o(u)-o(v)$,
        Consequently, it remains to consider the case when $f_t - g_t > o(u)-o(v)$ holds for all $t\in [a\dd b]$.

        Let us express $P$ as a vertical composition of $e$ contexts $P=P_0\star \cdots \star P_{e-1}$, where $e$ is the depth of $P$.
        Observe that the occurrences of $P$ at node $u$ in $F$ and $v$ in $G$,
        for each $i\in [0\dd e)$, induce occurrences of $P_i$ at some nodes $u_i$ in $F$ and $v_i$ in $G$.
        Since $F(o(u_i)\dd o(u_i)+|P_{i,L}|)$ and $F(c(u_i)-|P_{i,R}|\dd c(u_i))$ are balanced,
        we conclude that $|P_i| \le 2\cdot (74k^3+1)\le 150k^3$.
        We can decompose $[0\dd e)$ into at most $k$ individual indices $i$ such that $\A$ does not match perfectly
        the occurrence of $P_i$ at $v_i$ and at most $k+1$ intervals $[i\dd i')$ such that $\A$ matches the occurrence $P_i\star \cdots \star P_{i'-1}$ at $v_i$ perfectly to a context in $F$.
        Let us choose such an interval $[i\dd i')$ maximizing $|P_i\star \cdots \star P_{i'-1}|$;
        this length is at least $\frac{578k^4-k\cdot 150k^3}{k+1}\ge 214k^3$.
        Let $i''\in [i\dd i')$ be the maximum index such that $|P_{i''}\star \cdots \star P_{i'-1}|> 8k$;
        note that $|P_i\star \cdots \star P_{i''-1}|\ge 214k^3 - (150k^3+8k) \ge 56k^2$.

        For each $j\in [i\dd i'']$, denote by $u'_j$ be the node of matched with $v_j$ by $\A$.
        Note that $|o(u'_j)-o(u_j)|\le 4k$ and $|c(u'_j)-c(u_j)|\le 4k$.
        Moreover, $o(u'_j) - o(v_j) > o(u)-o(v)=o(u_j)-o(v_j)$ implies $o(u'_j) > o(u_j)$.
        Since $|P_j\star \cdots \star P_{i'-1}| > 8k$, we conclude that $u'_j = u_{j'}$ for some $j' \in [j\dd i')$.
        Moreover, if $j>i$, then $u'_j$ must be a child of $u'_{j-1}$.
        Hence, there exists $\delta > 0$ such that $u'_{j}=u_{j+\delta}$ holds for all $j\in [i\dd i'']$.
        For $j\in [i\dd i'')$, this implies $P_j = P_{j+\delta}$ and that both halves of $P_j\star \cdots \star P_{j+\delta-1}$
        are of length at most $4k$. 
        In particular, if we define $Q=P_i\star \cdots \star P_{i+\delta-1}$, then,
        due to $\frac{|P_i\star \cdots \star P_{i''-1}|}{8k} \ge \frac{56k^2}{8k} \ge 6k+1$,
        we conclude that $Q^{6k+1}$  occurs in $F$ and $G$ at positions $u_i$ and $v_i$, respectively.
        This contradicts the assumption that $P$ avoids vertical periodicity.
    \end{proof}

    The alignments $\A$ and $\A'$ only differ in how they align the following fragments:
    \begin{itemize}
        \item $F[f_a\dd f_{t_L})$ with $G[g_a\dd g_{t_L})$: here, $\A'$ matches one fragment perfectly with a suffix of the other; by \cref{fct:quasi}, this is optimal.
        \item $F[f_{t_L}\dd f_{t_R})$ with $G[g_{t_L}\dd g_{t_R})$ if $b > c$: here, the cost of $\A'$ is equal to the cost of deleting $F[o(u)+|P_L|\dd c(u)-|P_R|]$ and inserting $G[o(v)+|P_L|\dd c(v)-|P_R|]$. By \cref{fct:quasi}, these two costs do not exceed the cost of $\A$
        aligning $F[f_{t_L}\dd f_c)$ with $G[g_{t_L}\dd g_c)$ and aligning $F[f_b\dd f_{t_R})$ with $G[g_b\dd g_{t_R})$.
        \item $F[f_{t_L}\dd f_b)$ with $G[g_{t_L}\dd g_b)$ if $b \le c$: here, $\A'$ matches one fragment perfectly with a prefix of the other; by \cref{fct:quasi}, this is optimal.
        \item $F[f_c\dd f_{t_R})$ with $G[g_c\dd f_{t_R})$ if $b \le c$: here, $\A'$ matches one fragment perfectly with a suffix of the other; by \cref{fct:quasi}, this is optimal.
        \item $F[f_{t_R}\dd f_d)$ with $G[g_{t_R}\dd g_d)$: here, $\A'$ matches one fragment perfectly with a prefix of the other; by \cref{fct:quasi}, this is optimal.
    \end{itemize}
    
    If the occurrences of $P$ at nodes $u$ in $F$ and $v$ in $G$ are modified to occurrences of $P'$,
    then $\A'$ can be trivially adapted without modifying its cost and hence $\ted^w(F',G')\le \ted^w_{\A'}(F,G)=\ted^w(F,G)$.
    The converse inequality follows by symmetry between $(F,G,P)$ and $(F',G',P')$.
\end{proof}

\subsection{Algorithms}\label{subsec:treealg}

\newcommand{\Pc}{\mathcal{P}}
We say a \emph{piece} of a forest $F$ is a balanced fragment $F[i\dd j)$ or a pair of fragments $\langle F[i\dd i');F[j'\dd j)\rangle$ that form a context,
that is, $F[i\dd j)$ is a tree and $F[i'\dd j')$ is balanced. For a fragment $F[i\dd j)$, we denote the set of pieces contained in $F[i\dd j)$ by $\Pc(F[i\dd j))$. Moreover, let $\Pc(F)=\Pc(F[0\dd |F|))$.

\newcommand{\D}{\mathsf{D}}
\begin{definition}\label{def:decomp}
A set $\D\sub \Pc(F[i\dd j))$ is a \emph{piece decomposition} of 
a balanced fragment $F[i\dd j)$ of a forest $F$ if it satisfies one of the following conditions:
\begin{itemize}
    \item $\D=\emptyset$ and $i=j$;
    \item $\D = \{F[i\dd j)\}$ and $i<j$;
    \item $\D = \D_L\cup \D_R$ for some piece decompositions $\D_L$ of $F[i\dd m)$ and $\D_R$ of $F[m\dd j)$,
    where $m\in (i\dd j)$.
    \item $\D = \{\langle F[i\dd i');F[j'\dd j)\rangle\}\cup \D'$ for a context $\langle F[i\dd i');F[j'\dd j)\rangle\in \Pc(F)$
    and a piece decomposition $\D'$ of $F[i'\dd j')$.
\end{itemize}
\end{definition}

\newcommand{\DC}{\mathcal{D}}
\begin{lemma}\label{lem:decomp}
There exists a linear-time algorithm that, given a forest $F$ and an integer $t\ge 2$,
constructs a piece decomposition $\D$ of $F$ consisting of at most $\max(1,\frac{6|F|}{t}-1)$ pieces of length at most $t$
each.
\end{lemma}
\begin{algorithm}
    \lIf{$j=i$}{\KwRet{$\emptyset$}}\label{ln:empty}
    \lIf{$j\le i+t$}{\KwRet{$\{F[i\dd j)\}$}}\label{ln:small}
    $i'\gets i; j'\gets j$\;
    \While(\tcp*[f]{$F[i'\dd j')$ is balanced and $(i'-i)+(j-j')\le t$}){\KwSty{true}}{   
        Let $m\in [i'\dd j']$ be such that $F[i'\dd m)$ is a tree\;
        \lIf{$(m-i)+(j-j')\le t$}{$i'\gets m$}\label{ln:l}
        \lElseIf{$m<j'$ \KwSty{and} $(i'-i)+(j-m)\le t$}{$j'\gets m$}\label{ln:r}
        \lElseIf{$F[i\dd j)$ is not a tree}{\KwRet{$\DC(i,i')\cup \DC(i',m)\cup \DC(m,j')\cup \DC(j',j)$}}\label{ln:forest}
        \lElseIf{$m=j'$ \KwSty{and} $(i'+1-i)\le (j-j'+1)\le t$}{ $i'\gets i'+1; j'\gets j'-1$}\label{ln:root}
        \lElse{\KwRet{$\{\langle F[i\dd i'); F[j'\dd j)\rangle\}\cup \DC(i',m)\cup \DC(m,j')$}}\label{ln:tree}
    }
\caption{$\DC(i,j)$: Construct a decomposition of a balanced fragment~$F[i\dd j)$.}\label{alg:decomp}
\end{algorithm}
\begin{proof}
\cref{alg:decomp} provides a recursive procedure that, for every balanced fragment $F[i\dd j)$ of $F$,
constructs a piece decomposition $\DC(i,j)$ of $F[i\dd j)$ consisting of pieces of size at most~$t$.
In the corner cases of $j=i$ and $j\in (i\dd i+t]$, we return $\DC(i,j)=\emptyset$ and $\DC(i,j)=\{F[i\dd j)\}$, respectively.
Otherwise, we iteratively grow fragments $F[i\dd i')$ and $F[j'\dd j)$ (initially empty) maintaining the following invariants:
\begin{enumerate}[label=(\alph*)]
    \item\label{inv:bal} $F[i'\dd j')$ is balanced;
    \item\label{inv:size} $|F[i\dd i')| + |F[j'\dd j)|\le t$;
    \item\label{inv:tree} if $F[i\dd j)$ is a tree, then $F[i\dd i')=F[j'\dd j)=\eps$
    or $\langle F[i\dd i'); F[j'\dd j)\rangle$ is a context;
    \item\label{inv:nottree} if $F[i\dd j)$ is not a tree, then $F[i\dd i')$ and $F[j'\dd j)$ are balanced.
\end{enumerate}
At each iteration, we identify a position $m\in (i'\dd j']$ such that $F[i'\dd m)$ is a tree
(such a position always exists due to $j-i>t$ and by invariants~\ref{inv:bal},~\ref{inv:size}).
\begin{enumerate}[label=({\arabic*})]
    \item\label{cs:l} We set $i'\gets m$ as long as it would not violate invariant~\ref{inv:size}.
    \item\label{cs:r} If $m\ne j'$, we set $j\gets m$ as long as it would not violate invariant~\ref{inv:size}.
    \item\label{cs:t} If $m=j'$ and $F[i\dd j)$ is a tree, we set $(i',j')\gets(i'+1,j'-1)$ as long as it would not violate invariant~\ref{inv:size}.
    \item\label{cs:f} Otherwise, we return $\DC(i,j):= \{\langle F[i\dd i'); F[j'\dd j)\rangle\}\cup \DC(i',m)\cup \DC(m,j')$
    if $F[i\dd j)$ is a tree and $\DC(i,j):=\DC(i,i')\cup \DC(i',m)\cup \DC(m,j')\cup \DC(j',j)$ if $F[i\dd j)$ is not a tree.
\end{enumerate}
It is easy to see that cases~\ref{cs:l}--\ref{cs:t} preserve the invariants and hence
case~\ref{cs:f} results in a valid piece decomposition with pieces of size at most $t$.

Next, we prove that the number of pieces is at most $\max(1,\frac{6(j-i)}{t}-3)$ if $F[i\dd j)$ is a tree and at most $\max(1,\frac{6(j-i)}{t}-1)$ otherwise. This holds trivially if $j\le i+t$, where \cref{alg:decomp} terminates at Line~\ref{ln:empty} or~\ref{ln:small}.
If $F[i\dd j)$ is a tree of size $j-i>t$, then \cref{alg:decomp} terminates at Line~\ref{ln:tree}.
We consider several sub-cases:
    \begin{enumerate}
        \item If $m-i' \le \frac{2t}{3}$ and $j'-m \le \frac{t}{3}$, then 
        $|\DC(i,j)|\le 3 < \frac{6(j-i)}{t}-3$ because $j-i>t$.
        \item If $m-i' \le \frac{2t}{3}$ and $j'-m > \frac{t}{3}$, then $(m-i)+(j-j')>t$ because the test in Line~\ref{ln:l} failed.
        Hence, $j'-m<j-i-t$ and $|\DC(i,j)|\le 2+|\DC(m,j')|\le 2 + \frac{6(j'-m)}{t} - 1 < \frac{6(j-i-t)}{t} + 1 < \frac{6(j-i)}{t}-3$.
        \item If $m-i' > \frac{2t}{3}$ and $j'-m=0$, then $(i'+1-i)+(j-j'+1)>t$ because the test in Line~\ref{ln:root} failed.
        Hence, $j'-i'\le j-i-t+1$ and $|\DC(i,j)|\le 1+|\DC(i',j')| \le 1 + \frac{6(j'-i')}{t}-3 \le \frac{6(j-i-t+1)}{t}-2 < \frac{6(j-i)}{t}-3$.
        \item If $m-i' > \frac{2t}{3}$ and $0 < j'-m \le \frac{t}{3}$, then $(i'-i)+(j-m)>t$ because the test in Line~\ref{ln:r} failed.
        Hence, $m-i'<j-i-t$ and $|\DC(i,j)|\le 2+|\DC(i',m)|\le 2 + \frac{6(m-i')}{t} - 3 < \frac{6(j-i-t)}{t} -1 < \frac{6(j-i)}{t}-3$.
        \item If $m-i' > \frac{2t}{3}$ and $j'-m > \frac{t}{3}$, 
        then $|\DC(i,j)|\le 1+|\DC(i',m)|+|\DC(m,j')| \le 1 + \frac{6(m-i')}{t}-3+\frac{6(j'-m)}{t}-1 \le \frac{6(j-i)}{t}-3$.
    \end{enumerate}
If $F[i\dd j)$ is not a tree, then  \cref{alg:decomp} terminates at Line~\ref{ln:forest}.
We consider several sub-cases:
    \begin{enumerate}
        \item If $m-i'\le \frac{2t}{3}$ and $j'-m \le \frac{t}{3}$, then 
        $|\DC(i,j)|\le 4 < \frac{6(j-i)}{t}-1$ because $j-i>t$.
        \item If $m-i'\le \frac{2t}{3}$ and $j'-m > \frac{t}{3}$, then $(m-i)+(j-j')>t$ because the test in Line~\ref{ln:l} failed.
        Hence, $j'-m<j-i-t$ and  $|\DC(i,j)|\le 3+|\DC(m,j')|\le 3 + \frac{6(j'-m)}{t} - 1 < \frac{6(j-i-t)}{t} + 2 < \frac{6(j-i)}{t}-1$.
        \item If $m-i'> \frac{2t}{3}$ and $j'-m=0$, then $|\DC(i,j)|\le 2+|\DC(i',j')| \le 2 + \frac{6(j'-i')}{t}-3 \le \frac{6(j-i)}{t}-1$.
        \item If $m-i'> \frac{2t}{3}$ and $0 < j'-m \le \frac{t}{3}$, then $(i'-i)+(j-m)>t$ because the test in Line~\ref{ln:r} failed.
        Hence, $m-i'<j-i-t$ and  $|\DC(i,j)|\le 3+|\DC(i',m)|\le 3 + \frac{6(m-i')}{t} - 3 < \frac{6(j-i-t)}{t} < \frac{6(j-i)}{t}-1$.
        \item If $m-i'> \frac{2t}{3}$ and $j'-m> \frac{t}{3}$, 
        then $|\DC(i,j)|\le2+|\DC(i',m)|+|\DC(m,j')| \le 2 + \frac{6(m-i')}{t}-3+\frac{6(j'-m)}{t}-1 \le \frac{6(j-i)}{t}-2 < \frac{6(j-i)}{t}-1$.
        \end{enumerate}

It remains to provide a linear-time implementation of our algorithm.
We assume that there are bidirectional pointers between the opening and the closing parentheses representing the same node.
Such pointers can be constructed using a linear-time stack-based preprocessing of the input forest $F$.
Each iteration of the \KwSty{while} loop increases $j-j'+i'-i$ (except for the final one), so a single call to the $\DC(i,j)$ function
costs $\Oh(t)$ due to invariant~\ref{inv:size}.
The total number of calls is $\Oh(|\DC(0,|F|)|)=\Oh(\frac1t\cdot |F|)$, so the overall running time, including preprocessing, is $\Oh(|F|)$.
\end{proof}

\begin{lemma}\label{lem:dp}
Given forests $F$ and $G$ of total size $n$, a piece decomposition $\D$ of $F$, and an integer $s\in \Zp$,
one can find in $\Oh(n+|\D|s^3)$ time a maximum-size set $S\sub \D \times \Pc(G)$ that, for some alignment $\A \in \ta(F,G)$ of width at most $s$,
contains only pairs of pieces that $\A$ matches perfectly.
\end{lemma}
\SetKwFunction{Pair}{Pairs}
\newcommand{\Ss}{\mathcal{S}}
\newcommand{\mgets}{\stackrel{\max}{\gets}}
\begin{algorithm}
    \Fn{$\Pair(\D_{i,j},G[i'\dd j'))$}{
    $S \gets \emptyset$\;
    \lIf{$i' < \min(j',i+s)$}{$S \mgets \Pair(D_{i,j}, G[i'+1\dd j'))$}
    \lIf{$j' > \max(i',j-s)$}{$S \mgets \Pair(D_{i,j}, G[i'\dd j'-1))$}
    \If{$\D_{i,j}=\{F[i\dd j)\}$ \KwSty{and} $F[i\dd j)=G[i'\dd j')$}{
        $S \mgets \{(F[i\dd j),G[i'\dd j'))\}$\;
    }
    \If{$\D_{i,j}=\D_{i,m}\cup \D_{m,j}$ for some $m\in (i\dd j)$}{
        \ForEach{$m'\in [m-s\dd m+s]\cap [i'\dd j']$}{
            $S\mgets \Pair(D_{i,m}, G[i'\dd m'))\cup \Pair(D_{m,j}, G[m'\dd j'))$\;
        }
    }
    \If{$\D_{i,j}=\{\langle F[i\dd i+\ell); F[j-r\dd j)\rangle\}\cup \D_{i+\ell,j-r}$}{
        \If{$F[i\dd i+\ell)=G[i'\dd i'+\ell)$ \KwSty{and} $F[j-r\dd j)=G[j'-r\dd j')$ \KwSty{and} $G[i'+\ell\dd j'-r)$ is balanced}{
            $S\mgets \{(\langle F[i\dd i+\ell); F[j-r\dd j)\rangle,\langle G[i'\dd i'+\ell); G[j'-r\dd j')\rangle)\}\cup \Pair(D_{i+\ell,j-r}, G[i'+\ell\dd j'-r))$
        }
        $S\mgets \Pair(D_{i+\ell,j-r}, G[\max(i+\ell-s,i')\dd \min(j-r+s,j')))$\;
    }
    \KwRet{$S$}\;
    }
    \caption{Compute a maximum-size element of $\Ss(D_{i,j}, G[i'\dd j'))$.}\label{alg:dp}
\end{algorithm}

\begin{proof}
    For a piece decomposition $\D_{i,j}$ of a balanced fragment $F[i\dd j)$ and a fragment $G[i'\dd j')$, let $\Ss(D_{i,j}, G[i'\dd j'))$ be the family of all subsets of $\D_{i,j}\times \Pc(G[i'\dd j'))$ that, for some alignment $\A \in \aa(F[i\dd j),G[i'\dd j'))$ of width at most $s$, contain only pairs of pieces that $\A$ matches perfectly.
    \Cref{alg:dp} implements a recursive procedure $\Pair(\D_{i,j},G[i'\dd j'))$ that computes a maximum-size element of $\Ss(D_{i,j}, G[i'\dd j'))$
    assuming that $i'\in [i-s\dd i+s]$ and $j'\in [j-s\dd j+s]$.  It uses an $S \mgets S'$ operator that assigns $S\gets S'$ if $|S'|>|S|$.
    The algorithm returns the largest of the following candidates:
    \begin{enumerate}
        \item\label{cs:empty} $\emptyset$. This is trivially valid because every alignment $\A\in \aa(F[i\dd j),G[i'\dd j'))$ of width at most $s$ witnesses $\emptyset \in \Ss(D_{i,j}, G[i'\dd j'))$.
        \item\label{cs:mona} $\Pair(D_{i,j}, G[i'+1\dd j'))$ if $i'< \min(j',i+s)$.
        Let $S = \Pair(D_{i,j}, G[i'+1\dd j'))$ with a witness alignment $\A'\in \aa(F[i\dd j),G[i'+1\dd j'))$.
        An alignment obtained from $\A'$ by prepending $(i,i')$, which corresponds to inserting $Y[i']$, witnesses $S\in \Ss(D_{i,j}, G[i'\dd j'))$.
        \item\label{cs:monb} $\Pair(D_{i,j}, G[i'\dd j'-1))$ if $j' > \max(i',j-s)$.
        Let $S = \Pair(D_{i,j}, G[i'\dd j'-1))$ with a witness alignment $\A'\in \aa(F[i\dd j),G[i'\dd j'-1))$.
        An alignment obtained from $\A'$ by appending $(j,j')$, which corresponds to inserting $Y[j'-1]$, witnesses $S\in \Ss(D_{i,j}, G[i'\dd j'))$.
        \item\label{cs:for} $\{(F[i\dd j),G[i'\dd j'))\}$ if $\D_{i,j}=\{F[i\dd j)\}$ and $F[i\dd j)=G[i'\dd j')$. 
        The alignment $(i+t,\allowbreak j+t)_{t=0}^{i'-i}\in \aa(F[i\dd j),G[i'\dd j'))$ witnesses $\{(F[i\dd j),G[i'\dd j'))\}\in \Ss(D_{i,j}, G[i'\dd j'))$.
        \item\label{cs:maxp} $\Pair(D_{i,m}, G[i'\dd m'))\cup \Pair(D_{m,j}, G[m'\dd j'))$ if $\D_{i,j}=\D_{i,m}\cup \D_{m,j}$ for some $m\in (i\dd j)$
        and $m'\in [m-s\dd m+s]\cap [i'\dd j']$. Denote $S_L=\Pair(D_{i,m}, G[i'\dd m'))$ and $S_R=\Pair(D_{m,j}, G[m'\dd j'))$
        with witness alignments $\A_L\in \aa(F[i\dd m),G[i'\dd m'))$ and $\A_R\in \aa(F[m\dd j),G[m'\dd j'))$, respectively.
        Stitching $\A_L$ and $\A_R$ at the common endpoint $(m,m')$ yields an alignment witnessing $S_L\cup S_R\in \Ss(D_{i,j}, G[i'\dd j'))$.
        \item\label{cs:cont} $\{(\langle F[i\dd i+\ell); F[j-r\dd j)\rangle,\langle G[i'\dd i'+\ell); G[j'-r\dd j')\rangle)\}\cup \Pair(D_{i+\ell,j-r}, G[i'+\ell\dd j'-r))$
        if $\D_{i,j}=\{\langle F[i\dd i+\ell); F[j-r\dd j)\rangle\}\cup \D_{i+\ell,j-r}$ and $\langle G[i'\dd i'+\ell); G[j'-r\dd j')\rangle$ is a context in $G$
        matching $\langle F[i\dd i+\ell); F[j-r\dd j)\rangle$. Consider a set $S'=\Pair(D_{i+\ell,j-r}, G[i'+\ell\dd j'-r))$ and a witness alignment $\A'\in \aa(F[i+\ell\dd j-r),G[i'+\ell\dd j'-r))$.
        Stitching $(i+t,i'+t)_{t=0}^{\ell}$, $\A'$, and $(j+t,j'+t)_{t=-r}^0$ at the common endpoints yields an alignment witnessing $S'\cup \{(\langle F[i\dd {i+\ell});\allowbreak F[j-r\dd j)\rangle,\langle G[i'\dd i'+\ell); G[j'-r\dd j')\rangle)\} \in \Ss(D_{i,j}, G[i'\dd j'))$.
        \item\label{cs:jump} $\Pair(D_{i+\ell,j-r}, G[\max(i+\ell-s,i')\dd \min(j-r+s,j')))$ if $\D_{i,j}=\{\langle F[i\dd i+\ell); F[j-r\dd j)\rangle\}\cup \D_{i+\ell,j-r}$. Let $S = \Pair(D_{i+\ell,j-r}, G[i'+\ell'\dd j'-r'))$,  where $\ell'=\max(0, i+\ell-s-i')$ and $r' = \max(0, j'-j+r-s)$, with a witness alignment $\A'\in \aa(F[{i+\ell}, j-r),G[i'+\ell',\allowbreak {j'-r'}))$.
        Stitching $(i+t,i'+t)_{t=0}^{\ell'}$, $(i+t,i'+\ell')_{t=\ell'}^{\ell}$, $\A'$, $(j+t,j'-r')_{t=-r}^{-r'}$, and $(j+t,j'+t)_{t=-r'}^0$ 
        yields an alignment witnessing $S\in \Ss(D_{i,j}, G[i'\dd j'))$.
    \end{enumerate}

    Next, consider a maximum-size element $S\in \Ss(D_{i,j}, G[i'\dd j'))$ and a witness alignment $\A\in \aa(F[i\dd j),\allowbreak G[i'\dd j'))$ of width at most $s$.
    \begin{enumerate}[label=(\alph*)]
        \item If $\D_{i,j}=\emptyset$, we must have $S=\emptyset$, which is covered by candidate~\ref{cs:empty}.
        \item Suppose that $\D_{i,j}=\{F[i\dd j)\}$. The case of $S=\emptyset$ is covered by candidate~\ref{cs:empty}.
        Otherwise, $S=\{(F[i\dd j),G[i''\dd j''))\}$ for some $i''\in [i'\dd i+s]$ and $j''\in [j-s\dd j']$.
        This is covered by $i''-i'$ applications of candidate~\ref{cs:mona}, $j'-j''$ applications of candidate~\ref{cs:monb}, and
        finally an application of candidate~\ref{cs:for}.
        \item Suppose that $\D_{i,j}=\D_{i,m}\cup \D_{m,j}$ for some $m\in (i\dd j)$.
    Since the width of $\A$ does not exceed $s$, we must have $(m,m')\in \A$ for some $m'\in [m-s\dd m+s]\cap [i'\dd j']$.
    Consequently, $S$ can be expressed as a union of an element of $\Ss(\D_{i,m},G[i'\dd m'))$
    and an element of $\Ss(D_{m,j}, G[m'\dd j'))$. This case is thus covered by candidate~\ref{cs:maxp}. 
    \item Suppose that $\D_{i,j}=\{\langle F[i\dd i+\ell); F[j-r\dd j)\rangle\}\cup \D_{i+\ell,j-r}$.
    If $S$ does not contain any pair of the form $(\langle F[i\dd i+\ell); F[j-r\dd j)\rangle,\langle G[i''\dd i''+\ell); G[j''-r\dd j'')\rangle)$,
    then $S \in \Ss(D_{i+\ell,j-r}, G[\max(i+\ell-s,i')\dd \min(j-r+s,j')))$, and this case is covered by candidate~\ref{cs:jump}.
    Otherwise, we must have $i''\in [\max(i',i-s)\dd i+s]$ and $j''\in [j-s\dd \min(j',j+s)]$.
    This is covered by $i''-i'$ applications of candidate~\ref{cs:mona}, $j'-j''$ applications of candidate~\ref{cs:monb}, and
    finally an application of candidate~\ref{cs:for} because $S\sm \{(\langle F[i\dd i+\ell); F[j-r\dd j)\rangle,\langle G[i'\dd i'+\ell); G[j'-r\dd j')\rangle)\} \in \Ss(D_{i+\ell,j-r}, G[i'+\ell\dd j'-r))$.
    \end{enumerate}
    This completes the proof that \cref{alg:dp} is correct.
    The sought set $S$ is obtained via a call $\Pair(\D,G[0\dd |G|))$ which is valid as long as $s\ge \big||F|-|G|\big|$.
    Otherwise, there is no alignment $\A \in \ta(F,G)$ of width at most $s$, and thus we return $S=\emptyset$.

    As for the efficient implementation, we use memoization to make sure that each call to $\Pair$ is executed at most once.
    The number of calls is $\Oh(|\D|\cdot s^2)$ and each one performs $\Oh(s)$ instructions.
    In order to implement every instruction in $\Oh(1)$ time, we implement sets as persistent linked lists augmented with their size (this is valid because the arguments of every union operation are guaranteed to be disjoint).
    Moreover,  we use \cref{thm:lce} (for checking whether fragments of $F$ match fragments of $G$)
    and \cref{fct:balanced} (for checking whether fragments of $G$ are balanced).
    Including the necessary preprocessing, the overall runtime is $\Oh(n+|\D|s^3)$.
\end{proof}

\begin{lemma}\label{lem:horizontal_reduction}
    There exists a linear-time algorithm that, given a forest $P$ and an integer $k\in \Zp$,
    constructs a forest of length at most $74k^3$ that is $\ted_{\le k}^w$-equivalent to $P$ for every normalized quasimetric~$w$.
\end{lemma}
\SetKwFunction{HorizontalReduction}{HorizontalReduction}
\begin{algorithm}
    \caption{Construct a forest of length at most $74k^3$ that is $\ted_{\le k}^w$-equivalent to $P$.}\label{alg:hor_per_reduc}
    \Fn{$\HorizontalReduction(P,k)$}{
        $P' \gets \PeriodicityReduction(P, 4k, \{Q\in \Sigma^+ : |Q|\le 4k \text{ and $Q$ is a primitive forest}\})$\;
        \lIf{$|P'| \geq 74k^3$}{\KwRet{$\op_a^{37k^3} \cl_a^{37k^3}$ for some $a\in \Sigma$}}
        \lElse{\KwRet{$P'$}}
    }
\end{algorithm}
\begin{proof}
    We apply~\cref{lem:perred} with $e=4k$ and $\Qf$ consisting of all primitive forests of length at most $4k$.
    We return $P'':=\op_a^{37k^3} \cl_a^{37k^3}$ (for an arbitrary label $a\in \Sigma$) or $P'$ depending 
    on whether $|P'|\le 74k^3$ or not.

    Observe that $\Qf$ is chosen so that $P'$ avoids horizontal $k$-periodicity and,
    by \cref{lem:horizontal_periodic_reduction}, $P'$ is $\ted_{\le k}^w$-equivalent to $P$ for every normalized quasimetric~$w$.
    Thus, the algorithm is correct if $|P'|< 74k^3$.
    Otherwise, \cref{lem:horizontal_aperiodic_reduction} implies that $P'$ and $P''$
    are $\ted_{\le k}^w$-equivalent for every normalized quasimetric~$w$ (both avoid horizontal $k$-periodicity
    and are of length at least $74k^3$).

    The oracle testing in constant time whether a given fragment of $P$ belongs to $\Qf$ can be implemented using \cref{fct:balanced}. Thus, by \cref{lem:perred}, the overall running time is linear.
\end{proof}

\newcommand{\Pp}{\mathbf{P}}
\newcommand{\Qq}{\mathbf{Q}}

\begin{lemma}\label{lem:vertical_reduction}
    There exists a linear-time algorithm that, given a context $P$ and an integer $k\in \Zp$,
    constructs a context of length at most $578k^4$ that is $\ted_{\le k}^w$-equivalent to $P$ for every normalized quasimetric~$w$.
\end{lemma}
\SetKwFunction{VerticalReduction}{VerticalReduction}
\begin{algorithm}
    \caption{Construct a context of length at most $578k^4$ that is $\ted_{\le k}^w$-equivalent to $P$.}\label{alg:family_per_reduc}
    \Fn{$\VerticalReduction(P,k)$}{
        Let $P=P_0\star \cdots \star P_{e-1}$, where each $P_i$ is a context of depth $1$\;

        \For{$i\gets 0$ \KwSty{to} $e$}{
            Let $P_i = \langle \op_{a_i} F_i; G_i \cl_{a_i}\rangle$\;
            $\Pp \gets \Pp\cdot \langle \op_{a_i} \cdot \HorizontalReduction(F_i,k); \HorizontalReduction(G_i,k) \cdot \cl_{a_i}\rangle$\;
        }
        $\Qf \gets \{\Qq : \bigstar_{i=0}^{|\Qq|-1}\Qq[i]\text{ is a primitive context of length at most }8k\}$\;
        $\Pp' \gets \PeriodicityReduction(\Pp, 6k, \Qf)$\;
        $P' \gets \bigstar_{i=0}^{|\Pp'|-1}\Pp'[i]$\;
        \lIf{$|P'| \geq 578k^4$}{\KwRet{$\bigstar_{i=0}^{17k^2-1} \langle \op_a (\op_a \cl_a)^{i};  (\op_a \cl_a)^{17k^2-1-i}\cl_a\rangle$ for some $a\in \Sigma$}}
        \lElse{\KwRet{$P'$}}
    }
\end{algorithm}
\begin{proof}
    Let $P=P_0\star \cdots \star P_{e-1}$, where each $P_i$ is a context of depth $1$,
    that is, $P_i = \langle \op_{a_i} F_i; G_i \cl_{a_i}\rangle$ for some label $a_i\in \Sigma$ and forests $F_i,G_i$.
    As the first step, our algorithm constructs a string $\Pp[0\dd e)$ whose characters are depth-1 contexts
    defined so that $\Pp[i]=\langle \op_{a_i} F'_i; G'_i \cl_{a_i}\rangle$, where forests $F'_i=\HorizontalReduction(F_i,k)$ and $G'_i=\HorizontalReduction(G_i,k)$ are constructed using \cref{lem:horizontal_reduction}.
    Next, we transform $\Pp$ using \cref{lem:perred} with $e=6k$ and a family $\Qf$ defined so that $\mathbf{P}[i\dd j)\in \Qf$
    if and only if $\Pp[i]\star \cdots \star \Pp[j-1]$ is a primitive context of length at most $8k$ (this implies $j-i\le 4k$).
    In order to apply \cref{lem:perred} to $\Pp$, we use linear-time string sorting~\cite{PT87,AN94} to map characters of $\Pp$ (depth-1 contexts) to integer identifiers. 
    By composing the contexts corresponding to the resulting string $\Pp'$, we obtain a context $P'$.
    We return $P'':=\bigstar_{i=0}^{17k^2-1} \langle \op_a (\op_a \cl_a)^{i};  (\op_a \cl_a)^{17k^2-1-i}\cl_a\rangle$ (for an arbitrary label $a\in \Sigma$)
    or $P'$ depending on whether $|P'|\ge 578k^4$ or not.

    Note that $|P''| = \sum_{i=0}^{17k^2-1} (1 + 2\cdot i + 2\cdot (17k^2-1-i)+1) =17k^2 \cdot 2\cdot  17k^2 = 578k^4$.
    Thus, the resulting context (either $P'$ or $P''$) is guaranteed to be of length at most $578k^4$.
    Let us now argue that it is $\ted_{\le k}^w$-equivalent to $P$ for every normalized quasimetric~$w$.
    By \cref{lem:horizontal_reduction}, the forests $F'_i$ and $G'_i$ are $\ted_{\le k}^w$-equivalent to $F_i$ and $G_i$,
    respectively, and thus $\bigstar_{i=0}^{e-1} \Pp[i]$ is $\ted_{\le k}^w$-equivalent to $P$.
    By \cref{lem:perred}, the context $P'$ is obtained from $\bigstar_{i=0}^{e-1} \Pp[i]$ by repeatedly replacing $Q^{6k+1}$ with $Q^{6k}$
    for primitive contexts  $Q$ of length at most $8k$.
    By \cref{lem:vertical_periodic_reduction}, $Q^{6k+1}$ is then  $\ted_{\le k}^w$-equivalent to $Q^{6k}$,
    so this operation preserves $\ted_{\le k}^w$-equivalence, i.e., $P'$ is also $\ted_{\le k}^w$-equivalent to $P$.
    Moreover, each depth-$1$ context in $\Pp'$ originates from $\Pp$, so each forest occurring in (either half of) $P'$ is of length at most $74k^3$.
    Furthermore, \cref{lem:perred} guarantees that $P'$ is not of the form $C\star Q^{6k+1} \star D$ for any context $Q$ of length at most $8k$,
    and thus $P'$ avoids vertical $k$-periodicity. 
    By construction, $P''$ avoids vertical $k$-periodicity and its halves  contain only forests of lengths at most $74k^3$ (in fact, at most $34k^2$).
    Consequently, \cref{lem:vertical_aperiodic_reduction} implies that $P''$ is $\ted_{\le k}^w$-equivalent to $P'$ (and, by transitivity, to $P$)
    provided that  $|P''|\ge 578k^4$.

    As for the running time analysis, we note that all applications of \cref{lem:horizontal_reduction}
    concern disjoint fragments of $P$, so the total cost of the calls to $\HorizontalReduction$ is linear.
    Assigning integer identifiers to contexts $\Pp[i]$ and applying \cref{lem:perred} also takes linear time.
    Finally, $P''$ is constructed only if $|P'|\ge 578k^4$, so the cost of this step is also be bounded by $\Oh(k^4)=\Oh(|P|)$.
\end{proof}

\begin{theorem}\label{thm:forestKernel}
    There exists an $\Oh(n)$-time algorithm that, given forests $F$, $G$ of size at most $n\ge 12716k^5$ and an integer $k\in \Zp$,
    constructs forests $F'$, $G'$ of lengths at most $\frac{n}{2}+6358k^5$ such that $\ted^w_{\le k}(F,G)=\ted^w_{\le k}(F',G')$ holds for every normalized quasimetric $w$.
\end{theorem}
% \SetKwFunction{Decompose}{Decompose}
% \SetKwFunction{Alignment}{Alignment}
% \begin{algorithm}[ht]
% fun($\D_{i,j}$){
%     \If{$\D_{i,j}=\emptyset$}{
%         \KwRet{$(\eps, \eps)$}\;
%     }\ElseIf{$\D_{i,j}=\{F[i\dd j)\}$}{
%         $G[i'\dd j') \gets \A(F[i\dd j))$\;
%         \If{$G[i'\dd j')=F[i\dd j)$}{
%             $P \gets \HorizontalReduction(F[i\dd j),k)$\;
%             \KwRet{$(P,P)$}
%         }\Else{
%             \KwRet{$(F[i\dd j),G[i'\dd j'))$}
%         }
%     }
%     \ElseIf{$\D_{i,j}=\D_{i,m}\cup \D_{m,j}$}{
%         \KwRet{$fun(\D_{i,m})\cdot \fun(\D_{m,j})$}\;
%     }\Else{
%         $\D_{i,j}=\{\langle F[i\dd i+\ell)\; F[j-r\dd j)\rangle \}\cup \D_{i+\ell,j-r}$\;
%         $G[i'\dd i'+\ell') \gets \A(F[i\dd i+\ell))$\;
%         $G[j'-r'\dd j') \gets \A(F[j-r\dd j))$\;
%         \If{$G[i'\dd i'+\ell')=F[i\dd i+\ell)$ \KwSty{and} $G[j'-r'\dd j')=F[j-r\dd j)$}{
%             $P \gets \VerticalReduction(\langle F[i\dd i+\ell)\; G[j-r\dd j)\rangle,k)$\;
%             \KwRet{$P\star fun(\D_{i+\ell,j-r})$}
%         }\Else{
%             \KwRet{$(\langle F[i\dd i+\ell)\; F[j-r\dd j)\rangle,\langle G[i'\dd i'+\ell')\; G[j'-r'\dd j')\rangle)\star fun(\D_{i+\ell,j-r})$}
%         }
%     }
% }
% \end{algorithm}
\begin{proof}
    By symmetry, we assume without loss of generality that $|F|\ge |G|$.
    We start by applying \cref{lem:decomp} to construct a piece decomposition $\D$ of $F$ consisting of at most $12k-1$ pieces of length at most $\lceil \frac{n}{2k}\rceil$ each.
    Next, we use \cref{lem:dp} to identify a maximum-size set $S\sub \D\times \Pc(G)$ that, for some alignment $\A\in \aa(F,G)$ of width at most $2k$,
    contains only pairs of pieces that $\A$ matches perfectly.
    If $|S|<|\D|-k$, we return $F'=(\op_a\cl_a)^{k+1}$ and $G'=\eps$ for some $a\in \Sigma$.
    Otherwise, for each pair of matching forests $F[i\dd j)=P=G[i'\dd j')$ in $S$,
    we use \cref{lem:horizontal_reduction} to construct a forest $P'$ of length at most $74k^3$  that is $\ted^w_{\le k}$-equivalent to $P$ for every normalized quasimetric $w$. We replace the occurrences of $P$ at $F[i\dd j)$ and $G[i'\dd j')$ by occurrences of $P'$.
    Similarly, for every pair of matching contexts $\langle F[i\dd i+\ell); F[j-r\dd j)\rangle = P = \langle G[i'\dd i'+\ell); G[j'-r\dd j')\rangle$ in $S$,
        we use \cref{lem:vertical_reduction} to construct a context $P'$ of length at most $578k^4$  that is $\ted^w_{\le k}$-equivalent to $P$ for every normalized quasimetric $w$. We replace the occurrences of $P$ at $\langle F[i\dd i+\ell); F[j-r\dd j)\rangle$ and $\langle G[i'\dd i'+\ell); G[j'-r\dd j')\rangle$ by occurrences of $P'$.

    If $\ted^w(F,G) \le k$ holds for any normalized weight function $w$,
    then the unweighted cost of the underlying optimal alignment does not exceed $k$.
    Thus, its width is at most $2k$, and it matches perfectly all but at most $k$ pieces of $\D$.
    Consequently, in that case, $|S|\ge |\D|-k$.    In particular, if we return $F'=(\op_a\cl_a)^{k+1}$ and $G'=\eps$ for some $a\in \Sigma$, then $\ted^w_{\le k}(F,G)=\infty=\ted^w_{\le k}(F',G')$
    holds as claimed. In that case, $|F'|,|G'| \le 2k+2 \le 4k < 6358k^5$.
    Otherwise, by definition of $\ted^w_{\le k}$-equivalence, the resulting forests $F'$ and $G'$ satisfy $\ted^w_{\le k}(F,G)=\ted^w_{\le k}(F',G')$.
    Moreover, $|F'| \le k \cdot \lceil \frac{n}{2k} \rceil + (11k-1)\cdot 578k^4 \le \frac{n}2 + k + (11k-1)\cdot 578k^4
    \le \frac{n}2+11\cdot 578\cdot k^5 = \frac{n}2+6358k^5$ and $|G'|=|F'|+|G|-|F|\le |F'|\le \frac{n}2+6358k^5$.
   
    The applications of \cref{lem:decomp,lem:dp} cost $\Oh(n)$ and $\Oh(n+k^4)=\Oh(n)$ time, respectively.
    The calls to $\HorizontalReduction$ and $\VerticalReduction$ concern disjoint pieces of $F$, so their total
    cost is $\Oh(n)$ by \cref{lem:horizontal_reduction,lem:vertical_reduction}, respectively.
\end{proof}

\begin{corollary}\label{cor:forestKernel}
    There exists a linear-time algorithm that, given forests $F$, $G$ and an integer $k\in \Zp$,
    constructs forests $F'$, $G'$ of lengths at most $12717k^5$ such that $\ted^w_{\le k}(F,G)=\ted^w_{\le k}(F',G')$ holds for every normalized quasimetric $w$.
\end{corollary}
\begin{proof}
    We iteratively apply \cref{thm:forestKernel} as long as $\max(|F|,|G|)>12717k^5$ and return the resulting pair of forests.
    Formally, we construct a sequence $(F_i,G_i)_{i=0}^t$ such that $(F_0,G_0)=(F,G)$ and $\ted^w_{\le k}(F_{i+1},G)_{i+1}=\ted^w_{\le k}(F_i,G_i)$
    holds for every $i\in [0\dd t)$.
    Consider forests $(F_i,G_i)$ at iteration $i$. If $n_i := \max(|F_i|,|G_i|) \le 12717k^5$, we set $t:=i$ and return $(F',G'):=(F_i,G_i)$. 
    Otherwise, we apply \cref{thm:forestKernel} to derive forests $F_{i+1}$ and $G_{i+1}$ 
    of lengths at most $n_{i+1} := \max(|F_{i+1}|,|G_{i+1}|)\le \frac12 n_i+6358k^5$ such that $\ted^w_{\le k}(F_{i+1},G)_{i+1}=\ted^w_{\le k}(F_i,G_i)$.
    Since $n_{i+1}-12716k^5 \le \frac12(n_i-12716k^5)$, the value $n_i$ strictly decreases at each iteration and thus the process terminates.
    Moreover, the running time of each iteration is $\Oh(n_i) = \Oh(n_{i}-12716k^5)$. The latter values form a geometric series
    dominated by the leading term at $i=0$. Hence, the total running time is linear in the input size.
\end{proof}

%\begin{theorem}
%Given forests $F,G$ of length at most $n$, an integer $k\in \Zp$, and a normalized quasimetric $w$,
%the value $\ted_{\le k}^w(F,G)$ can be computed in $\Oh(n+k^{15})$ time.
%Moreover,  $\ted_{\le k}(F,G)$ can be computed in $\Oh(n+k^{7}\log k)$ time
%\end{theorem}
\weightedtree*
\begin{proof}
    We first apply \cref{cor:forestKernel} to build forests $F',G'$ of length $\Oh(k^5)$ such that 
    $\ted_{\le k}^w(F,G)=\ted_{\le k}^w(F',G')$. Then, we compute $\ted_{\le k}^w(F',G')$ using 
    the algorithm of Demaine, Mozes, Rossman, and Weimann~\cite{10.1145/1644015.1644017}.
    The running times of these two steps are $\Oh(n)$ and $\Oh((k^5)^3)$, respectively,
    for a total of $\Oh(n+k^{15})$.
    If $w$ is the discrete metric (the unweighted case), then we compute $\ted_{\le k}(F',G')$ using 
    the algorithm of Akmal and Jin~\cite{DBLP:conf/icalp/AkmalJ21}, which costs $\Oh(k^5\cdot k^2 \cdot \log(k^5))=\Oh(k^{7}\log k)$ time.
\end{proof}

\section{Dyck Edit Distance}\label{sec:dyck}

In this section we give a deterministic algorithm that computes weighted Dyck edit distance of a given input string. Formally we show the following.

\weightedDyck*
%\begin{theorem}
%\label{thm:maindyck}
%  Given a string $X$ of length $n$, an integer $k\in \Zp$ and a normalized weight function $w$, the value $\dyck_{\le k}^w(X)$ can be computed in $O(n+k^{12})$ time.
% \end{theorem}

\subsection{Preliminaries}
In Dyck Language, the alphabet $\Sigma$ consists of two disjoint sets $T$ and $\rev{T}$ of \emph{opening} and \emph{closing} parentheses, respectively, with a bijection $f: T\to \rev{T}$ mapping each opening parenthesis to the corresponding closing parenthesis. We extend this mapping to an involution $f:T\cup \rev{T}\to T\cup\rev{T}$ and
then to an involution $f: \Sigma^*\to \Sigma^*$ mapping each string 
$X[0] X[1]\cdots X[|X|-1]$ to its reverse complement $\rev{X[|X|-1]}\cdots \rev{X[1]}\,\rev{X[0]}$. 
Given two strings $X,Y$, we denote their concatenation by $XY$ or $X\cdot Y$.

The \emph{Dyck} language $\DYCK(\Sigma)\sub \Sigma^*$ consists of all well-parenthesized expression over $\Sigma$;
formally, it can be defined using a context-free grammar whose only non-terminal $S$ admits productions $S\rightarrow SS$, $S\rightarrow \varnothing$ (empty string), and $S\rightarrow aS\overline{a}$ for all $a\in T$.

\begin{definition}[Heights]
  Given an alphabet set $\Sigma$, define the function $h :\Sigma \rightarrow \{-1, 1\}$ where $h(a)=1$ if $a\in \Sigma$ is an opening
  parenthesis and $h(a)= -1$ otherwise. Given a string $X\in \Sigma^n$, define the height of a position $i$ where $0 \le i \le n$,
  as $H(i)=\sum_{j=0}^{i-1}h(X[j])$.
  \end{definition}
  
  Here $H(i)$ is the difference between the number of opening parentheses and the number of closing parentheses in $X[0 . . i)$.
  
  \begin{definition}[Peaks and valleys]
  Given a string $X\in \Sigma$, an index $i\in [1\dots n)$ is called a peak if $H(i-1) < H(i) > H(i+1)$ and a valley if $H(i-1) > H(i) < H(i+1)$.
  \end{definition}

\newcommand{\Z}{\mathbb{Z}}
\subsection{Dyck Language Alignments and Weighted Dyck Edit Distance}
\label{sec:dyckweight}
We say that $\M\subseteq \{(i,j) \sub \Z^2 : i < j\}$ is a \emph{non-crossing matching} if any two distinct pairs $(i,j),(i',j')\in \M$ satisfy $i<j < i' < j'$ or $i < i' < j' < j$. Such a matching can also be interpreted as a function $\M : \Z \to \Z\cup\{\bot\}$ with $\M(i)=j$ if $(i,j)\in \M$ or $(j,i)\in \M$ for some $j\in \Z$, and $\M(i)=\bot$ otherwise.
For a string $X\in \Sigma^*$ we define its \emph{Dyck language alignment} to be a matching function $\M$ as defined above.

For two fragments $X[p\dd q)$ and $X[p'\dd q')$ of $X$, we write $X[p\dd q)\simeq_{\M} \rev{X[p'\dd q')}$
if $X[p\dd q)=\rev{X[p'\dd q')}\in T^*$ and $(r, q'-r-1)$ holds for every $r\in [p\dd q)$.

% \iffalse
% Given an alignment $\M \in \mm(X)$, for every $i\in [0\dd |X|)$, we say that
% \begin{itemize}
%   \item  $\M$ \emph{deletes} $X[i]$ if $\M(i)=\bot$.
%   %\item  $\M$ \emph{inserts} $\rev{X[i]}$ right after $X[i]$ if $\M(i)=+$ and $X[i]\in T$. 
%   %\item $\M$ \emph{inserts} $\rev{X[i]}$ right before $X[i]$ if $\M(i)=+$ and $X[i]\in \rev{T}$.
%   \item $\M$ \emph{aligns} $X[i]$ to $X[\M(i)]$, denoted $X[i] \sim X[\M(i)]$, if $\M(i)\neq \bot$. 
%   \item $\M$ \emph{matches} $X[i]$ with $X[\M(i)]$, denoted $X[i] \simeq X[\M(i)]$, if $\M$ aligns $X[i]$ to $X[\M(i)]$ and $X[i]=\rev{X[\M(i)]}$.
%   \item $\M$  \emph{substitutes} $X[i]$ with $\rev{X[\M(i)]}$ if $\M$ aligns $X[i]$ to $X[\M(i)]$ but $X[i]\ne \rev{X[\M(i)]}$.
% \end{itemize}
% Deletions, and substitutions are jointly called (character) \emph{edits}.
% \fi

% \iffalse
% \noindent
% \textbf{Weight Function.}
% For an alphabet $\Sigma$, we define $\eSigma := \Sigma\cup\{\eps\}$, where $\eps$ is the empty string over $\Sigma$.
% We say that a function $w : \eSigma\times \eSigma \to \mathbb{R}_{\ge 0}$ is a \emph{weight function} if 

% \begin{itemize}
%     \item $w(a,a)=0$ holds for all $a\in \eSigma$. %\tknote{This is not true. $w(a,a)=0$.}
%     \item $w(a,\eps)$ is the cost of deleting symbol $a$ for all $a\in \Sigma$.
%     \item $w(\eps,a)$ is the cost of inserting symbol $a$ for all $a\in \Sigma$.
%     \item $w(a,b)$ is the cost of substituting symbol $a$ by $b$ for all $a,b\in \Sigma$.
% \end{itemize}
% \tknote{A weight function is already defined in the Strings part. Is there any reason to redefine it here?}
% \fi

Similar to Section~\ref{sec:editdistance} we define a weight function $w$ on $\eSigma := \Sigma\cup\{\eps\}$.
We call this weight function a \emph{skewmetric} if it satisfies the triangle inequality, that is, $w(a,b)+w(b,c)\ge w(a,c)$ holds for every $a,b,c\in \eSigma$ and skew-symmetry, that is, $w(a,b)=w(\rev{b},\rev{a})$ holds for every $a,b\in \eSigma$. In the rest of this section we assume the weight function $w$ to be skewmetric unless stated otherwise.

\begin{definition}
The weighted Dyck edit distance of a string $X\in \Sigma^*$ with respect to a weight function $w$ is the minimum edit distance $\ed^w(X,Y)$ between $X$ and a string $Y\in \DYCK(\Sigma)$. Formally, \[\dyck^w(X)=\min_{Y\in \DYCK(\Sigma)}\ed^w(X,Y).\]
\end{definition}
For $k\in \mathbb{R}_{\ge 0}$, we also denote
\[\dyck^w_{\le k}(X) = \begin{cases}
\dyck^w(X) & \text{if }\dyck^w(X)\le k,\\
\infty & \text{otherwise.}
\end{cases}\]

The \emph{cost} of an alignment $\M\in \mm(X)$ with respect to a \emph{weight function} $w$, denoted $\dyck_\M^w(X)$, is defined as 
\[\dyck_\M^w(X)=\sum_{(i,j)\in \M}\dyck^w(X[i]X[j])+\sum_{i\in [0\dd |X|) : \M(i)=\bot}\dyck^w(X[i]).\]

\begin{fact}
\label{fact:sim}
For every string $X$ and weight function $w$, we have $\dyck^w(X)=\min_{\M \in \mm(X)} \dyck^w_{\M}(X)$.
\end{fact}
\begin{proof}
We first show by induction on $|X|$ that $\dyck^w(X)\le \dyck^w_\M(X)$ holds for every $\M\in \mm(X)$.
	The claim is trivial if $|X|=0$.
	If $\M(0)=\bot$, then we construct $\M':=\{(i-1,j-1) : (i,j)\in \M\}$ and $X':=X[1\dd |X|)$.
	By the inductive assumption, $\dyck^w(X)\le \dyck^w(X')+\dyck^w(X[0])\le \dyck^w_{\M'}(X')+\dyck^w(X[0])= \dyck^w_\M(X)$.
	If $\M(0)=|X|-1$, then we construct $\M':=\{(i-1,j-1) : (i,j)\in \M\sm \{(0,|X|-1)\}$ and $X'=[1\dd |X|-1)$.
	By the inductive assumption, $\dyck^w(X)\le \dyck^w(X')+\dyck^w(X[0]X[|X|-1])\le \dyck^w_{\M'}(X') + \dyck^w(X[0]X[|X|-1]) = \dyck^w_\M(X)$.
	Otherwise, we have $(0,p)\in M$ for some $p\in [1\dd |X|-1)$.
	In this case, we construct $\M':=\{(i,j)\in \M : j\le p\}$ and $X' := X[0\dd p]$,
	as well as $\M'' := \{(i-p-1,j-p-1) : (i,j)\in \M\text{ and }i>p\}$ and $X'':=X[p+1\dd |X|)$.
	By the inductive assumption, $\dyck^w(X)\le \dyck^w(X')+\dyck^w(X'')\le \dyck^w_{\M'}(X') + \dyck^w_{\M''}(X'') = \dyck^w_\M(X)$;
	here, the last equality follows from the fact that $|\M|=|\M'|+|\M''|$: any $(i,j)\in M$ with $i \le p$ and $j>p$ would violate the non-crossing property of $\M$.
	
%    \tknote{The updated proof below now does not require the triangle inequality.}

	Next, we show by induction on $|X|$ that there exists $\M\in \mm(X)$ such that $\dyck^w_\M(X)\le \dyck^w(X)$; again, the claim is trivial for $|X|=0$.
    Let us fix $Y\in \DYCK(\Sigma)$ and $\A \in \aa(X,Y)$  such that $\dyck^w(X)=\ed^w_{\A}(X,Y)$.
    If $\A$ deletes $X[0]$, we consider $X':=X[1\dd |X|)$. The inductive assumption yields a matching $\M'$
    such that $\dyck^w_{\M'}(X') \le \dyck^w(X')$.
    In this case, we set $\M := \{(i+1,j+1) : (i,j)\in \M'\}$ so that
    $\dyck^w_{\M}(X) = \dyck^w(X[0])+\dyck^w_{\M'}(X')\le \ed^w(X[0],\eps)+\dyck^w(X')\le w(X[0],\eps) + \ed^w(X',Y') \le \ed^w_{\A}(X,Y)= \dyck^w(X)$.
    The case when $\A$ deletes $X[|X|-1]$ is symmetric, so we may assume that $\A$ deletes neither $X[0]$ nor $X[|X|-1]$; in particular, $Y\ne \eps$.

    Suppose that $Y=Y'\cdot Y''$ for some non-empty strings $Y',Y''\in \DYCK(\Sigma)$.
    This yields a decomposition $X=X'\cdot X''$ such that $\ed^w(X,Y)=\ed^w(X',Y')+\ed^w(X'',Y'')$. Moreover, the optimality of $Y$ guarantees that $X'$ and $X''$ are both non-empty.
    The inductive assumption yields matchings $\M',\M''$ such that $\dyck^w_{\M'}(X') \le \dyck^w(X')$
    and $\dyck^w_{\M'}(X'') \le \dyck^w(X'')$.
    In this case, we set $\M := \M' \cup \{(i+|X'|,j+|X'|) : (i,j)\in \M''\}$ so that 
    $\dyck^w_{\M}(X)=\dyck^w_{\M'}(X') + \dyck^w_{\M'}(X'') \le \dyck^w(X')+\dyck^w(X'') \le \ed^w(X',Y')+\ed^w(X'',Y'') \le \ed^w_\A(X,Y)=\dyck^w(X)$.

    In the remaining case, we must have $Y=aY'\rev{a}$ for $a\in T$ and $Y'\in \DYCK(\Sigma)$.   
    Let us first suppose that $\A$ aligns $Y[0]=a$ with $X[0]$ and $Y[|Y|-1]=\rev{a}$ with $X[|X|-1]$.
    In this case, $\A$ aligns $Y[1\dd |Y|-1)=Y'$ with $X':=X[1\dd |X|-1)$.
    The inductive assumption yields a matching $\M'$ such that $\dyck^w_{\M'}(X') \le \dyck^w(X')$.
    In this case, we set $\M := \{(0,|X|-1)\}\cup \{(i+1,j+1) : (i,j)\in \M'\}$
    so that  $\dyck^w_{\M}(X) = \dyck^w(X[0] X[|X|-1])+\dyck^w_{\M'}(X')\le \ed^w(X[0]X[|X|-1],a\rev{a}) + \dyck^w(X')
    \le w(X[0],a)+w(X[|X|-1],\rev{a})+\ed^w(X',Y') \le \ed^w_{\A}(X,Y)=\dyck^w(X)$.
    Next, suppose that $\A$ aligns $Y[0]=a$ with $X[0]$ but inserts $Y[|Y|-1]=\rev{a}$.
    In this case, $\A$ aligns $Y[1\dd |Y|-1)=Y'$ with $X':=X[1\dd |X|)$.
    The inductive assumption yields a matching $\M'$ such that $\dyck^w_{\M'}(X') \le \dyck^w(X')$.
    In this case, we set $\M := \{(i+1,j+1) : (i,j)\in \M'\}$
    so that  $\dyck^w_{\M}(X) = \dyck^w(X[0])+\dyck^w_{\M'}(X')\le \ed^w(X[0],a\rev{a}) + \dyck^w(X')
    \le w(X[0],a)+w(\eps,\rev{a})+\ed^w(X',Y') \le \ed^w_{\A}(X,Y)=\dyck^w(X)$.
    The case when $\A$ inserts $Y[0]=a$ and aligns $Y[|Y|-1]=\rev{a}$ with $X[|X|-1]$ is symmetric.
    The case when $\A$ inserts both $Y[0]=a$ and $Y[|Y|-1]=\rev{a}$ is impossible by optimality of $Y$.
    Finally, we note that, since $\A$ deletes neither $X[0]$ nor $X[|X|-1]$, 
    the alignment $\A$ cannot align $Y[0]$ to any character other than $X[0]$ and $Y[|Y|-1]$ to any character other than $Y[|Y|-1]$.
    Thus, the case analysis above is complete.
\end{proof}

%\tknote{The claims below do not require normalization, but they do require a skewmetric.}

\begin{claim}
\label{claim:dyck1}
For every $x\in \Sigma$ and skewmetric weight function $w$, $\dyck^w(x)=w(x,\epsilon)=w(\epsilon,\rev{x})$.
\end{claim}
\begin{proof}
We consider the following three different cases.

\vspace{2mm}
\noindent
\textbf{Case 1:} $x$ is deleted. In this case $\dyck^w(x)=w(x,\epsilon)$.

\vspace{2mm}
\noindent
\textbf{Case 2:} $\rev{x}$ is inserted after $x$ if $x\in T$ and before $x$ if $x\in \rev{T}$. In this case $\dyck^w(x)=w(\epsilon,\rev{x})=w(x,\epsilon)$. The last equality follows as $w$ is skew-symmetric. Thus an insertion can be replaced with a deletion. From now on wards we assume that only allowed edits are deletion and substitutions.

\vspace{2mm}
\noindent
\textbf{Case 3:} $x$ is substituted by some $y\in \Sigma$. Here we also need to insert $\rev{y}$.
Thus $\dyck^w(x)= w(x,y)+w(\epsilon,\rev{y})=w(x,y)+w(y,\epsilon)\ge w(x,\epsilon)$. The second equality follows as $w$ is skew-symmetric and the last inequality follows as $w$ obeys triangle inequality. Trivially $\dyck^w(x)\le w(x,\epsilon)$. Thus the claim follows.
\end{proof}

\begin{claim}
\label{claim:dyck2}
For every $x,y\in \eSigma$ and skewmetric weight function $w$, $\dyck^w(xy)=\min_{z\in T\cup\{\epsilon\}}w(x,z)+w(y,\rev{z})$.
\end{claim}

\begin{proof}
Let $z\in T\cup\{\epsilon\}$ minimizes  $w(x,z)+w(y,\rev{z})$. It is straight forward to argue that $\dyck^w(xy)\le w(x,z)+w(y,\rev{z})$ as $x,y$ can be substituted by $z,\rev{z}$ respectively. Next we argue the converse. Following Claim~\ref{claim:dyck1} we assume the only allowed edits are deletions and substitutions.

\vspace{2mm}
\noindent
\textbf{Case 1:} Both $x$ and $y$ are deleted. Here $\dyck^w(xy)\ge w(x,\epsilon)+w(y,\epsilon)$. The claim follows as $\epsilon \in T\cup \{\epsilon\}$ and $\rev{\epsilon}=\epsilon$.

\vspace{2mm}
\noindent
\textbf{Case 2:}
$x$ is substituted by $\rev{y}$. Here $\dyck^w(xy)\ge w(x,\rev{y})=w(x,\rev{y})+w(y,y)$. Thus the claim follows as $\rev{y} \in T$.

\vspace{2mm}
\noindent
\textbf{Case 3:}
$y$ is substituted by $\rev{x}$. Here $\dyck^w(xy)\ge w(y,\rev{x})=w(x,x)+w(y,\rev{x})$. Thus the claim follows as $x\in T$.

\vspace{2mm}
\noindent
\textbf{Case 3:}
$x$ is substituted by $z$ and $y$ is substituted by $\rev{z}$. Here $\dyck^w(xy)\ge w(x,z)+w(y,\rev{z})$. Thus the claim follows as $z\in T$.
\end{proof}

From now on wards we assume $w$ to be skew-symmetric.

\subsubsection{Preprocessing.}
\label{sec:preprocessdyck}
Given the input string $X\in \Sigma^n$, preprocess $X$ as follows. As long as there are two neighboring indices $i,i+1$ such that $X[i+1]=\rev{X[i]}$ and $X[i]\in T$ remove them. Let the resulting string be $X'$. We make the following claim.

\begin{claim}
\label{claim:preprocess1}
$\dyck^w(X)=\dyck^w(X')$.
\end{claim}

\begin{proof}
Let $\M$ be an optimal alignment of $X$. For contradiction assume for two consecutive indices $i,i+1$, $X[i+1]=\rev{X[i]}$, $X[i]\in T$ but $(i,i+1)\notin \M$. Next depending on the matching indices of $i,i+1$, we consider the following three cases.

\vspace{2mm}
\noindent
\textbf{Case 1:}
 Let $(j,i),(i+1,k)\in \M$ where $j\in [0\dots i)\cup \{\bot\}$ and $k\in (i+1\dots |X|)\cup \{\bot\}$. In this case we create another alignment $\M'=\M \setminus \{(j,i),(i+1,k)\}\cup \{(i,i+1), (j,k)\}$. We argue $\dyck^w(X[j]X[k])\le \dyck^w(X[j]X[i])+\dyck^w(X[i+1]X[k])$, thus proving $\dyck^w_{\M'}(X)\le \dyck^w_{\M}(X)$. Following Claim~\ref{claim:dyck2}, let $a,b\in T\cup \{\epsilon\}$ be such that $\dyck^w(X[j]X[i])=w(X[j],a)+w(X[i],\rev{a})$ and $\dyck^w(X[i+1]X[k])=w(X[i+1],b)+w(X[k],\rev{b})$. Thus,
\begin{align*}
\dyck^w(X[j]X[i])+\dyck^w(X[i+1]X[k])&= w(X[j],a)+w(X[i],\rev{a})+w(X[i+1],b)+w(X[k],\rev{b})\\
&=w(X[j],a)+w(a,\rev{X[i]})+w(X[i+1],b)+w(X[k],\rev{b})\\
&\ge w(X[j],\rev{X[i]})+w(X[i+1],b)+w(X[k],\rev{b})\\
&\ge w(X[j],b)+w(X[k],\rev{b})\\
&\ge \dyck^w(X[j],X[k])
\end{align*}

The second equality follows as $w$ is skew-symmetric; thus $w(X[i],\rev{a})=w(a,\rev{X[i]})$. The third and fourth inequality follows as $w$ follows triangle inequality and $\rev{X[i]}=X[i+1]$. The last inequality follows from Claim~\ref{claim:dyck2}.

\vspace{2mm}
\noindent
\textbf{Case 2:}
 Let $(k,i),(j,i+1)\in \M$ where $k,j\in [0\dots i)\cup \{\bot\}$ and $j<k$. In this case we create another alignment $\M'=\M \setminus \{(k,i),(j,i+1)\}\cup \{(i,i+1), (j,k)\}$. We argue $\dyck^w(X[j]X[k])\le \dyck^w(X[k]X[i])+\dyck^w(X[j]X[i+1])$, thus proving $\dyck^w_{\M'}(X)\le \dyck^w_{\M}(X)$. Following Claim~\ref{claim:dyck2}, let $a,b\in T\cup \{\epsilon\}$ be such that $\dyck^w(X[k]X[i])=w(X[k],a)+w(X[i],\rev{a})$ and $\dyck^w(X[j]X[i+1])=w(X[j],b)+w(X[i+1],\rev{b})$. Thus,
\begin{align*}
\dyck^w(X[k]X[i])+\dyck^w(X[j]X[i+1])&= w(X[k],a)+w(X[i],\rev{a})+w(X[j],b)+w(X[i+1],\rev{b})\\
&=w(X[k],a)+w(a,\rev{X[i]})+w(X[j],b)+w(X[i+1],\rev{b})\\
&\ge w(X[k],\rev{X[i]})+w(X[i+1],\rev{b})+w(X[j],b)\\
&\ge w(X[k],\rev{b})+w(X[j],b)\\
&\ge \dyck^w(X[j],X[k])
\end{align*}

\vspace{2mm}
\noindent
\textbf{Case 3:}
 Let $(i,k),(i+1,j)\in \M$ where $k,j\in (i+1\dots |X|)\cup \{\bot\}$ and $j<k$. In this case we create another alignment $\M'=\M \setminus \{(i,k),(i+1,j)\}\cup \{(i,i+1), (j,k)\}$. We argue $\dyck^w(X[j]X[k])\le \dyck^w(X[i]X[k])+\dyck^w(X[i+1]X[j])$, thus proving $\dyck^w_{\M'}(X)\le \dyck^w_{\M}(X)$. Following Claim~\ref{claim:dyck2}, let $a,b\in T\cup \{\epsilon\}$ be such that $\dyck^w(X[i]X[k])=w(X[i],a)+w(X[k],\rev{a})$ and $\dyck^w(X[i+1]X[j])=w(X[i+1],b)+w(X[j],\rev{b})$. Thus,
\begin{align*}
\dyck^w(X[i]X[k])+\dyck^w(X[i+1]X[j])&= w(X[i],a)+w(X[k],\rev{a})+w(X[i+1],b)+w(X[j],\rev{b})\\
&=w(X[i],a)+w(X[k],\rev{a})+w(X[j],\rev{b})+w(\rev{b},\rev{X[i+1]})\\
&\ge w(X[i],a)+w(X[k],\rev{a})+w(X[j],\rev{X[i+1]})\\
&\ge w(X[j],a)+w(X[k],\rev{a})\\
&\ge \dyck^w(X[j],X[k])
\end{align*}

%Here at least one of the index is matched. W.l.o.g let $(i,j)\in \M$ where $X[j]=\rev{X[i]}$. Hence either $(i+1,\bot)\in \M$. In this case we create another alignment $\M'=\M \setminus \{(i,j),(i+1,\bot)\}\cup \{(i,i+1), (j,\bot)\}$. As $X[i+1]=X[j]$, $\dyck^w(X[i+1])= \dyck^w(X[j])=$; thus $\dyck^w_\M(X)=\dyck^w_{\M'}(X)$. Otherwise $(i+1,j')\in \M$ for $j'<j$. We create another alignment $\M'=\M \setminus \{(i,j),(i+1,j')\}\cup \{(i,i+1), (j',j)\}$. As $X[i+1]=X[j]$, Using Fact~\ref{fact:tridyck}, $\dyck^w(X[i+1]X[j'])=\dyck^w(X[j]X[j'])= \dyck^w(X[j']X[j])=$; thus $\dyck^w_\M(X)=\dyck^w_{\M'}(X)$.
\end{proof}

The preprocessing can be done in time $O(n)$. Also, we can assume that
in the preprocessed string no two neighbouring symbols can be aligned. Following this and Claim 35 from~\cite{}, we can make the following claim. 

\begin{claim}
\label{claim:preprocess2}
Let $X\in \Sigma^n$. There exists an algorithm that preprocesses $X$ in $O(n)$
time, and either declares $\dyck^w(X)>k$, or outputs a string $X'$ of length at most $n$ such that $\dyck^w(X)=\dyck^w(X')$ and $X'$ has at most $2k$ valleys.
\end{claim}

Thus from now on wards we assume $X$ to be preprocessed and has at most $2k$ valleys.

\subsection{Periodicity Reduction}

\begin{definition}
For $k\in \Zz$ a fragments
$X[a\dd b)$ and $X[c\dd d)$ of a string $X$ are \emph{$k$-synchronized} if
 $X[a\dd b)\in T^*$, $X[c\dd d)\in \rev{T}^*$,  
$b-a=d-c$, $b\le c$, and $H(b)+H(c)-2\min_{m\in [b\dd c]} H(m) \le 2k$.
\end{definition}
Note that $X[a\dd b)$ and $X[c\dd d)$ are $0$-synchronized if and only if $(a,b,c,d)$ is a trapezoid.

%\tknote{The notion of equivalence simplifies the exposition for strings and forests, so it should also be useful here.}
\begin{definition}
For $k\in \Zz$ and a skewmetric weight function $w$, strings $P,P'\in T^*$ are called \emph{$\dyck^w_{\le k}$-equivalent} if
\[\dyck^w_{\le k}(X) = \dyck^w_{\le k} (X[0\dd a)\cdot P' \cdot  X[b\dd c) \cdot \rev{P'} \cdot X[d\dd |X|))\]
holds for every string $X$ with $k$-synchronized fragments $X[a\dd b)=P$ and $X[c\dd d)=\rev{P}$.
\end{definition}

\begin{fact}[Fact 36,~\cite{BO16}]\label{fct:heights}
  Let $\M$ be an alignment such that $\dyck^w_\M(X)\le k$. If $X[a\dd b)\simeq_{\M} \rev{X[c\dd d)}$, then the fragments $X[a\dd b)$ and $X(c\dd d]$ are $k$-synchronized.
\end{fact}

\begin{fact}\label{fct:sync}
Consider a string $X$ and an alignment $\M\in \mm(X)$ such that $\dyck_\M(X)\le k$ for some $k\in \Zz$.
Moreover, let $X[a\dd b)$ and $X[c\dd d)$ be $k$-synchronized fragments of length $\ell > 6k$.
Then, there exist $k$-synchronized fragments $X[a'\dd b')$ and $X[c'\dd d')$ of length $\ell'\ge \frac{\ell-6k}{k+1}$, 
such that $X[a'\dd b')\simeq_{\M} \rev{X[c'\dd d')}$ and $a\le a' \le b'\le b \le c \le c' \le d' \le d$.
Furthermore, we then have $|(a+d)-(a'+d')|\le 4k$.
\end{fact}
\begin{proof}
Since $\M$ is non-crossing, it is disjoint with $[a\dd b)\times [d\dd |X|)$ or $[0\dd a)\times [c\dd d)$.
By symmetry (up to the reverse complement), let us assume that $\M$ is disjoint with $[a\dd b)\times [d\dd |X|)$.
Consider  $x\in [a\dd b-4k)$ such that $X[x]\simeq_{\M} \rev{X[y]}$.
The assumption implies that $y < d$. 
Moreover, $b-x = H(b) - H(x) > 4k$, so $H(x) < H(b)-4k$.
At the same time, $|H(y+1)-H(x)|\le 2k$, so $H(y+1) < H(b)-2k$.
Since $X[a\dd b)$ and $X[c\dd d)$ are $k$-synchronized, this means that $y+1\notin [b\dd c]$, i.e., $y\in [c\dd d)$.
Consider the fragment $X[a\dd b-4k)$ and the minimal subfragment of $X[c\dd d)$ containing positions that $\M$ matches perfectly to positions $X[x]$
with $x\in [a\dd b-4k)$.
These two fragments contain at most $2k$ positions that are deleted or matched imperfectly.
The remaining positions constitute a common subsequence of $X[a\dd b-4k)$ and $X[c\dd d)$;
this subsequence can be interrupted at most $k$ times, so there is a contiguous subsequence $X[a'\dd b')\simeq_{\M} X[c'\dd d')$
of length at least $\frac{\ell-6k}{k+1}$.
Due to $|H(a)-H(d)|\le 2k$ and $|H(a')-H(d')|\le 2k$,
we have $4k \ge |H(a)-H(d)-H(a')+H(d')| = |a-a'+d-d'|$.
\end{proof}

\begin{lemma}\label{lem:periodic_reduction_dyck}
    Let $k\in \Zp$, let $Q\in T^*$ be a string, and let $e,e'\in \mathbb{Z}_{\ge 8k}$.
    Then $Q^e$ and $Q^{e'}$ are $\dyck^w_{\le k}$-equivalent for every skewmetric weight function $w$.
\end{lemma}
%\tknote{From the fact above, $e,e'\ge 7k$ is enough.}
\begin{proof}
  We assume without loss of generality that $Q$ is primitive. (If $Q=R^m$ for $m\in \mathbb{Z}_{\ge 2}$, then $Q^e=R^{me}$ and $Q^{e'}=R^{me'}$
  can be interpreted as powers of $R$ rather than powers of $Q$.)
  Let $q=|Q|$. Consider a string $X$ and positions $p_T$, $p_{\rev{T}}$ such that $Q^e=X[p_T\dd p_T+e\cdot q)$ and $\rev{Q^e}=X(p_{\rev{T}}-{e \cdot q}\dd p_{\rev{T}}]$
  are $k$-synchronized fragments.
  Denote $X[0 \dd p_T) \cdot Q^{e'} \cdot X[p_T+e\cdot q \dd p_{\rev{T}}-e\cdot q]\cdot \rev{Q^{e'}}\cdot X(p_{\rev{T}} \dd |X|)$.
  Moreover, let $\M\in \mm(X)$ be an alignment such that $\dyck^w(X,Y)= \dyck^w_\M(X,Y) \le k$. 
    
      \begin{claim}\label{clm:periodic_reduction_dyck}
        There exist $i_T,i_{\rev{T}}\in [0\dd 7k]$ such that \[X[p_T+i_T\cdot q\dd p_T+(i_T+1)\cdot q)\simeq_{\M} \rev{X(p_{\rev{T}}-(i_{\rev{T}}+1)\cdot q\dd p_{\rev{T}}-i_{\rev{T}}\cdot q]}.\]
        \end{claim}
    \begin{proof}
Consider the $8k$ occurrences of $Q$ starting at positions $p_T+i\cdot q$ for $i\in [0\dd 7k]$ (let this fragment be $P$) and $8k$ occurrences of $\rev{Q}$ ending at positions $p_{\rev{T}}-i\cdot q$ for $i\in [0\dd 7k]$ (let this fragment be $\rev{P}$). Note $P, \rev{P}$ are also are $k$-synchronized fragments. Thus following Fact~\ref{fct:sync}, there exists at least one occurrence of $Q$ in $P$, starting at index $\ell$ such that $\M$ matches it exactly with a fragment in $\rev{P}$. We can thus define $i_T\in [0\dd 7k]$ so that $\M$ matches $X[p_T+i_T\cdot  q\dd p_T+(i_T+1)\cdot q)$ exactly to some fragment $X(s_{\rev{T}}-q\dd s_{\rev{T}}]\in \rev{P}$. By definition of $\rev{P}$, we have $s_{\rev{T}}\ge p_{\rev{T}} - 7kq$.
Furthermore, since $Q$ is primitive (i.e., distinct from all its non-trivial cyclic rotations), we conclude that $s_{\rev{T}} = p_{\rev{T}}-i_{\rev{T}} \cdot q$ for some $i_{\rev{T}}\in [0\dd 7k]$. 

  \iffalse
    Consider the $4k$ occurrences of $Q$ starting at positions $p_T+i\cdot q$ for $i\in [0\dd k]$. There exists at least one occurrence starting at index $\ell$ such that $\M$ matches it exactly and moreover there is an index $\ell'\le\ell$ where $\M(\ell')\le p_{\rev{T}}$. \tknote{Why? \cref{fct:sync}?}
    We can thus define $i_T\in [0\dd 4k]$ so that $\M$ matches $X[p_T+i_T\cdot  q\dd p_T+(i_T+1)\cdot q)$ exactly to some fragment $X(s_{\rev{T}}-q\dd s_{\rev{T}}]$ and index $\ell'\le p_T+i_T\cdot q$ satisfy $\M(\ell')\le p_{\rev{T}}$. The non-crossing property of $\M$ implies that $s_{\rev{T}} \le \M(\ell') \le p_{\rev{T}}$. Moreover, since $\dyck_\M^w(X) \leq k$, $(p_T+i_T\cdot q,s_{\rev{T}})\in \M$ and $Q^e, \rev{Q^e}$ are $k$-synchronized, we have $s_{\rev{T}}\ge p_{\rev{T}}- i_T\cdot q - 4k \ge p_{\rev{T}} - 8kq$. Furthermore, since $Q$ is primitive (i.e., distinct from all its non-trivial cyclic rotations), we conclude that $s_{\rev{T}} = p_{\rev{T}}-i_{\rev{T}} \cdot q$ for some $i_{\rev{T}}\in [0\dd 8k]$. 
\fi
    \end{proof}

 Now, if $Q^e=X[p_T\dd p_T+e\cdot q)$ is replaced with $Q^{e'}$ and $\rev{Q^e}=X(p_{\rev{T}}-e\cdot q\dd p_{\rev{T}}]$ is replaced with $\rev{Q^{e'}}$ for $e'\ge e-1$,
    we can interpret this as replacing $Q=X[p_T+i_T\cdot q\dd p_T+(i_T+1)\cdot q)$ with $Q^{1+e'-e}$ and $\rev{Q}=X(p_{\rev{T}}-(i_{\rev{T}}+1)\cdot q\dd p_{\rev{T}}-i_{\rev{T}}\cdot q]$ 
    with $\rev{Q^{1+e'-e}}$. By \cref{clm:periodic_reduction_dyck}, $\M$ can be trivially adapted without 
    modifying its cost, and hence $\dyck^w(X')\le \dyck^w_{\M}(X)=\dyck^w(X)$.
    If $e'< e-1$, we repeat the above argument to decrement the exponent $e$ one step at a time, still concluding that $\dyck^w(X')\le \dyck^w(X)$.
    In either case, the converse inequality follows by symmetry between $(X,e)$ and $(X',e')$.
    \end{proof}   
    
    We say that a string $P\in T^*$ avoids $k$-periodicity if it does not contain any substring $Q^{8k+1}$ with $|Q|\in [1\dd 4k]$.

\begin{lemma}\label{lem:aperiodic_reduction_dyck}
  Let $k\in \Zp$ and let $P,P'\in T^*$ be strings of lengths at least $156k^3$ such that $P[0\dd 78k^3)=P'[0\dd 78k^3)$ and $P[|P|-78k^3\dd |P|)=P'[|P'|-78k^3\dd |P'|)$ avoid $k$-periodicity. Then, $P$ and $P'$ are $\dyck_{\le k}^w$-equivalent for every skewmetric weight function~$w$.
\end{lemma}
\begin{proof}
  Consider a string $X$ and positions $p_T$, $p_{\rev{T}}$ such that $P=X[p_T\dd p_T+|P|)$, $\rev{P}=X(p_{\rev{T}}-|P|\dd p_{\rev{T}}]$
  are $k$-synchronized fragments.
  Denote $X' = X[0 \dd p_T) \cdot P' \cdot X[p_T + |P| \dd p_{\rev{T}}-|P|]\cdot \rev{P'}\cdot  X(p_{\rev{T}} \dd |X|]$.
  Moreover, let $\M\in \mm(X)$ be an alignment such that $\dyck^w(X)= \dyck^w_\M(X) \le k$. 

    \begin{claim}\label{clm:aperiodic_reduction_dyck}
      There exist $d,e\in [0\dd 78k^3]$ such that $(p_T+d,p_{\rev{T}}-d)\in\M$ and $(p_T+|P|-e,p_{\rev{T}}-|P|+e)\in \M$.
    \end{claim}
  \begin{proof}
      By \cref{fct:sync}, $X[p_T\dd p_T+78k^3)$ contains a fragment of length at least $\frac{78k^3-6k}{k+1}\ge 36k^2$
      that $\M$ matches perfectly to a fragment of $X(p_{\rev{T}}-78k^3\dd p_{\rev{T}}]$
      Thus, let $R:=X[r_T\dd r_T+|R|)$ be a fragment of length at least $36k^2$ contained in $X[p_T\dd p_T+|P|)$
      that $\M$ matches perfectly to $X[r_{\rev{T}}-|R|\dd r_{\rev{T}})=\rev{R}$.
      Moreover, let $r'_{\rev{T}} := p_T + p_{\rev{T}} - r_T$.
      If $r_{\rev{T}} = r'_{\rev{T}}$, then the claim is satisfied for $d = r_T-p_T = p_{\rev{T}}-r_{\rev{T}}$.
      Otherwise, both $X[r_{\rev{T}}-|R|\dd r_{\rev{T}})$ and $X[r'_{\rev{T}}-|R|\dd r'_{\rev{T}})$ are occurrences of $\rev{R}$ in $X$. Moreover, $0<|r_{\rev{T}}-r'_{\rev{T}}|\le |(p_{\rev{T}}-r'_{\rev{T}})-(p_{\rev{T}}-r_{\rev{T}})| \le |(r_T-p_T)-(p_{\rev{T}}-r_{\rev{T}})| + |(r_T-p_T)-(p_{\rev{T}}-r'_{\rev{T}})| \le 2\dyck^{w}_{\M}(X)+2k\le 4k$.
      %\tknote{It is not true that $|(r_T-a)-(d-r_{\rev{T}})| \le \ed^{w}_{\M}(X)$. This should be $2\cdot \ed^{w}_{\M}(X)$, and it follows from \cref{fct:heights}.}
      Hence, $\per(\rev{R})\le |r_{\rev{T}}-r'_{\rev{T}}|\le 4k$ and $\exp(\rev{R})\ge \frac{|\rev{R}|}{4k}\ge 9k$.
      Since $X[r'_{\rev{T}}-|\rev{R}|\dd r'_{\rev{T}})$ is contained in $X(p_{\rev{T}}-|P|\dd p_{\rev{T}}]=\rev{P[0\dd 78k^3)}$, this contradicts the assumption that $\rev{P[0\dd 78k^3)}$ and thus $P[0\dd 78k^3)$ avoids $k$-periodicity.

      The second part of the claim is proved analogously.
  \end{proof}
     As $X[p_T+d\dd p_T+|P|-e)\in T$ and $X[p_T+d\dd p_T+|P|-e)=\rev{X(p_{\rev{T}}-|P|+e\dd p_{\rev{T}}-d]}$, the optimality of $\M$ guarantees 
     that $X[p_T+d\dd p_T+|P|-e)\simeq_\M \rev{X(p_{\rev{T}}-|P|+e\dd p_{\rev{T}}-d]}$.
    Hence, if $P=X[p_T\dd p_T+|P|)=\rev{X(p_{\rev{T}}-|P|\dd p_{\rev{T}}]}$ is replaced with $P'$,
    we can interpret this as $P[d\dd |P|-e)=X[p_T+d\dd p_T+|P|-e)=\rev{X(p_{\rev{T}}-|P|+e\dd p_{\rev{T}}-d]}$ with $P'[d\dd |P'|-e)$. 
    Since $X[p_T+d\dd p_T+|P|-e)\simeq_\M \rev{X(p_{\rev{T}}-|P|+e\dd p_{\rev{T}}-d]}$, the alignment $\M$ can be trivially adapted without modifying its cost,
    and therefore $\dyck^w(X')\le \dyck^w_{\M}(X)=\dyck^w(X)$.
    The converse inequality follows by symmetry between $(X,P)$ and $(X',P')$.
\end{proof}

 \begin{corollary}\label{cor:dyck_reduction}
Let $k\in \Zp$. For every string $P\in T^*$, there exists a string of length at most $156k^3$ that is $\dyck_{\le k}^w$-equivalent to $P$ for every skewmetric weight function~$w$.
 \end{corollary}
\begin{proof}
 We proceed by induction on $|P|$ with the trivial base case of $|P|\le 156k^3$.
 If $|P|\ge 156k^3$ and $P$ avoids $k$-periodicity, then \cref{lem:aperiodic_reduction_dyck} implies that $P$ is equivalent to  a string $P':= P[0\dd 78k^3)\cdot P[|P|-78k^3\dd |P|)$ of length $156k^3$.
 Thus, suppose that $P$ contains a fragment $P[i\dd j)=Q^{8k+1}$ and $|Q|\in [1\dd 4k]$.
 By \cref{lem:periodic_reduction_dyck}, $Q^{8k+1}$ is equivalent to $Q^{8k}$,
 and thus $P$ is equivalent to a string $P':=P[0\dd i)\cdot P[i+|Q|\dd |P|)$.
 By the inductive assumption, $P'$ is equivalent to some string $P''$ of length at most $156k^3$, and, by transitivity of the considered equivalence,
 $P$ is also equivalent to~$P''$.
 \end{proof}

\subsection{Algorithm}

\SetKwFunction{DyckReduction}{DyckReduction}
\begin{lemma}\label{lem:dyck_reduction}
    There exists a linear-time algorithm that, given a string $P$ and an integer $k\in \Zp$,
    constructs a string $P'$ of length at most $156k^3$ that is $\dyck_{\le k}^w$-equivalent to $P$ for every skewmetric weight function~$w$. Moreover $P'$ avoids $k$-periodicity.
\end{lemma}
\begin{algorithm}
  \caption{Construct a string of length at most $156k^3$ that is $\dyck_{\le k}^w$-equivalent to $P$.}\label{alg:dyck_per_reduc}
  \Fn{$\DyckReduction(P,k)$}{
      $P' \gets \PeriodicityReduction(P, 8k, \{Q\in T^+ : |Q|\le 4k \text{ and $Q$ is primitive}\})$\;
      \lIf{$|P'| \geq 156k^3$}{\KwRet{$P'[0\dd 78k^3) \cdot P'[|P'| - 78k^3 \dd |P'|)$}}
      \lElse{\KwRet{$P'$}}
  }
\end{algorithm}
\begin{proof}
  We apply \cref{alg:str_per_reduc} with $e=8k$ and $\Qf$ consisting of all primitive strings in $T^*$ of length in $[1\dd 4k]$. 
  If the resulting string $P'$ satisfies $|P'|< 156k^3$, we return $P'$.  By \cref{lem:periodic_reduction_dyck,lem:perred}, the string $P'$ is $\dyck_{\le k}^w$-equivalent to $P$ and avoids $k$-periodicity. 
  Thus, if  $|P'| \le 156k^3$, then the algorithm is correct.
  Otherwise, we return $P'[0\dd 78k^3) \cdot P'[|P'| - 78k^3 \dd |P'|)$. $P'[0\dd 78k^3)$  and $P'[|P'| - 78k^3 \dd |P'|)$ both avoid $k$-periodicity,
  so $P'[0\dd 78k^3) \cdot P'[|P'| - 78k^3 \dd |P'|)$ is  $\dyck_{\le k}^w$-equivalent to $P$ by \cref{lem:aperiodic_reduction_dyck}.
  Due to \cref{lem:perred}, the running time is linear (testing whether a primitive fragment belongs to $\Qf$ simplifies to checking if its length does not exceed $4k$.)
\end{proof}

\begin{theorem}
\label{thm:dyckkernel}
    There exists a $O(n+k^5)$-time algorithm that, given a preprocessed string $X$ and an integer $k\in \Zp$,
    constructs strings $X'$ of lengths at most $630k^4$ such that $\dyck^w_{\le k}(X)=\dyck^w_{\le k}(X')$ holds for every skewmetric weight function $w$.
\end{theorem}

\SetKwFunction{DyckKernel}{DyckKernel}
\SetKwFunction{Order}{Order}
\SetKwFunction{DyckReduction}{DyckReduction}
\begin{algorithm}
    \caption{Construct strings $X'$ of length at most $630k^4$ such that $\dyck^w_{\le k}(X)=\dyck^w_{\le k}(X')$}\label{alg:dyck_kernel}
    \Fn{$\DyckKernel(X,k)$}{
    \lIf{$|X|\le 630k^4$}{\KwRet{$(X)$}}\label{ln:trivial}
    \lIf{$\dyck(X) > k$}{\KwRet{$(a^{k+1})$} for some $a\in \Sigma$}\label{ln:large}
    Let $\M\in \mm(X)$ be a dyck language alignment satisfying $\dyck_{\M}(X)\le k$\;
    $X',P,Q\gets \varepsilon$\;
    \For{$i \gets 0$ \KwSty{to} $n-1$}
    {
        \If{$\M(i)=\bot$ \KwSty{or} $\dyck(X[i]X[\M(i)])\cdot \Order(i,\M(i))=1$ \KwSty{or} $\dyck(X[\M(i)]X[i])\cdot \Order(\M(i),i)=1$}
        {
            %$P \gets \StringReduction(P,k)$\;\label{ln:reduce1}
            %$X' \gets X'\cdot P$\;
           % $Q \gets \rev{\StringReduction(\rev{Q},k)}$\;\label{ln:reduce2}
            %$X' \gets X'\cdot Q$\;
           % $P \gets \varepsilon$\;
           % $Q \gets \varepsilon$\;
            $X' \gets X'\cdot X[i]$
        }
         \ElseIf{$X[i]\in T$ \KwSty{and} $\M(i+1)=\M(i)-1$ \KwSty{and} $\dyck(X[i+1]X[\M(i)-1])=0$}
        {
            $P \gets P\cdot X[i]$
        }
        
        %\ElseIf{$X[i]\in T$ \KwSty{and} $P=\varepsilon$}
        %{
        %    $P \gets P\cdot X[i]$
        %}
        %\ElseIf{$X[i]\in T$ \KwSty{and} $\M(i)=\M(i-1)-1$}{
         %   $P \gets P\cdot X[i]$
        %}
        \ElseIf{$X[i]\in T$}{
             $P \gets P\cdot X[i]$\;
             $P \gets \DyckReduction(P,k)$\;\label{ln:reduce1}
            $X' \gets X'\cdot P$\;
            $P \gets \varepsilon$\;
            }
        \ElseIf{$X[i]\in \rev{T}$ \KwSty{and} $\M(i+1)=\M(i)-1$ \KwSty{and} $\dyck(X[\M(i)-1]X[i+1])=0$}
        {
           $Q \gets Q\cdot X[i]$
        }
        \Else{
            $Q \gets Q\cdot X[i]$\;
             $Q \gets \rev{\DyckReduction(\rev{Q},k)}$\;\label{ln:reduce2}
            $X' \gets X'\cdot Q$\;
            $Q \gets \varepsilon$\;
        }
       % \Else
        %{
         %   $P \gets \StringReduction(P,k)$\;\label{ln:reduce1}
          %  $X' \gets X'\cdot P$\;
           % $Q \gets \rev{\StringReduction(\rev{Q},k)}$\;\label{ln:reduce2}
            %$X' \gets X'\cdot Q$\;
            %$P \gets \varepsilon$\;\label{ln:peps1}
            %$Q \gets \varepsilon$\;\label{ln:peps2}
           
        %}
    }
    % $S_0, \ldots, S_k, \{c_1, c_2, \ldots, c_i\} \gets$ Partition($X, Y, \A$)\;
    % $T_0, \ldots, T_k, \{d_1, d_2, \ldots, d_j\} \gets$ Partition($Y, X, \A^{-1})$\;
    % \ForEach{$\ell \in [0 \ldots k]$}{
    %     $S_\ell' \gets $
    %     StringPeriodReduction($S_\ell$)\;
    %     $T_\ell' \gets$ StringPeriodReduction($T_\ell$)\;
    % }
    % $X' \gets S_0' \cdot c_1 \cdot S_1' \cdot c_2 \cdot \ldots \cdot c_i \cdot S_i'$\;
    % $Y' \gets T_0' \cdot d_1 \cdot T_1' \cdot d_2 \ldots \cdot d_i \cdot T_i'$\;
    \KwRet{$(X')$} 
    }
    \end{algorithm}

\begin{proof}
Our procedure is implemented as \cref{alg:dyck_kernel}. First, if $X$ is already of length at most $630k^4$, then we return $X$ as it is.
If $\dyck(X)>k$, we return strings $a^{k+1}$, where $a\in \Sigma$ is an arbitrary character.
If $\dyck(X)\le k$, 
we construct a Dyck language alignment $\M\in \mm(X)$ of (unweighted) cost at most $k$. 
We then build the output string $X'$ using $\M$ as follows:
scan $X$ from left to right, if the the scanned character $X[i]$ is edited by $\M$ we append it to $X'$ (here $\Order(i,j)=1$ if $j>i$ otherwise it is $0$). Otherwise $X[i]$ is matched under $\M$. If $X[i]\in T$, we proceed with scanning the following characters to identify the maximal fragment $P=X[i\dd j)\in T^*$ such that there is a fragment $X(i'\dd j']$ where $\dyck(X[i\dd j)X(i'\dd j'])=0$ and $\M$ matches $X[i\dd j)$ with $ X(i'\dd j']$. Next we apply the reduction of \cref{lem:dyck_reduction} on $P$ and append the reduced string to $X'$.
Otherwise $X[i]\in \rev{T}$. Here also we proceed with scanning the following characters to identify the maximal fragment $Q=X[i\dd j)\in \rev{T}^*$ such that there is a fragment $X(i'\dd j']$ where $\dyck(X(i'\dd j']X[i\dd j))=0$ and $\M$ matches $X[i\dd j)$ with $ X(i'\dd j']$.
Next we consider the $\rev{Q}$ (note $\rev{Q}\in T^*$), apply the reduction of \cref{lem:dyck_reduction} on $\rev{Q}$ and append the reverse complement of the reduced string to $X'$.

%Next using $\M$ (following Lemma~\ref{fct:dyck_decomp}) we construct set $S=\{X[a_i\dd b_i) : i\in [1\dd s]\}$ of at most $2k+1$ disjoint fragments and of total length at least $\frac12|X|-k$ such that, for every $i\in [1\dd s]$, there is a fragment $X[c_i\dd d_i)$ such that $X[a_i\dd b_i)\simeq_{\M} \rev{X[c_i\dd d_i)}$. Note by construction this list is sorted in increasing order by the starting indices of the fragments. %Next we create a list $R=\{P_1, \dd, P_s\}$ where $P_i$ is generated by applying the reduction of \cref{lem:dyck_reduction} on $X[a_i\dd b_i)$. Next we then build the output strings $X'$ using $\M, S$ as follows:scan $X$ from left to right, if the the scanned character is edited by $\M$ we append it to $X'$. If the scanned character is a starting index of some fragment $X[a_i\dd b_i)$ from $S$, we apply the reduction of \cref{lem:dyck_reduction} to generate a string $P$ and append it to $X'$. Lastly if the scanned character is a starting index of some fragment $X[c_i\dd d_i)$, we apply the reduction of \cref{lem:dyck_reduction} on $\rev{X[c_i\dd d_i)}$ to generate a string $P$ and append $\rev{P}$to $X'$.

Let us now prove that the resulting string $X'$ satisfies $\dyck^w_{\le k}(X)=\dyck^w_{\le k}(X')$.
This is trivial when the algorithm returns $X$ in \cref{ln:trivial}.
If $\dyck(X)>k$, then $\dyck_{\le k}(X)=\dyck_{\le k}(a^{k+1})=\infty$
and thus also $\dyck^w_{\le k}(X)= \dyck^w_{\le k}(a^{k+1})=\infty $ because the weighted Dyck edit distance with a normalized weight function is at least as large as the unweighted Dyck edit distance. 
In the remaining case we assume $\dyck(X)\le k$. 
Let $S=\{P_1,\dd, P_\ell\}$ be the set of fragments from $T^*$ that are ever generated and processed using $\DyckReduction()$ routine at Line~\ref{ln:reduce1}. Similarly let $\rev{S}=\{Q_1,\dd, Q_{\ell'}\}$ be the set of fragments from $\rev{T}^*$ that are ever generated and whose reveres complements are processed using $\DyckReduction()$ routine at Line~\ref{ln:reduce2}. By construction it is trivial to follow that (i) the fragments are disjoint; (ii) for all $i\in [0\dd n)$, if $X[i]\in T$ and is not edited by $\M$, then there exist some $P_j$ such that $X[i]\in P_j$ and if $X[i]\in \rev{T}$ and is not edited by $\M$, then there exist some $Q_j$ such that $X[i]\in Q_j$; (iii) The fragments are maximal in a sense that if $P_i=X[a\dd b)\in S$, then either $\M(b)\neq \M(b-1)-1$ or $X[b]\neq \rev{X[\M(b)]}$ and same holds for the fragments in $\rev{S}$.
Next we prove for each $P_i=X[a\dd b)\in S$, $\exists Q_j=X(c\dd d]\in \rev{S}$, such that $|P_i|=|Q_j|$ and for all $k\in[0\dd |P_i|)$, $(a+k,d-k)\in \M$ and $X[a+k]=\rev{X[d-k]}$. 
For this we first claim that $\M(b-1)$ is a 
starting index of some $Q_j\in \rev{S}$. As otherwise $\M(b)=\M(b-1)-1$ and $X[b]\neq \rev{X[\M(b)]}$; this contradicts the maximality of $P_i$. Further by
construction for all $k\in [a\dd b)$, $X[\M(k)]\in Q_j$. Finally we argue $\M(a)$ is a ending index of $Q_j$. As otherwise $\M(a-1)=\M(a)+1$ and this contradicts the fact that $a$ is the starting index of some segment from $S$. 
Similarly we can show for each $Q_j\in \rev{S}$ there is a corresponding match $P_i\in S$ and this provides an one to one correspondence between a pair of fragments from $S$ and $\rev{S}$. Thus for a fragment $P_i\in S$ let $\M(P_i)$ represents the corresponding matched fragments from $\rev{S}$ and we can represent $S\cup \rev{S}=\cup_{i\in [\ell]}(P_i,\M(P_i))$. 
Following Fact~\ref{fct:heights}, $P_i,\M(P_i)$ are $k$-synchronized.
Next in the algorithm for each pair $(P_i,\M(P_i))$ we add strings $\DyckReduction(P_i)$ representing $P_i$ and $\rev{\DyckReduction(P_i)}$ (note $\rev{\M(P_i)}=P_i$) representing $\M(P_i)$ to $X'$. 
%Also as we scan $X$ left to right, we add the reduced strings exactly in the order in which their corresponding fragments appear in $X$.
Following the fact that every character that is not contained in a fragment from $S\cup \rev{S}$ is edited by $\M$ and thus copied to $X'$ directly, by applying Lemma~\ref{lem:dyck_reduction} repeatedly for every pair $(P_i,\M(P_i))$, we claim $\dyck^w_{\le k}(X)=\dyck^w_{\le k}(X')$.

Next, we show that the returned string is of length at most $630k^4$.
This is clear when the algorithm terminates at Line~\ref{ln:trivial} or~\ref{ln:large}.
Otherwise, we create a string $X'$, to which we directly copy the characters that are edited by $\M$. However there are at most $2k$ characters that $\M$ deletes or substitutes. Next we identify maximal fragments $P=X[i\dd j)\in T^*$ such that there is another fragment $X(i'\dd j']\in \rev{T}^*$ that is matched with $P$ by $\M$. The maximality of $P$ and the preprocessing of $X$ ensure that at least one of $X[j]$ and $X[i']$ is edited by $\M$, We call these characters the boundary characters for $P$. Notice for any two distinct fragments $P,P'\in T^*$, the the boundary characters are different and by construction $P,P'$ are disjoint. As there are at most $2k$ characters that $\M$ edits, we conclude there can be at most $2k$ fragments over $T^*$, that our algorithm can construct. For each such fragment following the reduction of Lemma~\ref{lem:dyck_reduction}, we add a substring of length $156k^3$ to $X'$. 
Thus the total length of all the substrings is $312k^4$. Similarly we can argue for the fragments $Q\in \rev{T}^*$. Thus we can bound the total length of $X'$ by $2\cdot 312k^4 +2k<630k^4 $.

%applying \cref{fct:dyck_decomp} we create set $S$ containing $2k+1$ fragments. For each such fragment $X[a_i\dd b_i)\in S$, and its corresponding match $X[c_i\dd d_i)\in \rev{S}$, following the reduction of Lemma~\ref{lem:dyck_reduction}, we add two substrings each of length $244k^3$ to $X'$. Moreover there are at most $2k$ characters that $\M$ deletes or substitutes which are directly copied to $X'$. Let $E$ be the set of these characters. Following the fact that $X=\cup_{X[a_i\dd b_i)\in S} (X[a_i \dd b_i)\cup X[c_i \dd d_i))\cup E$, we conclude that $|X'|\le 2k + (2k+1)\cdot 488k^3 \le k^5$.

It remains to analyze the complexity of our procedure.
We use the algorithm~\cite{OtherSubmission} to check whether $\dyck(X)\le k$ and, if so, construct the alignment $\M$.
This costs $\Oh(n+k^5)$ time.
Next we perform a single left to right scan of $X$. 
Throughout, all the conditions in the \emph{if/else} statements can be checked in $O(1)$ time. Moreover any character is passed to the $\DyckReduction()$ routine at most twice. Thus following Lemma~\ref{lem:dyck_reduction},
given $X$ and $\M$, $X'$ can be constructed in linear time. 
\end{proof}

\begin{proof}[Proof of Theorem~\ref{thm:weighted_dyck}]
We first preprocess $X$ in linear time following the steps described in Section~\ref{sec:preprocessdyck} to build strings $X'$ such that $\dyck_{\le k}^w(X')=\dyck_{\le k}^w(X)$. 
Next we apply \cref{thm:dyckkernel} on $X'$, to build strings $X''$ of length $\Oh(k^4)$ such that 
$\dyck_{\le k}^w(X'')=\dyck_{\le k}^w(X')$. This takes time $O(n+k^5)$. Lastly if $X= a^{k+1}$ (this can be checked in time $O(k)$) output the distance is $>k$. Otherwise we compute $\dyck_{\le k}^w(X'')$ using the dynamic program algorithm from \cite{M95} in time $O(k^{12})$. 
Thus the total running time is $O(n+k^{12})$.
\end{proof}

\appendix
\section{Deferred Proofs from Section~\ref{sec:ed}}\label{app:proofs}
In the following, we give the missing proofs of facts from Section~\ref{sec:ed}.

\edtri*

\begin{proof}

Consider arbitrary strings $X, Y, Z \in \Sigma^*$ as well as alignments $\A = (x_t,y_t)_{t=0}^m \in \aa(X,Y)$ and $\B = (\hy_t, \hz_t)_{t=0}^{\hm} \in \aa(Y, Z)$. We construct a \emph{product alignment} $\A\otimes \B \in \aa(X,Z)$ such that $\ed^w_{\A \otimes \B}(X, Z) \leq \ed^w_\A(X, Y) + \ed^w_\B(Y, Z)$.
Let us denote $\A' = (x_t,y_t)_{t=0}^{m-1}$ and $\B' = (\hy_t,\hz_t)_{t=0}^{\hm-1}$,
as well as $X'=X[0\dd |X|-1)$ if $X \ne \eps$, $Y'=Y[0\dd |Y|-1)$ if $Y\ne \eps$, and $Z'=Z[0\dd |Z|-1)$ if $Z\ne \eps$.
We proceed by induction on $m+\hm$ and consider several cases based on how $\A$ and $\B$ handle the trailing characters of $X$, $Y$, and $Z$.
\begin{enumerate}
    \item $m=\hm = 0$. In this case, $X = Y = Z = \varepsilon$, and we define $\A \otimes \B := (0,0)$. Trivially, $\ed^w_{\A\otimes \B}(X,Z)= 0 = \ed^w_\A(X, Y) + \ed^w_\B(Y, Z)$.
    \item\label{cs:del} $(x_{m-1},y_{m-1})=(|X|-1,|Y|)$, that is, $\A$ deletes $X[|X|-1]$. In this case, $\A'\in \aa(X', Y)$,
    and we define $\A \otimes \B := (\A' \otimes \B)\odot (|X|,|Z|)$, where $\odot$ denotes concatenation,
    so that $\A \otimes \B$ deletes $X[|X|-1]$.
    By the induction hypothesis, $\ed^w_{\A\otimes \B}(X,Z) = \ed^w_{\A'\otimes \B}(X',Z)+w(X[|X|-1],\eps) \le \ed^w_{\A'}(X',Y)+\ed^w_\B(Y,Z)+w(X[|X|-1],\eps) = \ed^w_\A(X,Y)+\ed^w_\B(Y,Z)$.
    \item\label{cs:ins} $(\hy_{\hm-1},\hz_{\hm-1})=(|Y|-1,|Z|)$, that is, $\B$ inserts $Z[|Z|-1]$. In this case, $\B'\in \aa(Y,Z')$,
    and we define $\A \otimes \B := (\A \otimes \B') \odot (|X|,|Z|)$ so that $\A \otimes \B$ inserts $Z[|Z|-1]$.
    By the induction hypothesis, $\ed^w_{\A\otimes \B}(X,Z) = \ed^w_{\A\otimes \B'}(X,Z')+w(\eps,Z[|Z|-1]) \le \ed^w_{\A}(X,Y)+\ed^w_{\B'}(Y,Z')+w(\eps, Z[|Z|-1])   = \ed^w_\A(X,Y)+\ed^w_\B(Y,Z)$.
    \item\label{cs:id} $(x_{m-1},y_{m-1})=(|X|,|Y|-1)$ and $(\hy_{\hm-1},\hz_{\hm-1})=(|Y|-1,|Z|)$, that is, $\A$ inserts $Y[|Y|-1]$ and $\B$ deletes $Y[|Y|-1]$.
    In this case, $\A'\in \aa(X,Y')$ and $\B'\in \aa(Y',Z)$,
    and we define $\A \otimes \B := \A' \otimes \B'$.
    By the induction hypothesis, $\ed^w_{\A \otimes \B}(X,Z) = \ed^w_{\A'\otimes \B'}(X,Z) \le \ed^w_{\A'}(X,Y')+\ed^w_{\B'}(Y',Z)\le \ed^w_\A(X,Y)+\ed^w_\B(Y,Z)$.
    \item\label{cs:ia} $(x_{m-1},y_{m-1})=(|X|,|Y|-1)$ and $(\hy_{\hm-1},\hz_{\hm-1})=(|Y|-1,|Z|-1)$, that is, $\A$ inserts $Y[|Y|-1]$
    and $\B$ aligns $Y[|Y|-1]$ with $Z[|Z|-1]$.
    In this case, $\A' \in \aa(X,Y')$ and $\B'\in \aa(Y',Z')$,
    and we define $\A \otimes \B := (\A'\otimes \B')\odot (|X|,|Z|)$ so that $\A \otimes \B$ inserts $Z[|Z|-1]$.
    By the induction hypothesis,  $\ed^w_{\A \otimes \B}(X,Z) = \ed^w_{\A'\otimes \B'}(X,Z') + w(\eps, Z[|Z|-1]) \le \ed^w_{\A'}(X,Y')+\ed^w_{\B'}(Y',Z')
    + w(\eps, Y[|Y|-1])+w(Y[|Y|-1], Z[|Z|-1])= \ed^w_\A(X,Y)+\ed^w_\B(Y,Z)$.
    \item\label{cs:ad} $(x_{m-1},y_{m-1})=(|X|-1,|Y|-1)$ and $(\hy_{\hm-1},\hz_{\hm-1})=(|Y|-1,|Z|)$, that is, $\A$ aligns $X[|X|-1]$ with $Y[|Y|-1]$
    and $\B$ deletes $Y[|Y|-1]$.
    In this case, $\A' \in \aa(X',Y')$ and $\B'\in \aa(Y',Z)$, and we define $\A \otimes \B := (\A'\otimes \B')\odot (|X|,|Z|)$ so that $\A \otimes \B$ deletes $X[|X|-1]$.
    By the induction hypothesis,  $\ed^w_{\A \otimes \B}(X,Z) = \ed^w_{\A'\otimes \B'}(X',Z) + w(X[|X|-1], \eps) \le \ed^w_{\A'}(X',Y')+\ed^w_{\B'}(Y',Z)
    + w(X[|X|-1], Y[|Y|-1])+w(Y[|Y|-1], \eps)= \ed^w_\A(X,Y)+\ed^w_\B(Y,Z)$.
    \item\label{cs:aa} $(x_{m-1},y_{m-1})=(|X|-1,|Y|-1)$ and $(\hy_{\hm-1},\hz_{\hm-1})=(|Y|-1,|Z|-1)$, that is, $\A$ aligns $X[|X|-1]$ with $Y[|Y|-1]$
    and $\B$ aligns $Y[|Y|-1]$ with $Z[|Z|-1]$.
    In this case, $\A' \in \aa(X',Y')$ and $\B'\in \aa(Y',Z')$,
    and we define $\A \otimes \B := (\A'\otimes \B')\odot (|X|,|Z|)$ so that $\A \otimes \B$ aligns $X[|X|-1]$ with $Z[|Z|-1]$.
    By the induction hypothesis,  $\ed^w_{\A \otimes \B}(X,Z) = \ed^w_{\A'\otimes \B'}(X',Z') + w(X[{|X|-1}], Z[|Z|-1]) \le \ed^w_{\A'}(X',Y')+\ed^w_{\B'}(Y',Z)  + w(X[|X|-1], Y[|Y|-1])+w(Y[{|Y|-1}], \allowbreak Z[|Z|-1])= \ed^w_\A(X,Y)+\ed^w_\B(Y,Z)$.
\end{enumerate}
It is easy to check that the above cases cover all the possibilities.
In particular, Case~\ref{cs:del} covers the case of $\hm = 0 < m$
whereas Case~\ref{cs:ins} covers the case of $m = 0 < \hm$.
We also remark that Cases~\ref{cs:del} and~\ref{cs:ins} are sometimes both applicable; by convention, we then follow Case~\ref{cs:del}.
Finally, we note that Cases~\ref{cs:ia}--\ref{cs:aa} rely on the assumption that $w$ satisfies the triangle inequality.
This completes the proof of the first part of the fact.

To show that $\ed^w(X, Y)$ can be equivalently defined as the minimum cost of a sequence of edits transforming $X$ into $Y$, we first consider each of the operations in a minimum alignment $\A$ of $X$ and $Y$ individually to build a sequence of edits $S$ from $\A$. We iterate through all pairs $(x_t, y_t)$ of $\A$ from right to left starting with $t = m - 1$, stopping after $t = 0$ has been handled, and building $S$ according to the definition of alignments:

\begin{enumerate}
    \item If $(x_t, y_t) = (x_{t+1} - 1, y_{t+1} - 1)$ and $X[x_t] = Y[y_t]$, we add nothing to $S$.
    \item If $(x_t, y_t) = (x_{t+1} - 1, y_{t+1} - 1)$ and $X[x_t] \ne Y[y_t]$, we add a substitution of $X[x_t]$ with $Y[y_t]$ to $S$.
    \item If $(x_t, y_t) = (x_{t+1} - 1, y_{t+1})$, we add a deletion of $X[x_t]$ to $S$.
    \item If $(x_t, y_t) = (x_{t+1}, y_{t+1} - 1)$, we add an insertion of $Y[y_t]$ at position $x_t$ in $X$ to $S$.
\end{enumerate}

In all cases, we decrement $t$ by 1. Clearly the resulting sequence of edits has the same cost as $\A$, and by the definition of alignments, $S$ transforms $X$ into $Y$. We now consider a minimum sequence of edits $S$ that transforms $X$ to $Y$ and build an alignment $\A \in \aa(X, Y)$ from $S$ such that $\ed_\A^w(X, Y) \leq cost(S)$ (we let $cost(S)$ denote the total cost of edits by $S$). We use notation $\A', X', Y'$ as before, and proceed by induction to construct $\A$:
\begin{enumerate}
    \item If $X[|X| - 1]$ is deleted by $S$ and a character is inserted at the end of $X$, then $\A' \in \aa(X',  Y')$ and we set $\A =  \A' \odot (|X|, |Y|)$. We note that the inserted character $c$ may be substituted to $Y[|Y| - 1]$. We let $S'$ be the sequence $S$ without the insertion, deletion, and if possible substitution on the last character of $X$. By the induction hypothesis and triangle inequality, $\ed_{\A}(X, Y) = \ed_{\A'}(X', Y') + w(X[|X| - 1], Y[|Y| - 1]) \leq \ed_{\A'}(X', Y') + w(X[|X| - 1], \eps) + w(\eps, Y[|Y| - 1]) \leq \ed_{\A'}(X', Y') + w(X[|X| - 1], \eps) + w(\eps, c), + w(c, Y[|Y| - 1]) \leq cost(S') + w(X[|X| - 1], \eps) + w(\eps, c), + w(c, Y[|Y| - 1]) = cost(S)$.
    \item If $X[|X| - 1]$ is deleted by $S$ and no character is inserted at the end of $X$, then $\A' \in \aa(X', Y)$ and we set $\A = \A' \odot (|X|, |Y|)$. We let $S'$ be the sequence $S$ without the deletion of $X[|X| - 1]$. By the induction hypothesis, $\ed_{\A}(X, Y) = \ed_{\A'}(X', Y) + w(X[|X| - 1], \eps) \leq cost(S') + w(X[|X| - 1], \eps) = cost(S)$.
    \item If a character is inserted at the end of $X$, then $\A' \in \aa(X, Y')$ and we set $\A = \A' odot (|X|, |Y|)$. We let $S'$ be the sequence $S$ without this insertion. By the induction hypothesis, $\ed_{\A}(X, Y) = \ed_{\A'}(X, Y') + w(\eps, Y[|Y| - 1], \eps) \leq cost(S') + w(\eps, Y[|Y| - 1]) = cost(S)$.
    \item If $X[|X| - 1]$ is substituted by $S$, then $\A' \in \aa(X', Y')$ and we set $\A = \A' \odot (|X|, |Y|)$. We let $S'$ be the sequence of $S$ without any substitutions of $X[|X| - 1]$ and let $C$ be an ordered list of characters substituted by $S$ at $X[|X| - 1]$. Then, by the induction hypothesis, $\ed_{\A}(X, Y) = \ed_{\A'}(X', Y') + w(X[|X|-1], Y[|Y| - 1]) \leq \ed_{\A'}(X', Y') + \sum_{i = 0}^{|C| - 1}w(c_i, c_{i+1}) \leq cost(S') + \sum_{i = 0}^{|C| - 1}w(c_i, c_{i+1})  = cost(S)$.
\end{enumerate}

By induction, we can see that there exists an alignment $\A$ with cost at most that of $S$. 

\end{proof}

\fctquasi*

\begin{proof}
        The unique alignment in $\aa(X[0 \dd i) \cdot X[j \dd |X|), \eps)$ deletes all characters, and therefore $\ed^w(X[0\dd i) \cdot X[j \dd |X|), \eps) = \sum_{u\in [0\dd i)\cup [j\dd |X|)} w(X[u], \eps)$.
        
        Next, consider an alignment $\A \in \aa(X, X[i \dd j))$ that deletes $X[0\dd i)$, matches $X[i\dd j)$ perfectly, and deletes $X[j\dd |X|)$.  Hence, $\ed^w(X, X[i \dd j)) \leq \ed^w_{\A}(X, X[i\dd j)) = \sum_{u\in [0\dd i)\cup [j\dd |X|)} w(X[u], \eps) = \ed^w(X[0\dd i) \cdot X[j \dd |X|), \eps)$.
        
        Now, let $Y = X[i\dd j)$ and consider an arbitrary alignment $\B \in \aa(X, Y)$. For each $u\in [0\dd i)\cup [j\dd |X|)$, we recursively define a sequence $(c_{u,t})_{t=0}^{m_u}$ so that $c_{u,0}=u$,
        the alignment  $\B$ aligns $X[c_{u,t}]$ to $Y[c_{u,t+1}-i]=X[c_{u,t+1}]$ for $t\in [0\dd m_u)$, and $\B$ deletes $X[c_{u,m_u}]$. 
        By construction, the sequences $(c_{u,t})_{t=0}^{m_u}$ are finite and each position in $X$ belongs to at most one such sequence.
        Moreover, since $w$ is quasimetric, $w(X[c_{u,t}], \eps) \le w(X[c_{u,t}], X[c_{u,t+1}])+w(X[c_{u,t+1}], \eps)$ holds for every $t\in [0\dd m_u)$,
        and thus 
         $w(X[c_{u,0}], \eps) \leq w(X[c_{u,m_u}], \eps) + \sum_{t\in [0\dd m_u)} w(X[c_{u,t}], X[c_{u,t+1}])$.
        %
        %
        %
        % Now, we consider an arbitrary alignment $\B \in \aa(X, X[i, j))$ and show that $\ed^w_\A(X, X[0, i) \cdot X[j, |X|), \eps) \leq \ed^w_\B(X, X[i, j))$. Let $D_\B$ be the set of indices $t \in [0, |X|]$ such that $X[t]$ is deleted by $\B$ and let $S_\B$ be the set of indices $t \in [0, |X|]$ such that $X[t]$ is aligned to some character in $X[i \dd j)$ by $\B$.
        % \begin{align*}
        %    \sum_{t=0}^{i-1} w(X[t], \eps) + \sum_{t = j}^{|X| - 1} w(X[t], \eps) = \sum_{t \in (D_\B \cap[0, i))} w(X[t], \eps) + \sum_{t \in (D_\B \cap [}
        % \end{align*}
        %
        Since $\B$ deletes $X[c_{u,m_u}]$ and $X[c_{u,t}]\sim_{\B} Y[c_{u,t+1}-i]=X[c_{u,t+1}]$ holds for $t\in [0\dd m_u)$, this yields
        \begin{align*}
            \ed^w_\B(X, Y) &\ge \sum_{u\in [0\dd i)\dd [j\cup |X|)} \left(  w(X[c_{u,m_u}],\eps) + \sum_{t\in [0\dd m_u)} w(X[c_{u,t}], X[c_{u,t+1}])\right)\\
            &\ge \sum_{u\in [0\dd i)\dd [j\cup |X|)} w(X[u], \eps) \\
            & =  \ed^w(X[0 \dd i) \cdot X[j \dd |X|), \eps).
        \end{align*}
        Since $\B$ was chosen arbitrarily, we conclude that $\ed^w(X, Y)\ge  \ed^w(X[0 \dd i) \cdot X[j \dd |X|), \eps)$.
\end{proof}

\bibliographystyle{alphaurl}
\bibliography{ted}

\end{document}